\documentclass[11pt]{article}
\usepackage{mathrsfs}
\usepackage{palatino}
\usepackage[pdftex]{graphicx}
\usepackage[skip=0pt]{caption}
\usepackage[a4paper, margin=1in]{geometry}
\usepackage{standalone}
\usepackage{amsfonts}
\usepackage{amsmath}
\usepackage{amssymb}
\usepackage{tikz-cd}
\usepackage{hyperref}
\usepackage{blindtext}
\usepackage{enumitem}
\usepackage{braket}
\usepackage{mathtools}
\usepackage{environ}
\usepackage{fancyhdr}
\usepackage{stmaryrd}
\usepackage{tikz-cd}
\usepackage{xcolor}
\usepackage{pgfplots}
\usepackage{multirow}
\usepackage{tikz}
\usepackage{colortbl}
\usepackage{comment}
\usepackage{verbatim}
\usepackage{caption}
\usepackage{booktabs}
\usepackage{float}
\usepackage{subcaption}
\usepackage{bbm}
\usepackage[blocks]{authblk}
\restylefloat{table}

\allowdisplaybreaks
\newtheorem{theorem}{Theorem}[section]
\newtheorem{proposition}[theorem]{Proposition}
\newtheorem{lemma}[theorem]{Lemma}

\newtheorem{definition}[theorem]{Definition}
\newtheorem{example}[theorem]{Example}
\newtheorem{remark}[theorem]{Remark}
\newtheorem{condition}[theorem]{Condition}
\let\olddefinition\definition
\let\oldexample\example
\let\oldremark\remark
\let\oldcondition\condition

\renewcommand{\definition}{\olddefinition\normalfont}
\renewcommand{\example}{\oldexample\normalfont}
\renewcommand{\remark}{\oldremark\normalfont}
\renewcommand{\condition}{\oldcondition\normalfont}
\newenvironment{proof}{\noindent{\bf Proof:}}{$\hfill \Box$ \vspace{10pt}}  

\newcommand\numberthis{\addtocounter{equation}{1}\tag{\theequation}}

\DeclareMathOperator*\esssup{ess\,sup}
\DeclareMathOperator*\essinf{ess\,inf}
\DeclareMathOperator*\argmax{arg\,max}
\DeclareMathOperator\TV{TV}

\DeclareMathOperator\Unif{Unif}
\DeclareMathOperator\clip{clip}

\usepackage[backend=bibtex, style=authoryear, hyperref=true, maxbibnames=99, sorting=nyt]{biblatex}
\title{Content Platform GenAI Regulation via Compensation}
\author{Wee Chaimanowong\footnote{The Chinese University of Hong Kong (chaimanowongw@gmail.com).} }
\date{\today}

\begin{document}

\maketitle

\begin{abstract}
    The use of Generative AI (GenAI) for creative content generation has gained popularity in recent years. GenAI allows creators to generate contents that are increasingly becoming indistinguishable to the human--generated counter--part at a much lower cost. While GenAI reshapes the competitive landscape of the contents market, the original creators were typically not compensated for their works that were used in the GenAI training. On the other hands, the wide--spread adoption of GenAI threatens to replace the human--generated shares of contents on content platforms, contaminating training data source for future GenAI models. In this paper, we argue that an unregulated usage of GenAI can also be harmful to the platform by causing a contents distribution distortion which can lower the consumers' engagement and the platform's profit. We show that a simple economically--driven creator compensation scheme, can incentivize more creation of high--value human--generated contents, without the need for an AI--detector. This reduces the data pollution for future GenAI training, while improves the consumer engagement and the platform's profit
\end{abstract}

\section{Introduction}

The rise of generative AI (GenAI) in recent years has enabled automation of many tasks previously thought to only be possible with human guidances. One area with wide--spread adoption of GenAI is the creative content generations, such as the AI generated videos, images, musics, articles. These GenAI contents are posted, shared, and received consumers responses, alongside other usual content created manually by human creators. Moreover, identifying the content made by GenAI from the one made by humans is becoming more challenging over time as the technology evolves. 

Some of the leading concerns of the GenAI usages trend are as follows. GenAI training requires a large amount of data, and these data were not always obtained with the consent or permission from the original creators, e.g. via web scraping, raising concerns of intellectual property and privacy violation (\cite{schurz2025fair, solove2025great, novelli2024generative}). In particular, the creators are often not compensated for the use of their works in GenAI training, while facing increased economic competition from the GenAI content (\cite{ducru2024ai, pasquale2024consent, kim2026does, oderinwale2025economics}). It was shown in \cite{kim2026does} that the negative impact, in term of the loss of viewers' attention, is particularly significant among the creators who do not adopt GenAI, although some creators managed to leverage GenAI in production cost reduction. The spread of GenAI contents also raises concerns about \emph{data pollution}. This negatively impacts the development of future GenAI models as training on a dataset contaminated with other GenAI output, instead of purely human content, can degrade the model quality in a phenomenon known as \emph{model collapse} (\cite{shumailov2023curse, bertrand2023stability}). Lastly, the adoption of GenAI creates a distributional distortion of contents, or \emph{content homogenization} (\cite{anderson2024homogenization, doshi2024generative}), or more colloquially known as `AI--slop' (\cite{ansari2025ai}), which could worsen the consumers' experience and engagement (\cite{liu2025generative}). Roughly speaking, no GenAI models are perfectly trained, and some types of content will be disproportionally generated more than others, reducing the overall contents diversity, over--saturating a certain segment of the market while under--serving some other niches.

To enable GenAI development while addressing the aforementioned challenges, human generated content should continue to be promoted, and the creators should be compensated for any economically valuable contribution to future training dataset. We review some of the current approaches in the following.

\emph{Ex--ante revenue redistribution:} One solution is the pay--to--train compensation model, such as Shutterstock Contributor Fund (\cite{ducru2024ai}), where the platform (Shutterstock) earns revenue by licensing the dataset to AI developers and redistributing the revenue back to the creators. How the redistribution should be done is a topic of discussion in itself. A basic form is to redistribute a percentage of the revenue to the contributing creators proportionally to the contributed volume, which is the case for the Shutterstock Contributor Fund (\cite{ducru2024ai}). Rather than relying on series of negotiations between AI developers and data providers, which leads to fragmented coverage of data while potentially leaving out many independent creators behind, some government bodies, such as the European Union, have also been considering a statutory licensing option (\cite{peukert2025eu}). The economics and mechanism design aspects of data purchases have been covered in various works (\cite{ai2025genai,oderinwale2025economics,liang2018survey,agarwal2019marketplace, zhang2025fairshare, liu2025online, keinan2025strategic}). A key assumption underlining the approach for most of these works is that quality can be determined or controlled in the procurement of any specific volume of data. In particular, the platform may be required to have a certain control over GenAI usages, such as a robust AI--detection tool, otherwise the platform may unintentionally encourage creators to submit a large volume of GenAI contents to gain a larger share of compensation, worsening the data pollution. 





\emph{Ex--post revenue redistribution:} Another intuitive compensation model is to share revenue from any GenAI output with the creators on the platform proportional to each creator's influence on the output (\cite{ducru2024ai, bonnet2025fair}). While this approach provides an incentive for creators to sustain the high quality content production, the challenge remains to quantifies the influence of each creator's work on the GenAI output. An economic approach is \emph{Data Shapley} (\cite{ghorbani2019data, wang2024economic, lee2025faithful}); however, its practical consideration is limited to revenue sharing among few major data providers, as it is computationally expensive and often involves re--training the GenAI model on several data subsets. Some progress has been made to improve the computability of Data Shapley, for example, by first training an explainer model (\cite{sun2025fast}), or by embedding the Shapley value calculation as part of the training process (\cite{wang2024data}). A related concept in machine learning literature is that of \emph{influence function}, a useful attribution method which quantifies, via perturbative up--weighting, the impact of a training data point on the trained parameters (\cite{koh2017understanding, zhu2025revisiting}). Influence function is expensive to compute for GenAI models due to the need to invert a large Hessian, hence much of the machine learning literature on the topic focuses on technical challenges of computation or approximation efficiency (\cite{kwon2023datainf, park2023trak, park2025concept}). Computation aside, both Data Shapley and influence functions rely on a choice of the utility function or the loss function; therefore, more consensus remains on what is considered a fair metric and how to best translate them into monetary economic influence or compensation (\cite{ducru2024ai}). Furthermore, the validity of the influence function as a proxy for the data point importance in large GenAI models remains controversial (\cite{li2024influence}), a factor to consider before a policy--level implementation.

\emph{AI--detection:} Many of the challenges we have discussed, from preventing model collapse via verified data curation to combating misinformation, could be addressed with a robust AI--detection tool. As GenAI models becomes more advanced, the detection becomes more challenging. The existing detection methods often lack generalizablity, performs poorly on samples from different GenAI models, and are not robust against post--processing such as scaling or rotating of the content (\cite{zhang2025unmasking, xiao2025high}). A different approach mandated by some major AI companies is to include certain \emph{content credentials}, such as watermark or meta data to content generated by their GenAI models (\cite{hilbert2026watermarking, balan2025framework}). Although this approach has gained traction from policy makers recently, there remain key technical difficulties, for example, some content credentials can easily be removed by editing or taking a screenshot (\cite{nemecek2025watermarking}).

In this work, we take an alternative route to examine the problem of how to maintain a content platform where human creators and GenAI co--exists. We consider a model where the platform is economically incentivized to regulate the level of GenAI contents to avoid the decline in consumers' engagement due to content distributional distortion from the consumers' preference. We focus on the setting where the access to GenAI is democratized, the platform has no regulation power on GenAI access and has no AI--detection capability (or that the GenAI contents are indistinguishable from the human generated counter--part). We will show that with no intervention, GenAI will be widely adopted by most creators, some market segments will be flooded with GenAI contents while some niche contents will be under--supplied, resulting in a lower overall consumers' engagement. We will consider a simple economically--driven compensation scheme which does not rely on AI--detection nor any computationally intensive mechanism. We will show that by redistributing some of the revenue back to a certain portion of the creators via the compensation scheme, the platform can encourage more manual creation of high--value contents, improve its profit, and reduce the data pollution for future GenAI training in the process.

Our model is closely related to \cite{yao2024human}, however their focus is on establishing the equilibrium competition outcome of content creators when GenAI is available, whereas our focus will be on the platform's regulation of the said competition. The impact of GenAI has been an active research topic in business and economics in recent years, we review some notable related works as follows. We note that although most works in this area focus on vertical differentiation, our work focuses exclusively on horizontal differentiation between GenAI and human--generated contents. \cite{luo2025platform} examined the platform's fine--tuning of GenAI for consumers, with and without compensation to creators, and found that although the creators' welfare is higher with compensation, the platform's profit is maximized without compensation. They argued that compensation policy choice is a trade--off between profitability and equitability. \cite{gao2025pandora} analyzed how GenAI, which helps creators improve content quality and enables repositioning freedom, effects the market outcome of the creative platform, subject to the platform's GenAI usage penalty. They found that although the contents' quality can increase under GenAI, the creativity and welfare may decrease. \cite{zou2026express} explored how GenAI reduces the quality--gap among creators and showed that the impact on welfare is not always positive. \cite{yan2024creative} argued that GenAI enables creators to shift their focus from execution to ideation, thus, the quality--gap depends on whether the high--skill creators are ideation--savvy or execution--savvy. \cite{chen2025designing} showed that the most accurate AI--detector does not necessarily lead to the best outcome for the platform. \cite{wu2026self} studied the mandatory self--disclosure policy for GenAI usage as a supplement to the AI--detector, and found that such policy is useful when GenAI is not fully mature. \cite{yuan2025generative} argued that the platform should charge fees for GenAI usages, otherwise the mass adoption of GenAI by low--quality creators would drive high--quality creators to leave the platform. \cite{zhang2025democratizing} examined changes in the nature of competition between GenAI and human--generated content subject to variation in GenAI's creativity and quality.



\section{Model}

In the following, we let $\mathcal{X} \subset \mathbb{R}^d$ be compact. We write $\mathcal{D}(\mathcal{X})$ for the space of upper--semicontinuous densities with a positive lower--bound over $\mathcal{X}$, and $\mathcal{P}(\mathcal{X})\subset \mathcal{D}(\mathcal{X})$ for the subset of probability densities (density normalized to one). We introduce our base model, where the game is played for a single period. Later, the generalization to a multi--period model with short--lived consumers and creators will be straightforward.

\subsection{Basic Setups}

Let $\mathcal{X} \subset \mathbb{R}^d$ be the space of content preferences of consumers and creators on the platform in the given period.  A unit mass continuum of both creators and consumers are distributed over $\mathcal{X}$ according to the probability distribution with density $p \in \mathcal{P}(\mathcal{X})$. We will refer to a consumer or creator with preference $x \in \mathcal{X}$ as `consumer $x$' or `creator $x$'. 

For convenience, the space $\mathcal{X}$ will also serve as the phase space of contents. Each creator $x \in \mathcal{X}$ can create at most one new content using one of the following creation action $a(x) \in \{\text{H}, \text{AI}, \text{O}\}$:

\begin{itemize}
    \item \emph{Creates a content manually ($\text{H}$ -- `Human--generated content'):} The creator $x$ creates a new original content of the same type as her preference at a production cost $c > 0$. In other words, the content manually created by the creator $x$ is $x \in \mathcal{X}$.
    \item \emph{Use GenAI ($\text{AI}$):} We assume that a GenAI model is available and it is represented by a distribution with a density $g \in \mathcal{P}(\mathcal{X})$ over $\mathcal{X}$. The creator $x$ generates a new content using GenAI by randomly draws $y \in \mathcal{X}$ from the distribution $g$, independent of $x$, at a production cost of zero. 
    \item \emph{Do nothing ($\text{O}$ -- `Outside option'):} The creator exits without creating any new content.
\end{itemize}

The creation strategy of the creator $x\in \mathcal{X}$ can be formalized as:
\begin{equation*}
    \beta(x) := (\beta_{\text{H}}(x), \beta_{\text{AI}}(x), \beta_{\text{O}}(x)) \in \Delta\{\text{H}, \text{AI}, \text{O}\}, \quad \beta_a(x) := \mathbb{P}[a(x) = a], \ \forall a \in \{\text{H}, \text{AI}, \text{O}\}.
\end{equation*}
Each creator decides on the strategy once enters the platform without observing other creators or consumers. Given the creators' creation strategy $\beta$, and that creators are distributed across the space of preference with density $p$, we obtain the (potentially not normalized) density $q$ of the contents on the platform in the given period (see (\ref{generaldecomposedformwitho})). 

\subsection{Revenue Production}

We now discuss the revenue of the platform and creators, to motivate the definition we present a heuristic derivation. Instead of a continuum, let us first assume the platform consists $N$ consumers i.i.d. drawn from the distribution $p$, and $N$ contents i.i.d. drawn from the distribution of contents $q$. Each consumer $x$ will only engage with any content $y$ in some close preference neighborhoods $x+\Delta x \subset \mathcal{X}$. This could be by design, shaped by the platform's recommender system, or by consumers' self--selection through search settings. Then there are $N p(x)\Delta x$ consumers and $Nq(x)\Delta x$ contents in the neighborhood $x + \Delta x$. A unit of revenue is generated when a consumer $x$ positively engages with, or `likes', the content $y$. For each unit revenue generated, the platform receives a $\gamma \in [0,1]$ share while the creator $y$ receives a $1-\gamma$ share. 

A consumer $x$ may not like all the contents within the scope of her preference: $x+\Delta x$, instead, the `like' decision is idiosyncratic, depends on various factors such as the search pattern on the particular platform visit. We model the number of likes using a matching function (\cite{pissarides2000equilibrium}): $\mathcal{M} : \mathbb{Z}_{\geq 0}\times \mathbb{Z}_{\geq 0}\rightarrow \mathbb{R}_{\geq 0}$. In general, $\mathcal{M}(n,m)$ represents the (expected) number of matches between $n$ people in the first group and $m$ people in the second group. A typical choice of matching function is the Cobb--Douglas form $\mathcal{M}(n,m) := An^{\alpha}m^{1-\alpha}$ for some $\alpha \in (0,1)$. For our purposes, we will choose $A = 1$ so that the revenue generated from the consumers $Np(x)\Delta x$ and the contents $Nq(x)\Delta x$ is 
\begin{equation*}
    (Np(x)\Delta x)^\alpha (Nq(x)\Delta x)^{1-\alpha} = p(x)^\alpha q(x)^{1-\alpha}\cdot N\Delta x.
\end{equation*}
Note that for a fixed number of consumers, the number of likes is concave in the number of contents, reflecting the fact that each consumer has a limited time and attention, and the search becomes less efficient when the amount of contents becomes over--saturated. Meanwhile, for a fixed number of contents, the number of likes is concave in the number of consumers, reflecting the fact that the platform with a limited contents lacks sufficient variety and appeal to a large number of consumers.

Since there are $Nq(x)\Delta x$ contents from $Nq(x)\Delta x$ creators collectively generating a revenue of $p(x)^\alpha q(x)^{1-\alpha}\cdot N\Delta x$ in the neighborhood $x+\Delta x$, if we assume all contents in $\Delta x$ to have an equal chance of being liked, then the competitive share of revenue for the content $y \in x+\Delta x$ is $p(x)^\alpha q(x)^{1-\alpha}\cdot N\Delta x/(Nq(x)\Delta x) = (p(x)/q(x))^\alpha$. Motivated by this heuristic derivation, returning to the continuum limit with $N\rightarrow \infty, \Delta x\rightarrow 0$, we define the creator's revenue from the content $y \in \mathcal{X}$ after the platform commission to be
\begin{equation}\label{creatorrevenue}
    V(y;q) := (1-\gamma)\left(\frac{p(y)}{q(y)}\right)^\alpha.
\end{equation}
For the platform, if $\mathcal{X}$ is partitioned into small neighborhoods, then the revenue per consumer is found by summing the revenue share from all the neighborhoods: $\frac{\gamma}{N}\sum p(x)^\alpha q(x)^{1-\alpha} \cdot N\Delta x$. Motivated by this, returning to the continuum limit, we define the platform's revenue per consumer to be $\gamma\int_{\mathcal{X}}p(x)^\alpha q(x)^{1-\alpha}dx$.

\subsection{Compensation Schemes}

In addition to the $1-\gamma$ share of revenue per `like' from the consumer for the creator, which may be monetary or simply a social gain, the platform can also redistribute its revenue back to the creators based on the pre--defined rule we refer to as the platform's compensation scheme.

\begin{definition}
    A platform \emph{compensation scheme} is a function $W : \mathcal{X}\times \mathcal{D}(\mathcal{X})\rightarrow \mathbb{R}_{\geq 0}$. In addition to the usual share of revenue, the creator behind the content $y$ receives a compensation of $W(y;q)$ from the platform at the end of the period.
\end{definition}

Here, we present the definition of a compensation scheme in full generality; we will leave the question of practicality to later. We will assume throughout that the platform has no problem delivering payments to the creator $x\in \mathcal{X}$ behind any given content $y\in \mathcal{X}$. We believe this is a reasonable assumption as any creator must have created an account with the platform before posting the content. What the platform cannot directly observe is the content creation process, whether $x$ used GenAI or created the content $y$ manually. Since we either consider a single--period game, or multi--period game with short--lived creators, given a single observable content $y$ with an unobservable production process, the assumption does not help the platform identify the true location (the true preference) of $x$.

\subsection{Summary and Equilibrium Concept}

The compensated revenue for the content $y\in \mathcal{X}$ under the compensation scheme $W$ is:
\begin{equation}
    V^W(y;q) := V(y;q) + W(y;q).
\end{equation}
Therefore, the profit of creator $x\in \mathcal{X}$ from manual creation versus using GenAI are:
\begin{equation}\label{manualandaiutility}
    U^W(x;q) := V^W(x;q) - c, \quad \text{and} \quad V^W(g;q) := \int_{\mathcal{X}}V^W(y;q) g(y)dy,
\end{equation}
respectively. Finally, the platform's revenue and profit per consumer are given by:
\begin{equation}\label{platformrevenue}
    R^W(q) := \gamma\int_{\mathcal{X}}p(x)^\alpha q(x)^{1-\alpha}dx, \qquad \Pi^W(q) = R^W(q) - \int_{\mathcal{X}}W(x;q)q(x)dx,
\end{equation}
respectively. We drop the superscript $W$ to refer to each of the aforementioned quantities under no compensation: $W = 0$, e.g. $V(g;q) := V^{W = 0}(g;q) = \int_{\mathcal{X}}V(y;q)g(y)dy$.

We summarize the game timeline in the given period as follows. The probability densities $p$ and $g$ are exogenously given, but they are not known to the consumers, the creators, or the platform. First, the platform decides the compensation scheme $W$. The continuum of consumers and creators distributed across $\mathcal{X}$ by $p$ enters the platform, and each creator $x$ decides a creation strategy $\beta(x) \in \Delta\{\text{H}, \text{AI}, \text{O}\}$ according to the common belief on the expected revenue of each option\footnote{Each creator could form the belief based on past experience or observation of peer's performance.}: $\text{H}:U^W(x;q)$, $\text{AI}:V^W(g;q)$, and $\text{O}:0$. Finally, the revenue generated from each content is publicly observed, shared between the platform and the creators, and the platform compensates each creator as specified by the compensation scheme. The equilibrium concept is as follows:
\begin{definition}\label{equilibriumdefinition}
    Consider the expected profit of a creator $x\in \mathcal{X}$ with a creation strategy $\beta = \beta(x)\in \Delta\{\text{H}, \text{AI}, \text{O}\}$ under the platform compensation scheme $W$ when the contents on the platform is distributed by $q$:
    \begin{equation*}
        \pi^W(x, \beta; q) = \beta_{\text{H}}(x)\cdot U^W(x;q)+ \beta_{\text{AI}}(x)\cdot V^W(g;q),
    \end{equation*}
    The pair $(\beta, q)$, where $\beta : \mathcal{X}\rightarrow \Delta\{\text{H}, \text{AI}, \text{O}\}$ is a creation strategy and $q\in \mathcal{D}(\mathcal{X})$, is a \emph{mean field equilibrium} under the platform's compensation scheme $W$ if
    \begin{equation}
        \beta(x) \in \argmax_{\beta\in \Delta\{\text{H}, \text{AI}, \text{O}\}} \pi^W(x,\beta;q), \qquad \forall x \in \mathcal{X}
    \end{equation}
    and 
    \begin{equation}\label{generaldecomposedformwitho}
        q(x) = \beta_{\text{H}}(x)p(x) + g(x)\int_{\mathcal{X}}\beta_{\text{AI}}(y)p(y)dy, \qquad \forall x \in \mathcal{X}.
    \end{equation}
\end{definition}

We remark that, given an arbitrary compensation $W$, the existence of an equilibrium as in Definition \ref{equilibriumdefinition} is not guaranteed (see Proposition \ref{equilibriumnonexistenceproposition}). It is clear from H\"{o}lder inequality that:
\begin{equation*}
    \Pi^W(q) \leq \gamma\int_{\mathcal{X}}p(x)^\alpha q(x)^{1-\alpha}dx \leq \gamma\left(\int_{\mathcal{X}}p(x)dx\right)^\alpha\left(\int_{\mathcal{X}}q(x)dx\right)^{1-\alpha} \leq \gamma,
\end{equation*}
therefore, the best possible outcome for the platform is if $q = p$ without any non--trivial compensation. As we will see, this is not possible when GenAI is available with $g\neq p$, thus, choosing an optimal compensation scheme becomes crucial for the platform. Lastly, we note that the Definition \ref{equilibriumdefinition} allows for a non--normalized $q$, i.e. $\int_{\mathcal{X}}q(x)dx \leq 1$. The first term in (\ref{generaldecomposedformwitho}) represents the density contribution from manual creation at each $x$, the second term represents the distortion of the content distribution by GenAI, while the density from creators who choose the outside option contributes nothing to $q$.

\subsection{Discussion and Comments}

With GenAI, a new content can be generated with zero effort by randomly drawing from $g$. This reflects the reality that GenAI content is easy to make but impersonal, as the creator has limited control over the generation. The key assumption is that the creator will not put any significant additional effort into content customization. It is possible for a creator to use GenAI as a creative assistant and customize the final product to align with their creative vision, but we assume that the effort cost from doing so would also not be zero. In our view, creators who use GenAI as a creative assistant should be considered a separate category from both manual and GenAI creators, which we will not consider in our main model. 

Although our model is similar to that of \cite{yao2024human}, we note the important conceptual distinction. In the setting of \cite{yao2024human}, different points in the set $\mathcal{X}$ correspond to different `topic' (e.g. sports, musics, education, etc.), whereas it is important that in our settings, $\mathcal{X}$ represents different preferences of the same topic. For example, our $\mathcal{X}$ could represents the space of preferences for (and the phase space of) pictures of `a dog playing with a tennis ball'. This justifies our assumption that $p$ and $g$ are unknown to all parties, and why creators cannot have more control over GenAI generation via simple prompt engineering to target a specific `topic' that appears more profitable: drawing new content $y\sim g(.|\text{prompt})$. Instead, our $\mathcal{X}$ is already the part of the platform's market after filtered by a specific prompt such as `a dog playing with a tennis ball'. This still allows substantial variations in details such as the `dog breed', `posture', `background environment', and overall style, etc. A consumer who previously searched for `Shiba Inu' would then be recommended pictures of `a Shiba Inu playing with a tennis ball' by the platform when entering the market. The creators can also further refine the prompt to be more specific to better reflect her personal preference, and we can further refine the meaning of $\mathcal{X}$, but as argued in the previous paragraph, at some point the effort costs will also not be zero.


\section{Pre--GenAI Equilibrium}\label{generativeaisection}

We consider a baseline setting when GenAI is not yet available. In this case, the creators' action set simplifies to $\{\text{H}, \text{O}\}$. The following result shows that if the revenue share for creators is sufficiently high, then all creators will participate, leading to an optimal profit. However, if the revenue share is not sufficient, then we show that a flat--rate compensation can be effective at improving the profit. This baseline result illustrates the traditional usage of compensation, to improve the social welfare by subsidizing the production costs in a certain market of contents with high creative or artistic values.

\begin{lemma}\label{pregenailemma}
    Suppose that GenAI is not available. Consider a pair $(\beta, q)$ where $\beta(x) := (\beta_{\text{H}}(x), \beta_{\text{O}}(x))$ is given by: 
    \begin{equation}\label{pregenaibeta}
        \beta_{\text{H}}(x) := \min\left\{1, \left(\frac{1-\gamma}{\max\{c - W(x;q),0\}}\right)^{1/\alpha}\right\}, \qquad \beta_{\text{O}}(x) := 1 - \beta_{\text{H}}(x).
    \end{equation}
    and $q(x) = \beta_{\text{H}}(x)p(x)$, then $(\beta, q)$ is the unique equilibrium under the platform compensation scheme $W : \mathcal{X}\times \mathcal{D}(\mathcal{X})\rightarrow \mathbb{R}_{\geq 0}$. If $1-\gamma \geq c$ then $q = p$ and the platform achieves the best possible profit $\Pi^W(q) = \gamma$ with $W = 0$. If $1-\gamma < c$ then the platform's optimal compensation scheme is given by $W^*(x;q) = \max\left\{0,c\frac{\gamma-\alpha}{1-\alpha}\right\}$ for all $(x,q)\in \mathcal{X}\times \mathcal{D}(\mathcal{X})$ and the corresponding optimal profit is $\Pi^{W^*}(q) = \min\{\gamma, \alpha\}\left(\frac{1-\min\{\gamma, \alpha\}}{c}\right)^{1/\alpha-1}$.
\end{lemma}

In the case where $1-\gamma \leq c - W(x;q)$, we remark that $q(x) = \beta_{\text{H}}(x)p(x) = \left(\frac{1-\gamma}{c - W(x;q)}\right)^{1/\alpha}$ is a functional equation for $q$ because $W(x;q)$ can depend on $q$. If the functional equation is not satisfied, then $(\beta, q)$ is not an equilibrium. In fact, not every choices of $W$ admits an equilibrium \footnote{For example, consider $W(x;q) := (c+\gamma-1)\mathbbm{1}[V(x;q)>1-\gamma]$ when $1-\gamma < c$. If $\beta_{\text{H}}(x) < 1$ then $V(x;q) > 1-\gamma$ and $U^W(x;q) > 1-\gamma+(c+\gamma-1)-c = 0$. But $(\beta_{\text{H}}=1,q=p)$ cannot be an equilibrium, since $W(x;q) = 0$ and $U^W(x;q) = 1-\gamma-c<0$.}. If $W$ admits an equilibrium $(\beta, q)$ then the platform's profit $\Pi^W(q)$ is given by the expression (\ref{pregenaiplatformprofit}) in the proof of Lemma \ref{pregenailemma}. But such an expression for $\Pi^W(q)$ is also valid for any arbitrary $W : \mathcal{X}\times \mathcal{D}(\mathcal{X})\rightarrow \mathbb{R}_{\geq 0}$, and we have shown in the proof of Lemma \ref{pregenailemma} that the global maximum of $\Pi^W(q)$ is achieved with a constant compensation: $W^*(x;q) = \max\left\{0, c\frac{\gamma-\alpha}{1-\alpha}\right\}$. Under constant compensation $W^*$ the equation $q(x) = \left(\frac{1-\gamma}{c - W^*}\right)^{1/\alpha}$ is trivial, therefore, we may ignore the issues of the existence of equilibrium when looking for the optimal compensation scheme, under which the platform's profit is as given in Lemma \ref{pregenailemma}.


The equilibrium content distribution $q$ can be thought of as the training data for the first generation of GenAI model that will be available in the next period. An example that fits the context of our work is an AI company that scrapes the contents on the platform to construct a distribution $g$. In practice, the training set will consist of a finite $N$ data points i.i.d. drawn from $q$ which can be used to train GenAI based on many available models such as VAE (\cite{kingma2013auto}), GAN (\cite{goodfellow2014generative}), or diffusion models such as DDPM (\cite{sohl2015deep}), SMLD (\cite{song2019generative}), and their generalization, the score-based SDE diffusion model (\cite{song2020score}). Each of these models represents different parametric ways that one can construct and sample the density $g$ that approximates $p$. Our consideration of $g$ will be model--free, hence independent of the fast--evolving technical details. However, it can be argued using a non--parametric lower--bound that with a finite $N$ data points, the trained $g$ will not be identical to $p$, even if $q = p$, or $q$ is proportional to $p$ as suggested in Lemma \ref{pregenailemma}. For example, if the SDE--based model is used, $p$ is $\beta$--Sobolev and $d$ is the dimension of the samples, then it was proven in \cite{zhang2024minimax} that the minimax risk of the total variation distance $\mathbb{E}\TV(p,g)$ is bounded by $N^{-2\beta/(2\beta+d)}(\log N)^C \rightarrow 0$ for some constant $C > 0$. 

\section{Compensation Schemes}\label{compensationschemesection}

For the remainder of this work, we assume that a GenAI model $g$ is available to all creators. We will study the choice of platform's compensation scheme, then characterize the equilibrium outcome and the resulting platform's profit. With GenAI, the `Outside option' is dominated by GenAI which always give a positive revenue, hence we can effectively restrict the creators' creation action set to $\{\text{H}, \text{AI}\}$. Since $\beta_{\text{H}}(x) = 1 - \beta_{\text{AI}}(x)$ and $\beta_{\text{O}}$ for any creation strategy $\beta$, to simplify the notation, we will refer to $\beta$ by $\beta_{\text{AI}}$ and write $\beta(x) = \beta_{\text{AI}}(x)$. In particular, (\ref{generaldecomposedformwitho}) simplifies to:
\begin{equation}\label{generaldecomposedform}
    q(x) = (1-\beta(x))p(x) + g(x)\cdot \int_{\mathcal{X}}\beta(y)p(y)dy,
\end{equation}
and $q$ is automatically normalized: $\int_{\mathcal{X}}q(x)dx = 1$. 

\subsection{Revenue--Threshold Compensation Schemes}

Based on our on--going discussion, we seek a compensation scheme satisfying the following properties:
\begin{itemize}
    \item \emph{Encourages economically valuable human--generated contents:} We follow--up on the idea of an ex--post revenue redistribution in the similar spirit to that of Data Shapley or influence function: compensation to the creators proportional to the contribution level and the revenue generated from the content. This is unlike the ex--ante revenue redistribution where the dataset are priced based on some past performance metrics, and could be vulnerable to spam submission of GenAI contents or low--quality human--generated contents. In our context, if the creator $y\in \mathcal{X}$ can generates more revenue than the creator $x\in\mathcal{X}$ from a manual creation under the current market condition, then $y$ should be more likely to create a manual content and receive more compensation compared to $x$ under our compensation scheme.

    \item \emph{Computationally feasible and inexpensive:} We can broadly summarize the main challenge of designing compensation allocation to creators in the era of GenAI to be the problem of quantifying the level of contribution. The mentioned attribution techniques such as Data Shapley and influence function address this problem in theory, but the main implementation hurdle remains computational efficiency. Meanwhile, a robust AI--detector would enable a direct GenAI content regulation as the platform can choose to reward a creator who generates a content manually or punish a creator who uses GenAI at--will, however, the feasibility of such a detector remains a question. Therefore, the compensation scheme we consider will not attempt to identify the creation process of any given content $x\in \mathcal{X}$.
\end{itemize}

We stress that our work is distinct from the literature on data attribution or procurement mechanism in that the majority of the literature is from the point--of--view of the GenAI provider, whereas we focus on the policy decision of the platform with no control over the GenAI model. In this sense, we are not proposing a compensation scheme to replace existing data attribution techniques. However, in a broader sense, we are addressing the common problem of data pollution due to GenAI and promoting human--generated contents with high economic value.

Let us analyze the required properties as follows. The distribution $q$ represents the current market condition on the platform. Suppose that a creator $x$ can generate $V(x;q)$ and $y$ can generate $V(y;q) \geq V(x;q)$ from creating a content manually. If $x$ is compensated $W(x;q)$, then $y$ should be compensated $W(y;q) \geq W(x;q)$, to reflect the greater economic contribution and to ensure a greater incentive for manual content production. However, the platform should only offer a compensation if it impacts a creator decision, and it should offers the minimum amount to do so. This means, if $V(x;q) - c + W(x;q) < V^W(g;q)$, then it is better for the platform to offer $W(x;q) = 0$ to $x$, and if $V(y;q)\geq V(x;q)$ then the platform should offer $W(y;q) \leq W(x;q)$ because it takes less compensation for $y$ manual creation to break--even with the GenAI option. It follows that the compensation amount is identical for all the recipient creators $x,y$: $W(x;q) = W(y;q)$. Finally, we assume the platform has no knowledge of $p$, $g$, or AI--detection capability. Consequently, the platform does not differentiate between a content $y$ created by a creator $y$ and a content $y$ created by a creator $x$ using GenAI. This is a typical moral hazard problem, the creators' creation decision cannot be observed or contracted, additionally, since the platform has no knowledge of $p$ or $g$, the compensation contract must entirely be in terms of the ex--post revenue $V(y;q)$. 

This motivates us to study the \emph{revenue--threshold} compensation scheme $W^{\underline{v}, w} : \mathcal{X}\times \mathcal{D}(\mathcal{X})\rightarrow \mathbb{R}_{\geq 0}$ where the creator $x$ receives a compensation $w \geq 0$ if her raw revenue is at least $\underline{v}\geq 0$: 
\begin{equation}\label{compensationformula}
    W^{\underline{v},w}(x;q) := w\cdot \mathbbm{1}[V(x;q) \geq \underline{v}].
\end{equation}
We will typically refer to $W^{\underline{v}, w}$ compensation scheme simply as a $(\underline{v}, w)$ compensation scheme. The threshold $\underline{v}$ can also be thought of as a shield against paying compensation to random `flops' contents, or any potential deliberate spam contents. We can also consider lowering the compensation $w>0$ for any content $x\in \mathcal{X}$ with $V(x;q)+w-c > V^{\underline{v}, w}(g;q)$ since this does not change the creator's creation decision. However, as we will later see in Lemma \ref{xhisemptylemma}, when GenAI is available, no creator strictly prefers manual creation at equilibrium under any $(\underline{v}, w)$ compensation scheme, i.e. we have $V^{\underline{v}, w}(x;q) \leq V^{\underline{v}, w}(g;q)+c$ for all $x \in \mathcal{X}$. In fact, as we will argue in \S\ref{equilibriumanalysissection}, the only critical parameter in the design of a revenue--threshold compensation scheme is the threshold $\underline{v}$, from which the compensation amount largely follows. Therefore, the platform's problem of optimal revenue--threshold compensation scheme selection can be reduced to an optimization problem in a single variable $\underline{v}$ (see Proposition \ref{platformrevennueproposition}). The following result summarizes the optimal decision rule for each creator under $(\underline{v}, w)$ given the belief $q$ of the platform's content distribution.

\begin{lemma}[Creators' Decision Under a Revenue--Threshold Compensation Scheme]\label{decisionundercompensationlemma}
    Suppose that the common belief that the platform's content distribution density is given by $q\in \mathcal{D}(\mathcal{X})$, then we characterize the creators' creation decision under the platform's compensation scheme $(\underline{v}, w)$ as follows.
    \begin{enumerate}
        \item If $\underline{v} \leq V^{\underline{v}, w}(g;q) + c\leq \underline{v}+w$ then the creator $x \in \mathcal{X}$ strictly prefers manual creation if $V(x;q) > \underline{v}$ and strictly prefers to use GenAI if $V(x;q) < \underline{v}$. 
        \item If $\underline{v}+w=V^{\underline{v}, w}(g;q) + c$ then the creator $x$ with $V(x;q) = \underline{v}$ is indifferent between manual creation and using GenAI. If $\underline{v}+w > V^{\underline{v}, w}(g;q) + c$, then the creator $x$ with $V(x;q) = \underline{v}$ strictly prefers manual creation.
        \item If $\underline{v} + w < V^{\underline{v}, w}(g;q) + c$ then the profit maximization creation decision of any creator $x \in \mathcal{X}$ under $(\underline{v}, w)$ is also a profit maximization creation decision under $(\tilde{\underline{v}}, \tilde{w})$ where $\tilde{\underline{v}} := V^{\underline{v}, w}(g;q) + c - w$ and some $\tilde{w} \in [0,w]$ chosen such that $\tilde{\underline{v}}+\tilde{w} = V^{\tilde{\underline{v}}, \tilde{w}}(g;q) + c \geq \tilde{\underline{v}}$.
        \item If $\underline{v} > V^{\underline{v}, w}(g;q) + c$ any creator $x\in\mathcal{X}$ who strictly prefers GenAI under no compensation $(\underline{v}_0, w_0)$, where $\underline{v}_0 = V(g;q)+c, w_0 = 0$, also strictly prefers GenAI under $(\underline{v}, w)$. Additionally, if none of the creators $x \in \mathcal{X}$ strictly prefers manual creation then the profit maximization creation decision of any $x \in \mathcal{X}$ under $(\underline{v}, w)$ is also a profit maximization creation decision under $(\underline{v}_0, w_0)$.
    \end{enumerate}
\end{lemma}

\subsection{Generalization and $\mathcal{X}$--based Compensation Schemes}

So far, we motivated the consideration of revenue--threshold class of compensation schemes from the implementation perspective, but it remains a question whether the platform can do much better if it is allowed to consider more general class of compensation schemes. For starters, we observe from (\ref{compensationformula}) that $W^{\underline{v}, w}(x;q)$ is technically a function of $q$, it only depends implicitly on $x$ via $q(x)$. On the opposite end of the spectrum we have a class of compensation schemes which are independent of $q$:

\begin{definition}
    We call a compensation scheme $W : \mathcal{X}\times \mathcal{D}(\mathcal{X}) \rightarrow \mathbb{R}_{\geq 0}$ an \emph{$\mathcal{X}$--based} compensation scheme if $W(x;q) := W(x)$ is independent of $q$.
\end{definition}

A nice property is that an equilibrium always exists under any $\mathcal{X}$--based compensation scheme.

\begin{lemma}\label{xbasedexistencelemma}
    Any $\mathcal{X}$--based compensation scheme $W:\mathcal{X}\times \mathcal{D}(\mathcal{X})\rightarrow \mathbb{R}_{\geq 0}$, $W(x;q) := W(x)$, admits an equilibrium. Further, the expected quantity of GenAI contents is strictly positive at any equilibrium under $W$: i.e. $\int_{\mathcal{X}}\beta(z)p(z)dz > 0$.
\end{lemma}

This is in contrast to the fact that the existence of an equilibrium is not guaranteed under a general compensation scheme $W:\mathcal{X}\times \mathcal{D}(\mathcal{X})\rightarrow \mathbb{R}_{\geq 0}$ due to the potentially discontinuity in $q$. Proposition \ref{equilibriumnonexistenceproposition} shows that revenue--threshold compensation scheme $(\underline{v}, w)$, for a certain range of parameters $\underline{v}, w\geq 0$, provides an example of a compensation scheme without an equilibrium. The following result shows the importance of $\mathcal{X}$--based compensation schemes:

\begin{lemma}[$\mathcal{X}$--based Equivalence]\label{rationalexpectationlemma}
    Suppose that $(\tilde{\beta}, \tilde{q})$ is an equilibrium under the platform's compensation scheme $\widetilde{W} : \mathcal{X}\times\mathcal{D}(\mathcal{X}) \rightarrow \mathbb{R}_{\geq 0}$ then $(\tilde{\beta}, \tilde{q})$ is also an equilibrium under an $\mathcal{X}$--based compensation scheme $W : \mathcal{X}\times \mathcal{D}(\mathcal{X})\rightarrow \mathbb{R}_{\geq 0}$ given by $W(x;q) = W(x) := \widetilde{W}(x;\tilde{q})$.
\end{lemma}

The converse of Lemma \ref{rationalexpectationlemma} is not necessary true, i.e. if $(\tilde{\beta}, \tilde{q})$ is an equilibrium under $W$ given by $W(x;q) = W(x) := \widetilde{W}(x;\tilde{q})$, it is not necessary true that $(\tilde{\beta}, \tilde{q})$ is an equilibrium under $\widetilde{W} = \widetilde{W}(x;q)$. Lemma \ref{rationalexpectationlemma} implies that if our objective is to find, under minimal constraints, the compensation scheme for the platform that maximizes the profit at equilibrium, then it is sufficiently general to consider the class of $\mathcal{X}$--based compensation schemes. The next result refines this further by showing that, under a mild condition, the optimal $\mathcal{X}$--based compensation scheme follows a simple description.

\begin{proposition}\label{rthresholdschemeproposition}
    Let $r(x) := p(x)/g(x)$. Suppose that the distribution of $r(x)$ under $p$ is atomless with density, and that $\gamma \leq \frac{2\alpha}{\alpha+1}$. Consider any compensation scheme $\widetilde{W} : \mathcal{X}\times \mathcal{D}(\mathcal{X}) \rightarrow \mathbb{R}_{\geq 0}$ with an equilibrium $(\tilde{\beta}, \tilde{q})$, then there exists an $\mathcal{X}$--based compensation scheme $W : \mathcal{X}\times\mathcal{D}(\mathcal{X})\rightarrow \mathbb{R}_{\geq 0}$ given for some $\underline{r} \in [\inf_{x\in \mathcal{X}}r(x), \sup_{x\in \mathcal{X}}r(x)]$ and $w\geq 0$ by:
    \begin{equation}\label{xbasedrthresholdcompensation}
        W(x;q) = W(x) := w \cdot \mathbbm{1}[r(x) \geq \underline{r}],
    \end{equation}
    with an equilibrium $(\beta, q)$, such that $\Pi^W(q) \geq \Pi^{\widetilde{W}}(\tilde{q})$.
\end{proposition}

The idea behind Proposition \ref{rthresholdschemeproposition} is simple, a content $y\in \mathcal{X}$ with high $p(y)$ is in high demand, while a low $g(y)$ indicates that $y$ is not easily produced with GenAI, therefore a high $r(y)$ indicates a high benefit for the platform to encourage the creator $y$ to manually creates the content $y$. Although the full formal proof is quite long, here we give a rough outline. Start from an arbitrary $\widetilde{W}$. The atomless assumption allows us to perturbatively adjust $\widetilde{W}$ to $W$, solve for the new equilibrium to the first--order, and show an incremental first--order improvement to the profit. Suppose that at equilibrium $(\tilde{\beta}, \tilde{q})$ under $\widetilde{W}$, if we can find a pair of small subsets $\overline{\Delta}\subset \mathcal{X}$, $\underline{\Delta} \subset \mathcal{X}$ with equal measure and $\inf_{y\in \overline{\Delta}}r(y) > \sup_{x\in \underline{\Delta}}r(x)$ such that creators $y \in \overline{\Delta}$ are more likely to use GenAI than creators $x\in \underline{\Delta}$. Then we argue that a under a new compensation scheme $W$ which re--allocates compensation from creators in $\underline{\Delta}$ to $\overline{\Delta}$ yields an equilibrium $(\beta, q)$ with an incremental profit compared to $(\tilde{\beta}, \tilde{q})$. This is because it takes a lower compensation to incentivize any creators $y \in \overline{\Delta}$ to create content manually with the same probability compared to any creators $x \in \underline{\Delta}$, because $r(y) > r(x)$. On the other hand, the platform gains a greater marginal benefit from the manual creation of $y\in\overline{\Delta}$ compared to $x\in\underline{\Delta}$ since the market is more under--supplied at $y$. This shows that we can replace $\widetilde{W}$ with $W$ such that $\Pi^W(q) \geq \Pi^{\widetilde{W}}(\tilde{q})$, where under $W$, there exists some threshold $\underline{r}$ for the minimum $r(y)$ a content $y \in \mathcal{X}$ needs to be qualified for a compensation, and $W(y) \geq W(x)$ for $x,y\in r^{-1}[\underline{r}, \infty)$ such that $r(y) \geq r(x)$. Now, for any small $\Delta_0 \subset r^{-1}[\underline{r}, \infty)$, the creators in $\Delta_0$ are indifferent between manual creation and using GenAI. We can show that the platform's profit from an indifferent creator is a concave function of the compensation if $\gamma \leq \frac{2\alpha}{\alpha+1}$, hence by Jensen's inequality, we can replace $W|_{\Delta_0}$ with a constant $\mathbb{E}_{y\sim p}[W(y)|\Delta_0]$ and incrementally improve the profit. The condition $\gamma\leq \frac{2\alpha}{\alpha+1}$ holds when the engagement is elastic (high $\alpha$) with respect to the number of consumers and the platform commission rate $\gamma$ is low. A high $\alpha$ indicates the platform is able to present each consumer with contents that match well with her preference, while a typical commission rates of major platforms such as Meta or Youtube is about $\gamma = 0.30$ to $0.45$, which is relatively low (\cite{wu2026self}).

The $\mathcal{X}$--based compensation scheme (\ref{xbasedrthresholdcompensation}) requires the platform to have the ability to compute $r(y) = p(y)/g(y)$ for any given content $y$. This, in fact, suggests that the best possible compensation scheme is the one relies on an AI--detector, since:
\begin{equation*}
    \mathbb{P}[\text{AI}\ |\ y] = \frac{\mathbb{P}[y\ |\ \text{AI}]\cdot \mathbb{P}[\text{AI}]}{\mathbb{P}[y]} = \frac{g(y)\cdot \int_{\mathcal{X}}\beta(z)p(z)dz}{q(y)} \propto \frac{V(y;q)^{1/\alpha}}{r(y)}.
\end{equation*}
In other words, to decide if $y$ is qualified for a compensation under (\ref{xbasedrthresholdcompensation}), the platform can divide the ex--post revenue generated by the content $y$ by the output of AI--detector on $y$ and compare the result with the pre--specified threshold. Thus, the feasibility of computing $r(y)$ is equivalent to that of an AI--detector. The main challenge is to determine $g(y)$. Usually, this can be done if the platform knows which GenAI model likely to have generated $y$ and has a white--box access to (the correct version of) such model, although it would still be computationally demanding. For example, for the score--based SDE diffusion model (\cite{song2020score}), with an access to the trained score function, the platform can write the probability flow ODE and integrate it to get $g(y)$. Realistically, given the variety of third--party GenAI models creators can use, it would be difficult for the platform to determine the correct model to test. Meanwhile, determining $p(y)$ is potentially easier, if the platform has a trained recommender system, then a list of consumers the given content $y$ will be recommended to can be populated. Otherwise, the platform may estimate $p(y)$ from some census data, or market survey.

One of the main equilibrium characterization results in \S\ref{equilibriumanalysissection}, namely Proposition \ref{equilibriumnonexistenceproposition}, is that each revenue--threshold compensation scheme $(\underline{v}, w)$ has an $\mathcal{X}$--based equivalence given by (\ref{xbasedrthresholdcompensation}) with the threshold $\underline{r}(\underline{v})$ given as a function of $\underline{v}$. This shows the revenue--threshold compensation schemes are among the best for platform's profit maximization (in the sense of Proposition \ref{rthresholdschemeproposition}), while being intuitively simple and allowing us to by--pass the computation of $r(x) = p(x)/g(x)$. For the remainder of this paper, we will focus exclusively on revenue--threshold compensation schemes.

\section{Equilibrium Analysis}\label{equilibriumanalysissection}

In this section, we characterize the equilibrium under the platform's $(\underline{v}, w)$ compensation scheme for $\underline{v}, w\geq0$ and a GenAI model $g$ is accessible to all creators. We may decompose the space of all creators as $\mathcal{X} = \mathcal{X}_{\text{AI}}\sqcup\mathcal{X}_{\text{IN}}\sqcup\mathcal{X}_{\text{H}}$, where $\mathcal{X}_{\text{AI}}$ is a subset of creators who strictly prefer GenAI ($\beta(x) = 1$ for all $x\in \mathcal{X}_{\text{IN}})$, $\mathcal{X}_{\text{H}}$ is a subset of creators who strictly prefer to manually generate content ($\beta(x) = 0$ for all $x \in \mathcal{X}_{\text{H}}$), and $\mathcal{X}_{\text{IN}} := \mathcal{X}\setminus (\mathcal{X}_{\text{AI}}\sqcup \mathcal{X}_{\text{H}})$ is a subset of creators who are indifferent between both methods ($\beta(x)\in[0,1]$ for all $x \in \mathcal{X}_{\text{IN}}$). 

If $
\underline{v} \leq V^{\underline{v}, w}(g;q) + c \leq \underline{v} + w$ then from the decision rules we derived in Lemma \ref{decisionundercompensationlemma}, we also have $\mathcal{X}_{\text{AI}} = \{x \in \mathcal{X}\ |\ V(x;q) < \underline{v}\}$, $\mathcal{X}_{\text{H}} = \{x\in \mathcal{X}\ |\ V(x;q) > \underline{v}\}$, and if $\underline{v} + w = V^{\underline{v},w}(g;q) + c$ then $\mathcal{X}_{\text{IN}} = \{x\in \mathcal{X}\ |\ V(x;q) = \underline{v}\}$, therefore the decomposition is nothing more than the level set decomposition of $V(x;q)$. In this case, the creators in $\mathcal{X}_{\text{AI}}$ receive no compensation, while the creators in $\mathcal{X}\setminus \mathcal{X}_{\text{AI}}$ each receive the same compensation of $w$. We begin with a simple observation that, in fact, no creator will strictly prefer to manually create content over using GenAI at equilibrium. 

\begin{lemma}\label{xhisemptylemma}
    We have $\mathcal{X}_{\text{H}} = \emptyset$ and $\mathcal{X}_{\text{AI}}$ has a positive measure at any equilibrium under a platform compensation scheme $(\underline{v}, w)$.
\end{lemma}

Lemma \ref{xhisemptylemma} is consistent with the rapid and widespread adoption of GenAI. The decomposition now simplifies to $\mathcal{X} = \mathcal{X}_{\text{AI}} \sqcup \mathcal{X}_{\text{IN}}$. The following shows that the availability of GenAI always distorts the equilibrium content distribution $q$ away from $p$, unless the GenAI model is perfectly trained: $g = p$, which is not possible in practice as discussed in \S\ref{generativeaisection}. Consequently, the platform's revenue will always be suboptimal $R^W(q) = \gamma\int_{\mathcal{X}}p(x)^\alpha q(x)^{1-\alpha}dx < \gamma$, but we will show how compensation scheme can offer improvement regardless.

\begin{lemma}["AI-slop"]\label{aislopresult}
If $g\neq p$ then there exists no equilibrium with $q = p$ under any platform compensation scheme $(\underline{v}, w)$.
\end{lemma}

When $\underline{v} \leq V^{\underline{v}, w}(g;q) + c \leq \underline{v} + w$, it is clear that the threshold $\underline{v}$ plays an important role in characterizing the decomposition $\mathcal{X} = \mathcal{X}_{\text{AI}}\sqcup\mathcal{X}_{\text{IN}}$, and hence, the equilibrium. Unlike in the pre--GenAI case, the decision of one $x$ also impacts the revenue of the other $x'$ distant away. Therefore, it appears more challenging to understand $x$ creation decision based on $p(x)/q(x)$, i.e. to directly compare $V(x;q)$ with the threshold $\underline{v}$ and $V^{\underline{v}, w}(g;q)$, since $q(x)$ depends on decision of all other creators at equilibrium. Instead, let us introduce some additional definitions before we proceed. For any given revenue level $\underline{v} \geq 1-\gamma$, let $r(x) := p(x)/g(x)$ and define $\underline{r}(\underline{v})\geq 0$ to be the solution to:
\begin{equation}\label{definingbarrbarv}
    \int_{\mathcal{X}}\max\left\{\frac{\underline{r}(\underline{v})}{r(y)}, 1\right\} p(y)dy = \left(\frac{\underline{v}}{1-\gamma}\right)^{1/\alpha}.
\end{equation}
Note that the LHS is continuous and strictly monotonically increasing in $\underline{r}(\underline{v})\in [\inf_{x \in \mathcal{X}}r(x), \infty)$ with a range $[1,\infty)$, thus, the equation (\ref{definingbarrbarv}) has a unique solution given $\underline{v} \geq 1-\gamma$. It is clear that $\underline{r}(\underline{v})$ is continuous and strictly monotonically increasing in $\underline{v}$. Additionally, we define:
\begin{equation}\label{mgmpvbarvdefinition}
    \begin{aligned}
     M_g(\underline{v}) &:= \int_{r^{-1}[0,\underline{r}(\underline{v}))}p(y)^\alpha g(y)^{1-\alpha}dy, \qquad M_p(\underline{v}) := \left(\frac{1-\gamma}{\underline{v}}\right)^{1/\alpha}\int_{r^{-1}[\underline{r}(\underline{v}), \infty)}p(y)dy\\
     \widetilde{V}^{\underline{v}}(g) &:= \frac{\underline{r}(\underline{v})}{1-M_p(\underline{v})}\left(\frac{1-\gamma}{\underline{v}}\right)^{1/\alpha}\left(\frac{1-\gamma}{\underline{r}(\underline{v})^\alpha}M_g(\underline{v})+c\right)-c,     
    \end{aligned}
\end{equation}
for $\underline{v} \geq 1-\gamma$. We can see that $M_g(\underline{v})$ is monotonically increasing while $M_p(\underline{v})$ is monotonically decreasing as a function of $\underline{v}$, however they might not be continuous in $\underline{v}$, especially if $r(x)$ is not atomless under $p$, i.e. informally, $r$ has some `flat region'.

\begin{proposition}\label{generalequilibriumproposition}
    Let $r(x) := p(x)/g(x)$ and $\underline{v}$ be given such that
    \begin{equation}\label{admissibleconditionforvbar}
        1-\gamma < \underline{v} < (1-\gamma)\sup_{x\in\mathcal{X}}r(x)^\alpha, \quad \text{and} \quad \underline{v} \leq \widetilde{V}^{\underline{v}}(g)+c
    \end{equation}
    then $(\beta, q)$, where
    \begin{equation}\label{generalequilibriumform}
    \begin{aligned}
        q(x) &= \left(\frac{1-\gamma}{\underline{v}}\right)^{1/\alpha}\left(p(x)\cdot \mathbbm{1}[r(x) \geq \underline{r}(\underline{v})] + \underline{r}(\underline{v}) g(x)\cdot \mathbbm{1}[r(x) < \underline{r}(\underline{v})]\right)\\
        \beta(x) &= \left(1 - \left(\frac{1-\gamma}{\underline{v}}\right)^{1/\alpha}\left(1- \frac{\underline{r}(\underline{v})}{r(x)}\right)\right)\cdot \mathbbm{1}[r(x) \geq \underline{r}(\underline{v})] + \mathbbm{1}[r(x) < \underline{r}(\underline{v})]
    \end{aligned},
    \end{equation}
    is an equilibrium under the platform's compensation scheme $(\underline{v}, w)$ where $w := \widetilde{V}^{\underline{v}}(g) + c - \underline{v}$. It also follows that $\mathcal{X}_{\text{AI}} = r^{-1}[0,\underline{r}(\underline{v}))$, $\mathcal{X}_{\text{IN}} = r^{-1}[\underline{r}(\underline{v}), \infty)$, and $V^{\underline{v}, w}(g;q) = \widetilde{V}^{\underline{v}}(g)$.
    
    Conversely, if $(\beta, q)$ is an equilibrium under a platform's compensation scheme $(\underline{v}, w)$ characterized by $\mathcal{X} = \mathcal{X}_{\text{AI}}\sqcup \mathcal{X}_{\text{IN}}$ such that $\mathcal{X}_{\text{IN}} \neq \emptyset$ and $\underline{v} \leq V^{\underline{v},w}(g;q) + c = \underline{v} + w$, then $1-\gamma < \underline{v} \leq (1-\gamma)\sup_{x\in \mathcal{X}}r(x)^\alpha$, $V^{\underline{v}, w}(g;q)=\widetilde{V}^{\underline{v}}(g)$ and $(\beta, q)$ is given by (\ref{generalequilibriumform}).
\end{proposition}

Note, as can be seen from (\ref{generalequilibriumform}), that $q(x) \rightarrow p(x)$ pointwise, for all $x \in \mathcal{X}$, and $\int_{\mathcal{X}}\beta(y)p(y)dy$ decreases as $\underline{v}\searrow 1-\gamma$, but it is not necessarily true that $\int_{\mathcal{X}}\beta(y)p(y)dy$ converges to zero. In other words, even though the distribution of contents $q$ can be regulated to be as close to $p$ as the platform wishes, the amount of GenAI contents remains positive and bounded from zero. 

The condition (\ref{admissibleconditionforvbar}) for $\underline{v}$ is critical in the application of Proposition \ref{generalequilibriumproposition}. In particular, although the formula for $(\beta, q)$ in (\ref{generalequilibriumform}) remains valid for all $\underline{v}\geq 1-\gamma$, without the condition (\ref{admissibleconditionforvbar}) the resulting $(\beta, q)$ may not correspond to an equilibrium from any compensation scheme $(\underline{v}, w)$. The $\underline{v} \leq \widetilde{V}^{\underline{v}}(g) + c$ part of the condition (\ref{admissibleconditionforvbar}) can always be satisfied assuming the distribution of $r(x)$ under $p$ is supported inside $\mathbb{R}_{>0}$ with density and at most a finite number of point masses, since we can always choose $\underline{v} > 1-\gamma$ sufficiently close to $1-\gamma$, then $\underline{v} < \widetilde{V}^{\underline{v}}(g)+c$. To see this, note that $\underline{r}(\underline{v})\searrow \underline{\underline{r}} := \inf_{x\in \mathcal{X}}r(x)$ as $\underline{v}\searrow 1-\gamma$. If the distribution of $r(x)$ has a point mass at $\underline{\underline{r}}$ then 
\begin{equation*}
    M_g(\underline{v}) \searrow \frac{1}{\sqrt{\underline{\underline{r}}}}\int_{r^{-1}\{\underline{\underline{r}}\}}p(y)dy, \qquad M_p(\underline{v}) \nearrow 1 - \int_{r^{-1}\{\underline{\underline{r}}\}}p(y)dy,
\end{equation*}
hence, the first term of the expression of $\widetilde{V}^{\underline{v}}(g)$ in (\ref{mgmpvbarvdefinition}) converges to $1-\gamma$, so we have $\widetilde{V}^{\underline{v}}(g)+c > 1-\gamma$ for all $\underline{v} > 1-\gamma$ sufficiently close to $1-\gamma$. If the distribution of $r(x)$ has no point mass at $\underline{\underline{r}}$ then by L'Hopital rule, we also find that the limit of the first term of $\widetilde{V}^{\underline{v}}(g)$ in (\ref{mgmpvbarvdefinition}) is $1-\gamma$. But note that the second term in fact tends to infinity as $c/(1-M_p(\underline{v}))\rightarrow \infty$, so we have $\widetilde{V}^{\underline{v}}(g) + c\rightarrow \infty$.

Meanwhile, we can see that $\lim_{\underline{v}\rightarrow \infty}\widetilde{V}^{\underline{v}}(g) + c = 0$, so $\underline{v} > \widetilde{V}^{\underline{v}}(g) + c$ for all sufficiently large $\underline{v}$. However, it may not be possible to find $\underline{v}$ such that $\underline{v} = \widetilde{V}^{\underline{v}}(g)+c$ since $\widetilde{V}^{\underline{v}}(g)$ may not be continuous in $\underline{v}$ when $r(x)$ is not atomless under $p$. The solution $\underline{v}_0$ to $\underline{v} = \widetilde{V}^{\underline{v}}(g)+c$ has a significant of being the equilibrium revenue level where the creator is indifferent between GenAI and manual creation under no compensation. Any threshold $\underline{v} > \underline{v}_0$ will be ineffective since no creator has a sufficiently high revenue to reach $\underline{v}$ and the equilibrium is given by the equilibrium under no compensation $(\underline{v}_0, w_0=0)$. Therefore, it is worth considering when the existence of the solution $\underline{v}_0$ is guaranteed. 

\begin{lemma}[Existence and Uniqueness of The Level $\underline{v}_0$]\label{solutionv0lemma}
    Suppose that the distribution of $r(x) := p(x)/g(x)$ under $p$ is given by a distribution compactly supported inside $\mathbb{R}_{>0}$ with density and at most a finite number of point masses and let $\overline{\overline{r}} := \sup_{x\in \mathcal{X}}r(x)$.
    \begin{enumerate}
        \item $\widetilde{V}^{\underline{v}}(g)$ is a piecewise continuous function for $\underline{v} \in (1-\gamma, (1-\gamma)\overline{\overline{r}}^\alpha)$, with at most finite discontinuities at point masses. Moreover, 
        \begin{multline*}\label{discontinuityofvvbar}
            (1-\gamma)\mathbb{E}_{x\sim g}r(x)^\alpha - \lim_{\underline{v}\nearrow (1-\gamma)\overline{\overline{r}}^\alpha}\widetilde{V}^{\underline{v}}(g)\\
            = \frac{\int_{r^{-1}[\overline{\overline{r}},\infty)}g(y)dy}{\int_{r^{-1}[0,\overline{\overline{r}})}g(y)dy}\left((1-\gamma)\overline{\overline{r}}^\alpha - (1-\gamma)\mathbb{E}_{x\sim g}r(x)^\alpha - c\right)\numberthis
        \end{multline*}
        which shows the discontinuity if a point mass exists at $\underline{v} = (1-\gamma)\overline{\overline{r}}^\alpha$.
        \item If $\underline{v}_1\leq \widetilde{V}^{\underline{v}_1}(g) + c$ for some $\underline{v}_1 \in (1-\gamma, (1-\gamma)\overline{\overline{r}}^\alpha)$, then $\widetilde{V}^{\underline{v}}(g)+c$ is monotonically decreasing in $\underline{v} \in (1-\gamma, \underline{v}_1]$. In particular, there exists at most one solution $\underline{v}_0$ to $\underline{v} = \widetilde{V}^{\underline{v}}(g) + c$.
        \item Additionally, suppose that the distribution of $r(x)$ under $p$ has at most one point mass at $\overline{\overline{r}}$. If $(1-\gamma)\sup_{x\in \mathcal{X}}r(x)^\alpha > (1-\gamma)\mathbb{E}_{x\sim g}r(x)^\alpha + c$, then there exists a solution $\underline{v}_0 \in (1-\gamma, (1-\gamma)\overline{\overline{r}}^\alpha)$ to $\underline{v} = \widetilde{V}^{\underline{v}}(g) + c$, otherwise we have $\underline{v} < \widetilde{V}^{\underline{v}}(g) + c$ for all $\underline{v} \in (1-\gamma, (1-\gamma)\overline{\overline{r}}^\alpha)$.
    \end{enumerate}
\end{lemma}

The following can be considered as the counterpart result to Proposition \ref{generalequilibriumproposition}, characterizing the equilibrium when (\ref{admissibleconditionforvbar}) does not hold, and that the platform can do no better than choosing to give no compensation in this case.

\begin{proposition}[``Pure--AI Platform"]\label{pureailemma}
    Suppose that the distribution of $r(x) := p(x)/g(x)$ under $p$ is given by a distribution compactly supported inside $\mathbb{R}_{>0}$ with density and at most a finite number of point masses and that:
    \begin{equation}\label{pureaicondition}
        \min\left\{(1-\gamma)\mathbb{E}_{x\sim g}r(x)^\alpha + c, \underline{v}\right\} \geq (1-\gamma)\sup_{x\in \mathcal{X}}r(x)^\alpha,
    \end{equation}
    then $(\beta = 1,q = g)$ is the only possible equilibrium under the platform compensation scheme $(\underline{v}, w)$ for any $w \geq 0$, and it is the unique equilibrium if the inequality in (\ref{pureaicondition}) is strict, or $\underline{v} + w \leq (1-\gamma)\mathbb{E}_{x\sim g}r(x)^\alpha + c$. Therefore, if $(1-\gamma)\mathbb{E}_{x\sim g}r(x)^\alpha+c \geq (1-\gamma)\sup_{x\in \mathcal{X}}r(x)^\alpha$, then any choice of the compensation scheme $(\underline{v}, w)$ for the platform that satisfies (\ref{pureaicondition}) is weakly dominated by $(\underline{v}, w = 0)$, under which $(\beta=1,q=g)$ is the unique equilibrium.

    Conversely, if $(\beta=1,q=g)$ is an equilibrium under any platform compensation scheme $(\underline{v}, w)$ then $(1-\gamma)\mathbb{E}_{x\sim g}r(x)^\alpha+c \geq (1-\gamma)\sup_{x\in \mathcal{X}}r(x)^\alpha$.
\end{proposition}

Proposition \ref{pureailemma} states that if the GenAI model is sufficiently well--trained, i.e. $g$ is sufficiently close to $p$, then the only equilibrium is where all the creators use AI, unless the platform offers a sufficiently generous compensation scheme. Clearly, if $g = p$ then $r(x) = 1$ for all $x$, and (\ref{pureaicondition}) becomes $(1-\gamma)+c > 1-\gamma$, which is trivial.

We have mentioned that an equilibrium does not always exist for a general compensation scheme $W : \mathcal{X}\times\mathcal{D}(\mathcal{X}) \rightarrow \mathbb{R}_{\geq 0}$. So far, we focus on the characterization of equilibrium $(\beta, q)$ with $\underline{v} + w = V^{\underline{v}, w}(g;q) + c \geq \underline{v}$ (i.e. Proposition \ref{generalequilibriumproposition}) or $\underline{v} > V^{\underline{v}, w}(g;q)+c$ (i.e. Proposition \ref{pureailemma}). Moreover, we have seen from Lemma \ref{decisionundercompensationlemma} that if an equilibrium $(\beta, q)$ under some 
$(\underline{v}, w)$ satisfies $\underline{v} + w < V^{\underline{v}, w}(g;q) + c$ then it is also an equilibrium under $(\tilde{\underline{v}}, \tilde{w})$ where $\tilde{\underline{v}} + \tilde{w} = V^{\tilde{\underline{v}}, \tilde{w}}(g;q) + c$, bringing us back to the previous case. The only possibility left to consider is an equilibrium with $\underline{v} + w > V^{\underline{v}, w}(g;q) + c \geq \underline{v}$. In the following, we argue that such an equilibrium cannot exists, in particular, no equilibrium exists under any revenue--threshold compensation scheme $(\underline{v}, w)$ with a sufficiently high $w > 0$.

\begin{proposition}\label{equilibriumnonexistenceproposition}
    If $\underline{v}$ satisfies condition (\ref{admissibleconditionforvbar}) but $\underline{v} + w > \widetilde{V}^{\underline{v}}(g) + c$ then there exists no equilibrium under the platform's compensation scheme $(\underline{v},w)$.
\end{proposition}

The key is that the creator's utility under the revenue--threshold compensation scheme is not continuous in the common belief $q$, thus an equilibrium is not guaranteed. When $w \geq 0$ is too large, under the setting outlined in Proposition \ref{equilibriumnonexistenceproposition}, a creator $x$ will either strictly prefers GenAI if $q(x)$ is just above a threshold, or strictly prefers manual creation when $q(x)$ is at or just below the threshold, making an equilibrium impossible. This result shows that Proposition \ref{generalequilibriumproposition} and Proposition \ref{pureailemma} constitute the complete classification of all the possible equilibrium under the class of revenue--threshold compensation schemes. Given $\underline{v}$ satisfying the condition for Proposition \ref{generalequilibriumproposition} or Proposition \ref{pureailemma} we can appropriately choose $w$ to ensure the existence of the corresponding equilibrium. But choosing $w = \widetilde{V}^{\underline{v}}(g) + c - \underline{v}$ might not be practical if the platform does not have a prior knowledge of the expected revenue from using GenAI, or know $p$ or $g$. We argue informally in the following that this concern should be limited in practice, and the platform can choose a wide--range of $w$ after specifying $\underline{v}$ for an equivalent equilibrium outcome to choosing $w = \widetilde{V}^{\underline{v}}(g) + c - \underline{v}$. Recall our heuristic derivation which motivates the definition of our model, where the platform consists of $N$ consumers and $N$ creators. The revenue $V_N(y;q)$ from the content $y$ is a random variable converges (in mean) to $V(y;q)$ as $N\rightarrow \infty$ and the compensated revenue is given by $V^{\underline{v}, w}_N(x;q) := V_N(x;q) + w\cdot \mathbbm{1}[V_N(x;q) \geq \underline{v}]$. However, at finite $N$, the positive variance of $V_N(x;q)$ means that the expected revenue $\mathbb{E}V_N(x;q)$ for each creator $x$ is continuous with respect to the creation decision $\beta$. Suppose that $\underline{v} \leq V^{\underline{v}, w}_N(g;q)+c \leq \underline{v}+w$, where the second inequality may be strict, a creator $x$ can be indifferent at finite $N$ equilibrium if revenue from the event that $V_N(x;q) < \underline{v}$ balances the revenue from the event that $V_N(x;q) \geq \underline{v}$ and we get $\mathbb{E}V^{\underline{v}, q}_N(x;q) = V^{\underline{v}, w}(g;q) + c$. As $N$ becomes large, $V_N(x;q)$ becomes more concentrated around the mean, which means all the indifferent creators must be concentrated around $\underline{v}$. Following this informal reasoning, we can argue that $V(x;q) = \underline{v}$ gives a indifference condition, for all $w \geq 0$ sufficiently large such that $\underline{v} + w \geq V^{\underline{v}, w}(g;q)+c$. If we adopt this view then the equilibrium classification in Proposition \ref{generalequilibriumproposition} and Proposition \ref{pureailemma} are already complete term of $\underline{v}$, and we can always assume $w \geq 0$ to be set sufficiently large.

Finally, we present the following result which guarantees that the platform revenue is always improved with higher compensation by aligning the resulting equilibrium content distribution $q$ closer to $p$. It remains for the platform to balance the revenue improvement benefit with the total compensation costs, and using Proposition \ref{generalequilibriumproposition} and Proposition \ref{pureailemma} this becomes an optimization problem in a single variable $\underline{v}$. We show under a mild condition that the platform's profit can be maximized with a certain threshold $\underline{v}^* > 1-\gamma$.

\begin{proposition}\label{platformrevennueproposition}
    Suppose that the distribution of $r(x) := p(x)/g(x)$ under $p$ is given by a distribution compactly supported inside $\mathbb{R}_{>0}$ with density and at most a finite number of point masses. 
    
    At equilibrium $(\beta, q)$ characterized by $\underline{v}$ satisfying the condition (\ref{admissibleconditionforvbar}) of Proposition \ref{generalequilibriumproposition}, the platform's revenue can be written as:
    \begin{equation}
        R(\underline{v}) := R^{\underline{v}, w}(q) = \gamma\left(\frac{\underline{v}}{1-\gamma}\right) M_p(\underline{v}) + \gamma\underline{r}(\underline{v})^{1-\alpha}\left(\frac{1-\gamma}{\underline{v}}
        \right)^{1/\alpha-1}M_g(\underline{v})
    \end{equation}
    and it is a monotonically decreasing continuous function in $\underline{v}$. The platform pays the same expected compensation $w = \widetilde{V}^{\underline{v}}(g)+c-\underline{v}$ to all $x \in \mathcal{X}_{\text{IN}}$ and no compensation to any $x \in \mathcal{X}_{\text{AI}}$, therefore, the platform's profit is given by:
    \begin{multline*}\label{platformcompensationschemeprofitvbar}
        \Pi(\underline{v}) := \Pi^{\underline{v}, w}(q) = \left(\frac{\underline{v}}{1-\gamma}\right)M_p(\underline{v}) + \underline{r}(\underline{v})^{1-\alpha}\left(\frac{1-\gamma}{\underline{v}}\right)^{1/\alpha-1}\left(\frac{\gamma-M_p(\underline{v})}{1-M_p(\underline{v})}\right)M_g(\underline{v})\\
        - c\underline{r}(\underline{v})\left(\frac{1-\gamma}{\underline{v}}\right)^{1/\alpha}\frac{M_p(\underline{v})}{1-M_p(\underline{v})}.\numberthis
    \end{multline*}
    At equilibrium $(\beta, q)$ characterized by $\underline{v}$ satisfying the condition (\ref{pureaicondition}) of Proposition \ref{pureailemma} with strict inequality, the platform's revenue and profit are given by:
    \begin{equation}\label{pureaiplatformcompensationschemeprofitvbar}
        \Pi(\underline{v}) := \Pi^{\underline{v}, w=0}(q) = R^{\underline{v}, w=0}(q) = \gamma\int_{\mathcal{X}}p(y)^\alpha g(y)^{1-\alpha}dy = \gamma M_g(\underline{v}).
    \end{equation}
    Additionally, suppose that $\overline{\overline{r}} := \sup_{x\in \mathcal{X}}r(x)$ is the only possible point mass, then the platform's profit $\Pi(\underline{v})$ is upper--semicontinuous over the set of $\underline{v}$ satisfying either the condition (\ref{admissibleconditionforvbar}) or the condition (\ref{pureaicondition}) with the only possible discontinuity at $\underline{v} = (1-\gamma)\overline{\overline{r}}^\alpha$ and a profit maximization threshold $\underline{v}^*$ exists.
\end{proposition}

Recall that we assumed that $p$ and $g$ are not known to the platform. Thus, the closed--form expression for $\Pi^{\underline{v}, w}(q)$ in Proposition \ref{platformrevennueproposition} may have limited practical use to the platform. However, Proposition \ref{platformrevennueproposition} suggests how profit maximization is feasible for a broad class of $p$ and $g$. Moreover, the problem reduces to an optimization problem in a single variable $\underline{v} \geq 1-\gamma$ which the platform may choose to tackle via experimental methods such as A/B--testing, or estimate via other empirical techniques.

\section{Examples}\label{examplessection}

In this section, we consider some examples with $\mathcal{X} \subset \mathbb{R}$. We set $\alpha = 1/2$ and $c = 1-\gamma$ throughout.

\subsection{Two--Levels GenAI Distribution}\label{twolevelunifromexample}

This toy example represents one of the simplest setting which our model can be solved explicitly. Consider $\mathcal{X} := [0,1]$ with a uniform preference distribution $p = \Unif[0,1]$, and let the GenAI model be given by the density:
\begin{equation*}
    g(x) = \begin{cases}
        \overline{g}, &\quad x\in [0,1/2]\\
        \underline{g}, &\quad x\in (1/2, 1]
    \end{cases}
\end{equation*}
where $\underline{g} \in [0,1]$ and $\overline{g} := 2-\underline{g}$. We have $r(x) = 1/\overline{g} \in [0,1]$ if $x \in [0,1/2]$ and $r(x) = 1/\underline{g} \in [1,\infty)$ if $x\in (1/2,1]$. First, let us consider the case where $(1-\gamma)\sup_{x\in \mathcal{X}}\sqrt{r(x)} > (1-\gamma)\mathbb{E}_{x\sim g}\sqrt{r(x)} + c$, which can also be written as:
\begin{equation}\label{nopureaiconditionuniform}
    \frac{1-\gamma}{\sqrt{\underline{g}}} > \frac{1-\gamma}{2}\left(\sqrt{\underline{g}} + \sqrt{\overline{g}}\right) + c = \frac{1-\gamma}{2}\left(\sqrt{\underline{g}} + \sqrt{2-\underline{g}}\right) + c.
\end{equation}
This condition gives an upper--bound $\underline{g} < \underline{g}^* \approx 0.285$. Under this condition, an equilibrium exists where some creators will not strictly prefer GenAI, even without any compensation. Intuitively, the GenAI model which satisfies the bound $\underline{g} < \underline{g}^*$ is not sufficiently well--trained to serve all niches of consumers, leaving some gaps in the market for the manual creators. For any $\underline{v} \in (1-\gamma, (1-\gamma)/\sqrt{\underline{g}})$ we have: 
\begin{equation*}
    \underline{r}(\underline{v}) = \frac{\underline{v}^2/(1-\gamma)^2-1/2}{\overline{g}/2} = \frac{2\underline{v}^2/(1-\gamma)^2-1}{\overline{g}} \in \left(\frac{1}{\overline{g}},\frac{1}{\underline{g}}\right),
\end{equation*}
and therefore we have from (\ref{mgmpvbarvdefinition}) that:
\begin{equation*}
    M_g(\underline{v}) = \frac{\sqrt{\overline{g}}}{2}, \quad M_p(\underline{v}) = \frac{1}{2}\left(\frac{1-\gamma}{\underline{v}}\right)^2, \quad \widetilde{V}^{\underline{v}}(g) = \frac{1-\gamma}{\sqrt{2 - \left(\frac{1-\gamma}{\underline{v}}\right)^2}}+c\cdot \frac{\underline{g}}{\overline{g}}.
\end{equation*}
We can find the indifference level $\underline{v}_0$ without compensation by solving $\underline{v} = \widetilde{V}^{\underline{v}}(g) + c$, or equivalently:
\begin{equation}\label{v0equationuniform}
    \underline{v} = \frac{1-\gamma}{\sqrt{2 - \left(\frac{1-\gamma}{\underline{v}}\right)^2}} + \frac{2c}{\overline{g}}.
\end{equation}
Note that the RHS of (\ref{v0equationuniform}) is monotonically continuously decreasing such that it approaches $1-\gamma + \frac{2c}{\overline{g}}$ as $\underline{v}\searrow 1-\gamma$ and approaches $\frac{1-\gamma}{\sqrt{\overline{g}}} + \frac{2c}{\overline{g}}$ as $\underline{v}\nearrow (1-\gamma)/\sqrt{\underline{g}}$. It is clear that the RHS is greater than $\underline{v}$ as $\underline{v}\searrow 1-\gamma$. To show that the RHS of (\ref{v0equationuniform}) is lower than $\underline{v}$ as $\underline{v}\nearrow (1-\gamma)/\sqrt{\underline{g}}$ we can rearrange the inequality (\ref{nopureaiconditionuniform}):
\begin{equation*}
    \frac{1-\gamma}{\sqrt{\underline{g}}} > \frac{\underline{g}}{2}\left(\frac{1-\gamma}{\sqrt{\underline{g}}}\right) + \frac{1-\gamma}{2}\sqrt{\overline{g}} + c \quad \implies \quad \frac{1-\gamma}{\sqrt{\underline{g}}} > \frac{2}{\overline{g}}\left(\frac{1-\gamma}{2}\sqrt{\overline{g}} + c\right) = \frac{1-\gamma}{\sqrt{\overline{g}}} + \frac{2c}{\overline{g}}.
\end{equation*}
Therefore, the solution $\underline{v}_0$ to (\ref{v0equationuniform}) is guaranteed to exist in $(1-\gamma, (1-\gamma)/\sqrt{\underline{g}})$. Now, any $\underline{v} \in (1-\gamma, \underline{v}_0]$ corresponding to an equilibrium under the platform's compensation scheme $(\underline{v}, w)$ where $w := \widetilde{V}^{\underline{v}}(g) + c - \underline{v}$. Thus, the compensation $w$ increases as we lower $\underline{v}$ towards $1-\gamma$. The resulting equilibrium $(\beta, q)$ is given by Proposition \ref{generalequilibriumproposition}:
\begin{equation*}
    \beta(x) = \begin{cases}
        1, &\quad x \in [0,1/2]\\
        1 - \frac{2}{\overline{g}}\left[\left(\frac{1-\gamma}{\underline{v}}\right)^2-\underline{g}\right], &\quad x\in (1/2, 1]
    \end{cases}, \quad q(x) = \begin{cases}
        2 - \left(\frac{1-\gamma}{\underline{v}}\right)^2, &\quad x \in [0,1/2]\\
        \left(\frac{1-\gamma}{\underline{v}}\right)^2, &\quad x\in (1/2, 1]
    \end{cases}
\end{equation*}
and we have $\mathcal{X}_{\text{AI}} := [0,1/2], \mathcal{X}_{\text{IN}} := (1/2,1]$. Clearly, we have that $q(x)$ approaches $p(x)$ for each $x$ as $\underline{v}\searrow 1-\gamma$. From Proposition \ref{platformrevennueproposition}, the platform's revenue is:
\begin{equation*}
    R(\underline{v}) = \frac{\gamma}{2}\left(\frac{1-\gamma}{\underline{v}}\right) + \frac{\gamma}{2}\sqrt{2 - \left(\frac{1-\gamma}{\underline{v}}\right)^2}
\end{equation*}
and the platform's profit is:
\begin{equation*}
    \Pi^{\underline{v}, w}(q) = \frac{1}{2}\left(\frac{1-\gamma}{\underline{v}}\right) + \frac{1}{2}\sqrt{2 - \left(\frac{1-\gamma}{\underline{v}}\right)^2} - \frac{1-\gamma}{\sqrt{2 - \left(\frac{1-\gamma}{\underline{v}}\right)^2}} - \frac{c}{\overline{g}}\left(\frac{1-\gamma}{\underline{v}}\right)^2
\end{equation*}
for $\underline{v} \in (1-\gamma, (1-\gamma)/\sqrt{\underline{g}})$.

Finally, we turn our attention to the case where $(1-\gamma)\mathbb{E}_{x\sim g}\sqrt{r(x)} + c > (1-\gamma)\sup_{x\in \mathcal{X}}\sqrt{r(x)}$. From (\ref{nopureaiconditionuniform}), this condition is equivalent to $\underline{g} \geq \underline{g}^*$. We know from Proposition \ref{pureailemma} that if the platform does not provide any compensation: $w = 0$, then the only equilibrium is $q = g$, where all creators strictly use GenAI. To encourage original content creation, the platform can choose $\underline{v} \in (1-\gamma, (1-\gamma)\sup_{x\in \mathcal{X}}\sqrt{r(x)}) = (1-\gamma, (1-\gamma)/\sqrt{\underline{g}})$, then we have from Proposition \ref{solutionv0lemma} that we automatically have $\underline{v} < (1-\gamma)/\sqrt{\underline{g}} < \widetilde{V}^{\underline{v}}(g)+c$. It follows that $w := \widetilde{V}^{\underline{v}}(g)+c-\underline{v}$ is bounded below, away from zero, by $\underline{w} := \frac{1-\gamma}{\sqrt{\overline{g}}} - \frac{1-\gamma}{\sqrt{\underline{g}}} + \frac{2c}{\overline{g}}$ for all $\underline{v} \in (1-\gamma, (1-\gamma)/\sqrt{\underline{g}})$, that is, there is a minimum compensation needed for the scheme to take effect.

Figure \ref{fig:optimal_compensation_uniform} shows a plot of the platform's profit under an example set of parameters. When $\underline{g} = 0.2 < \underline{g}^*$, we can see that it is optimal for the platform to set $w^*\approx 0.09$. When $\underline{g} = 0.3 > \underline{g}^*$, under no compensation, the equilibrium is $q = g$, and a minimum amount of compensation $\underline{w} = \frac{1-\gamma}{\sqrt{\overline{g}}} - \frac{1-\gamma}{\sqrt{\underline{g}}} + \frac{2c}{\overline{g}} \approx 0.017$ is needed for the compensation scheme to be effective. However, the platform can benefit from the compensation scheme, as we can see that it is optimal to set $w^*\approx 0.095$. Lastly, when $\underline{g} = 0.4$, the GenAI model is relatively well--trained and a large compensation $\underline{w} \approx 0.068$ is needed to incentivize creators to manually compete. In this case, it is best for the platform to pay no compensation and let all creators strictly use GenAI.
 
\begin{figure}[ht!]
    \centering
    \includegraphics[width=0.8\linewidth]{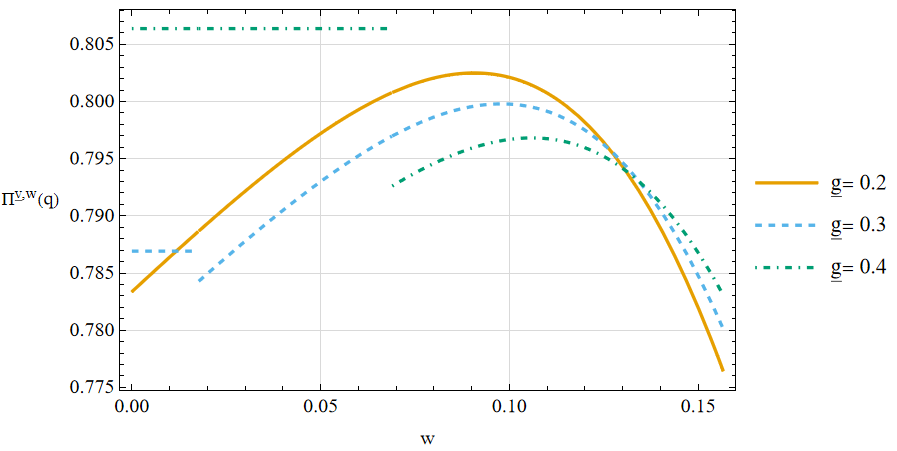}
    \caption{A plot of platform's profit $\Pi^{\underline{v}, w}(q)$ with $c = 1-\gamma = 0.15$, under the compensation scheme $(\underline{v}, w)$, as a function of $w$ when $\underline{g} = 0.2$, $0.3$, and $0.4$. When $\underline{g} \geq \underline{g}^*$, the equilibrium is $q = g$ and the compensation scheme only effective when $w > \underline{w} \geq 0$. We follow Proposition \ref{pureailemma} and assume that the platform implement the dominant strategy of no compensation when $w\leq \underline{w}$, hence a negative jump when $w\geq \underline{w}$ as the platform starts paying the compensation.}
    \label{fig:optimal_compensation_uniform}
\end{figure}

\subsection{Mixed Gaussian Simulation}\label{mixedgaussiansim}

Consider the distribution of consumer preferences given by: $p := 0.4\mathcal{N}(-2,0.5) + 0.6\mathcal{N}(+2,0.5)$. This is a distribution on $\mathbb{R}$, but we shall consider $\mathcal{X} := [-4,+4]$, clipping all data points to be in this compact range. We will use $\gamma = 0.9$, $c = 1-\gamma$, $\alpha=1/2$, and $N=26500$ for sampling purposes, throughout this example. For the purpose of modeling the platform's recommender system in our simulation, we partition $\mathcal{X}$ into $100$ bins of equal size. A consumer $x \in \mathcal{X}$ will only be presented with content $y \in \mathcal{X}$ in the same bin, and we assume that the revenue produced from a given bin follows the Cobb--Douglas form: $\sqrt{(\#\text{consumers in the bin})\cdot (\#\text{contents in the bin})}$. 

We use the score--based SDE diffusion architecture (\cite{song2020score}) for our GenAI model $g$ which will be available to all creators. First, we consider when $g$ is not well--trained. We use a small training dataset of $n = 2650$ i.i.d. drawn points from $p$ to train $g$ over $8$ epochs. The simulation proceeds under the platform's compensation scheme $(\underline{v}, w)$ as follows. We draw $N$ i.i.d. consumers from $p$ and $N$ i.i.d. data points from $p$ as the existing contents. A creator observes the consumers and existing contents in the bin she belongs to and forms a belief on the compensated revenue from her manual creation. We generate $N$ contents from GenAI, compute the compensated revenue from each generated content, then let the average be the current expected revenue from using GenAI for the creators. We draw $N$ i.i.d. creators from $p$. Each creator makes a utility maximization creation decision between manual creation and using GenAI, based on the drawn consumers, the existing contents, and the current expected revenue from GenAI. We update the existing contents and repeat the process, after several rounds, the content distribution converges to the equilibrium $q$.

We show the market outcome under a few different choices of $(\underline{v}, w)$ compensation scheme in Figure \ref{fig:oneshotexperiment}. Without compensation, the contents on the entire platform appears to be GenAI generated, as creators choose to avoid the manual production cost. Given that the GenAI is not well--trained, we can see a clear distortion between the consumer preference and the available contents on the platform. By implementing $(\underline{v},w) = (0.110, 0.150)$, the platform incentivizes many creators to create original contents, especially those with $|x| > 2$ where the consumers along the Gaussian tails are under--served by GenAI. The platform can further expand the compensation scheme by lowering the revenue--threshold. This improves revenue, but at a higher cost. We can see in Table \ref{tab:oneshotrevenue} that $(0.110, 0.150)$ yields a higher profit compared to no compensation, but the profit decreases when the platform further lowers the threshold to $(0.105, 0.150)$.

\begin{figure}[ht!]
    \centering
    \includegraphics[width=0.33\linewidth]{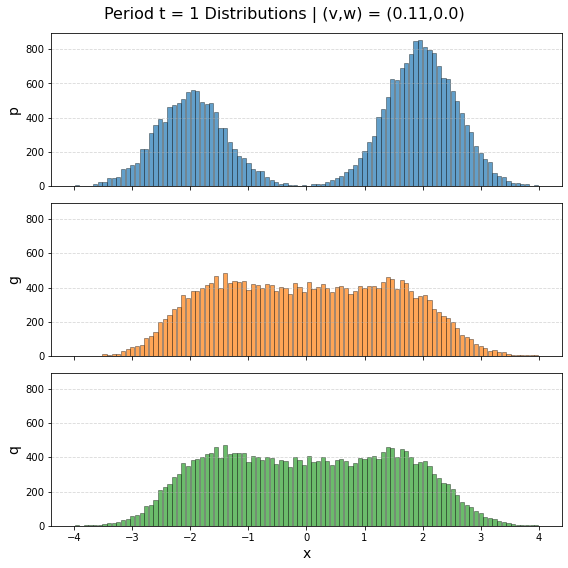}\includegraphics[width=0.33\linewidth]{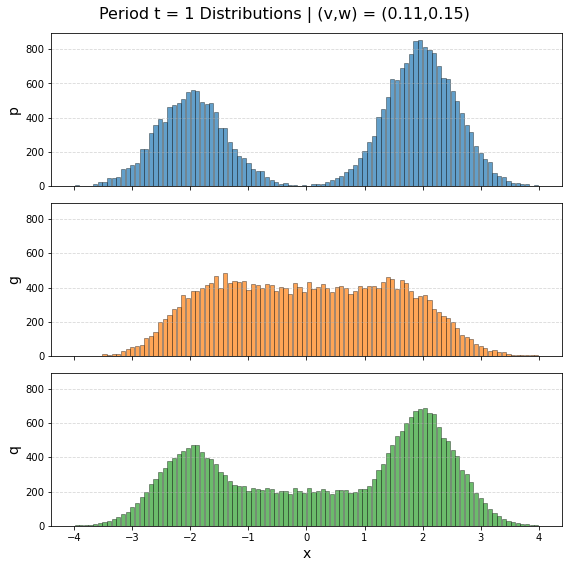}\includegraphics[width=0.33\linewidth]{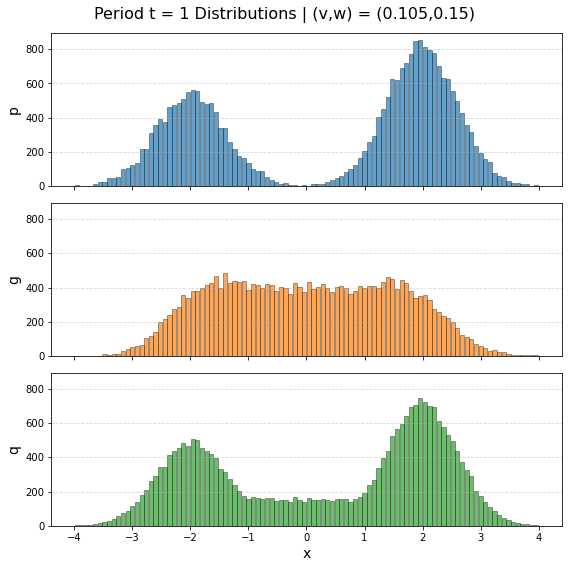}
    \caption{Simulation outcome under the compensation scheme $(\underline{v}, w) = (0.110,0)$ (no compensation), $(0.110,0.150)$, and $(0.105,0.150)$. The $N$ consumers is shown in top--row, and the $N$ contents on the platform is shown in bottom--row. The middle--row shows $N$ contents generated from the GenAI model we trained.}
    \label{fig:oneshotexperiment}
\end{figure}

\begin{table}[ht!]
    \centering
    \begin{tabular}{@{} l *{3}{c} @{}}
        \toprule
        $(\underline{v}, w)$ & $(0.110, 0.0)$ & $(0.110, 0.150)$ & $(0.105, 0.150)$ \\
        \midrule
        $R^{\underline{v}, w}(q)$ & 0.781 & 0.862 & 0.877 \\
        $\Pi^{\underline{v}, w}(q)$ & 0.781 & 0.806 & 0.789 \\
        \bottomrule
    \end{tabular}
    \caption{Simulated platform's revenue and profit under $(\underline{v}, w)$.}
    \label{tab:oneshotrevenue}
\end{table}

Although we have confined ourselves to a single--period game in our analysis, an extension to multi--period with short--lived consumers and creators is straight--forward, and can easily be explored via simulation. Therefore, let us continue with the setting in this example, with the game repeated for $t=1,\cdots,T:=10$ periods, where the $t$ period version of GenAI is trained on the $t-1$ period contents on the platform ($N$ training data points). We also increase the number of training epochs to $50$. The equilibrium outcome in each period is simulated as previously described. We compare the market outcome under various choices of the platform's compensation scheme $(\underline{v}, w)$. For simplicity, we restrict our attention to only static compensation strategies, where the platform applies a fixed $(\underline{v}, w)$ to all periods, leaving dynamic strategies to future research. We consider this simulation to be a game--theoretic extension of the \emph{model collapse} results in \cite{shumailov2023curse}. 

From Figure \ref{fig:multiperiodsimulation}, when the platform offers no compensation $(\underline{v}, w) = (0.110,0.0)$, we observe a classic case of model collapse. The initial GenAI model at $t = 1$ is well--trained on the human--generated content, but from the small initial distortion in distribution, the later version of GenAI collapses into a distribution centered around $x=0$ by $t=10$. As a result, without compensation scheme, the platform's contents are mostly concentrated around $x=0$ by $t=10$, despite most consumers prefer contents around $x=\pm 2$. We experimented with different levels of compensation, and we show that with the lower threshold $\underline{v}$ the platform can regulate the level of GenAI usage and retains the qualitative properties of both GenAI model and the contents distribution even after $t=10$ periods. From Table \ref{tab:multiperiodrevenue}, the compensation scheme $(\underline{v}, w) = (0.105, 0,150)$ appears to offer the highest total profit as well as the highest final period $t = 10$ profit. 

\begin{figure}[ht!]
    \centering
    \includegraphics[width=0.32\linewidth]{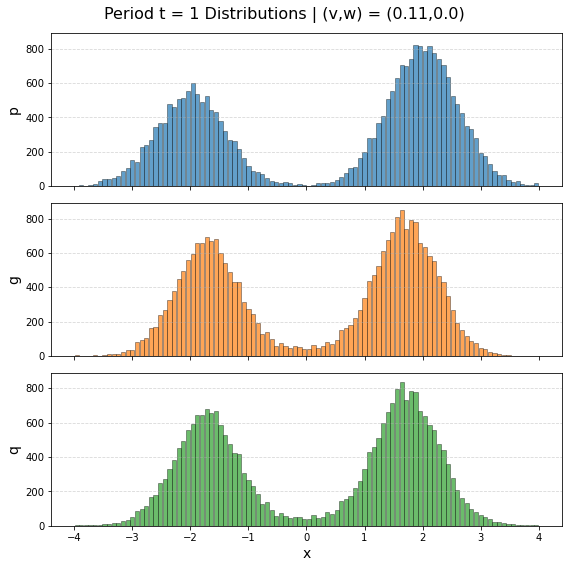}    \includegraphics[width=0.32\linewidth]{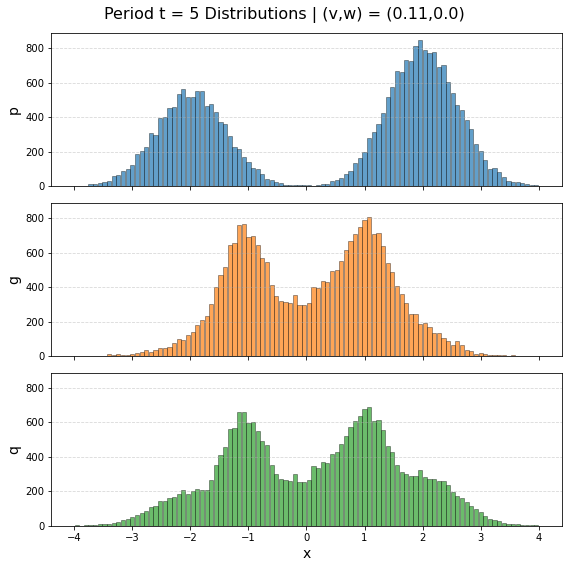}    \includegraphics[width=0.32\linewidth]{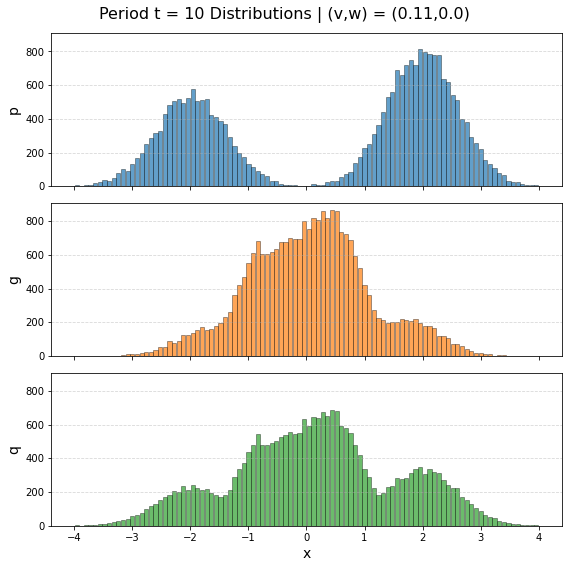}\\
    \includegraphics[width=0.32\linewidth]{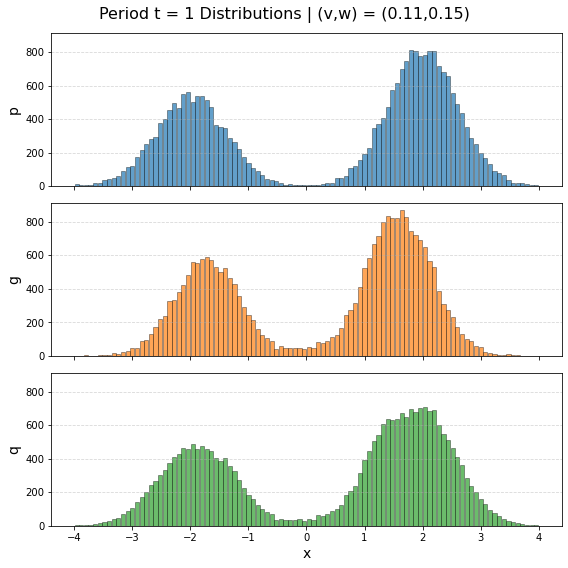}    \includegraphics[width=0.32\linewidth]{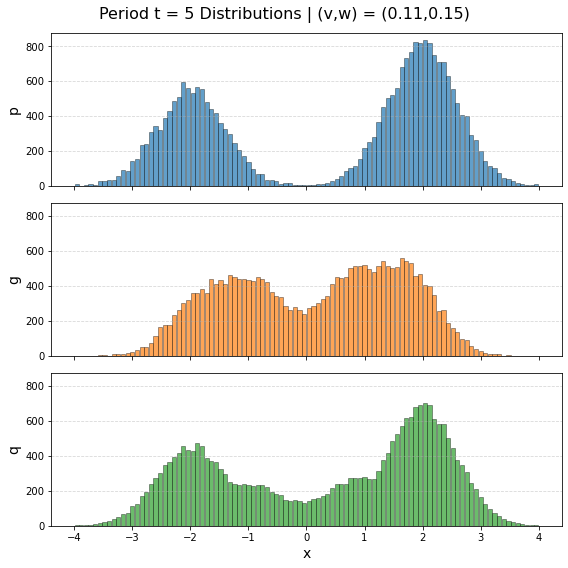}    \includegraphics[width=0.32\linewidth]{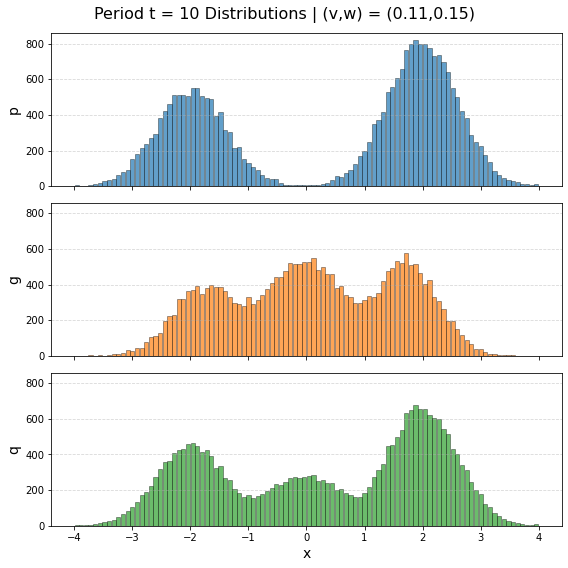}\\
    \includegraphics[width=0.32\linewidth]{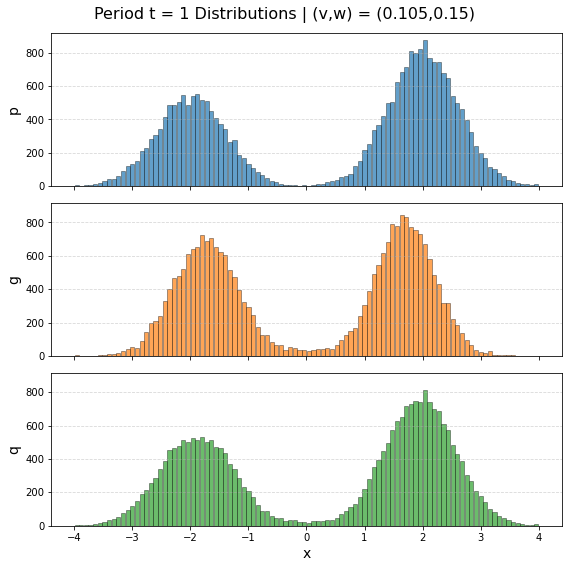}    \includegraphics[width=0.32\linewidth]{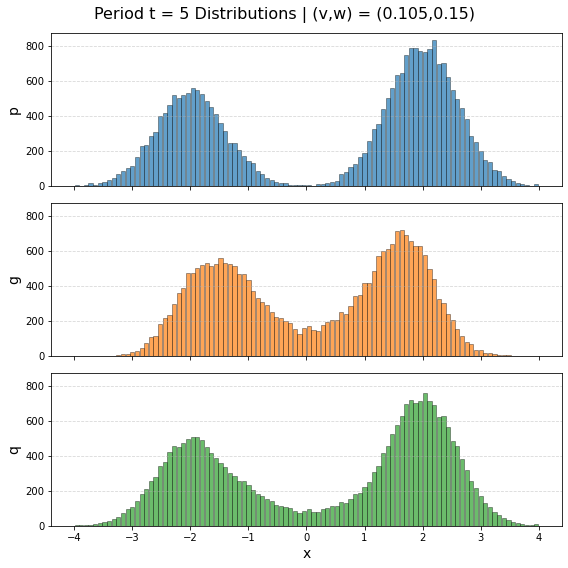}    \includegraphics[width=0.32\linewidth]{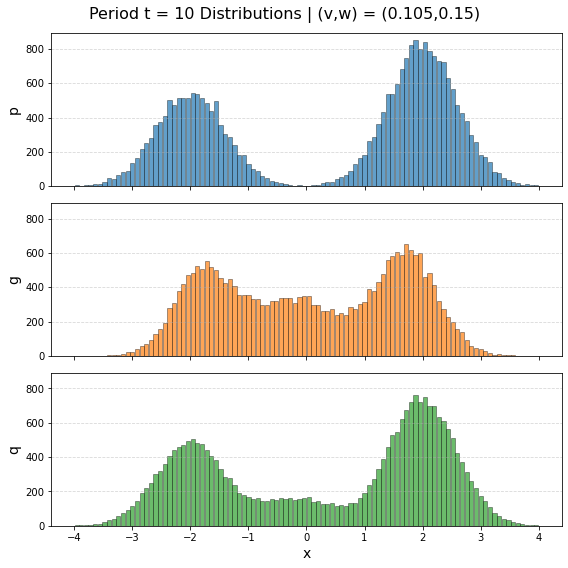}\\
    \includegraphics[width=0.32\linewidth]{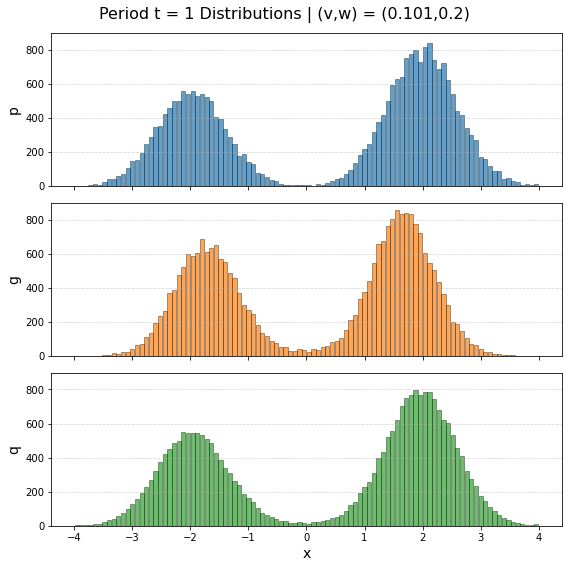}    \includegraphics[width=0.32\linewidth]{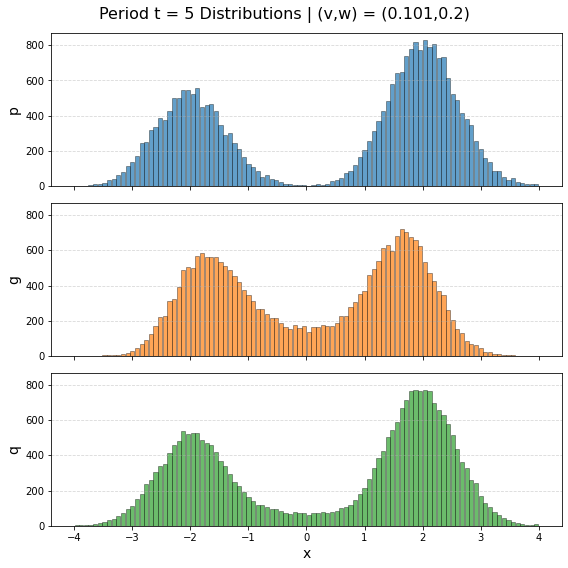}    \includegraphics[width=0.32\linewidth]{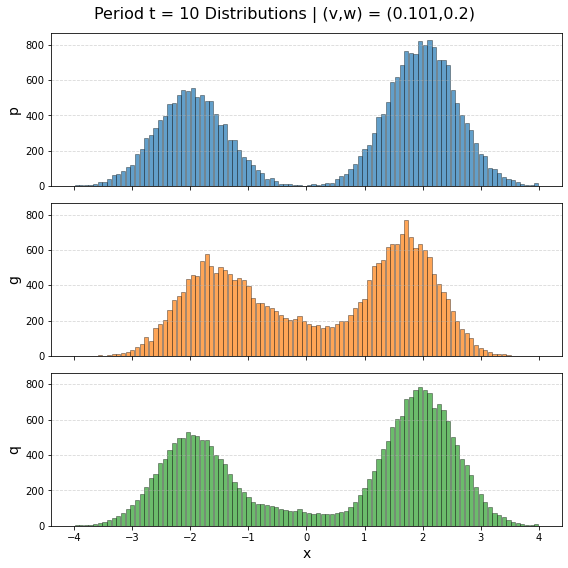}\\
    \caption{Multi--periods simulation outcomes under the compensation scheme $(\underline{v}, w) = (0.110,0.0)$ (no compensation, 1st row), $(0.110, 0.150)$ (2nd row), $(0.105, 0.150)$ (3rd row), and $(0.101, 0.200)$ (4th row), at $t = 1$ (1st column), $t=5$ (2nd column), and $t=10$ (3rd column).}
    \label{fig:multiperiodsimulation}
\end{figure}

\begin{table}[ht!]
    \centering
    \begin{tabular}{@{} l *{4}{c} @{}}
        \toprule
        $(\underline{v}, w)$ & $(0.110, 0.0)$ & $(0.110, 0.150)$ & $(0.105, 0.150)$ & $(0.101, 0.200)$ \\
        \midrule
        $R^{\underline{v}, w}_{t=T}(q)$ & 0.705 & 0.854 & 0.876 & 0.891\\
        $\Pi^{\underline{v}, w}_{t=T}(q)$ & 0.705 & 0.788 & 0.801 & 0.762\\
        $\sum_{t=1}^T\Pi^{\underline{v}, w}_{t}(q)$ & 7.602 & 8.163 & 8.191 & 7.759 \\
        \bottomrule
    \end{tabular}
    \caption{Simulated platform's revenue and profit under $(\underline{v}, w)$ at the final period $t = T$, and the total profit over $t=1,\cdots,T=10$ periods.}
    \label{tab:multiperiodrevenue}
\end{table}

Our two simulations highlights two--folds benefits of compensation scheme. First, the short--term benefit: to reduce the content distribution distortion by encouraging more human--generated contents in the given period. This is especially relevant when the GenAI model is poorly trained. Second, the long--term benefit: to reduce the data pollution, sustaining the GenAI model from collapsing. This is relevant even for a well-trained GenAI model.


\section{Conclusion}

This this paper, we argued that a simple economically--driven compensation scheme based on revenue--threshold can incentivizes more creation of high--value human--generated contents, reduces the data pollution, and improves the platform's profit. We showed that even with access to more information, or an AI--detector, the platform's optimal choice of a compensation scheme is not too different from the revenue--threshold scheme which does not require any expensive computation. We make some further remarks and comments on potential future directions as follows. 

In our model, the introduction of GenAI appears to universally improve the creators welfare. This can be understood since we have assumed that all creators are equally able to access GenAI. Therefore, the new technology such as GenAI does not harm those who can adopt, but the problem is when there is an inequality of adoption, see \cite{kim2026does} for a similar conclusion. For example, let us assume that $1-\gamma > c$ and that there are some (zero measure amount of) creators $x \in \mathcal{X}$ who cannot adopt GenAI. At the pre--GenAI equilibrium, we have $q = p$ and each creator revenue would be $1-\gamma-c > 0$. After GenAI is introduced, any non--adopter $x$ with $r(x) = p(x)/g(x) < \int_{\mathcal{X}}\beta(y)p(y)dy$ (i.e. $V(x;q) = (1-\gamma)(p(x)/q(x))^\alpha < 1-\gamma$) would face a high competition from GenAI contents and obtains a lower revenue than $1-\gamma-c$. Such a non--adopter $x$ is likely a mainstream creator who have contributed to the GenAI training dataset in the previous period, leading to a high $g(x)$. Our result shows that a compensation scheme would have lowered $\int_{\mathcal{X}}\beta(y)p(y)dy$, and improves the welfare of a traditional creator such as $x$ as the platform becomes less flooded with repetitive contents in the style of $x$. An extension of our model to account for manual--only and GenAI--only creators would be a natural future direction.

One of the motivations for creator compensation is to acknowledge their contribution to the future GenAI training data. Although we focused on a single--period game in this work, apart from what we touched briefly in \S\ref{mixedgaussiansim}, we may interpret any additional compensated human--generated content is also a compensated contribution to the future GenAI training data. A promising research direction is to analyze the multi--period version of our model, characterize the optimal dynamic platform's compensation strategy, and quantify the impacts of the compensated human--generated contents on preventing the model collapse.


\printbibliography

\appendix
\begin{center}
{\Large Appendix}
\end{center}
\section{Omitted Proofs}

\subsection{Proof of Lemma \ref{pregenailemma}}
\begin{proof}
    Suppose that $(\beta, q)$ is an equilibrium under $W$. From (\ref{generaldecomposedformwitho}) without GenAI, $q$ must take the form $q(x) = \beta_{\text{H}}(x)p(x)$ for some $\beta_{\text{H}}(x) \in [0,1]$. The creator $x$ profit for creating a content $x$ is $U^W(x;q) = (1-\gamma)/\beta_{\text{H}}(x)^\alpha + W(x;q)-c$. If $1-\gamma > c - W(x;q)$ then $U^W(x;q) > 0$, and therefore, it must be the case that $\beta_{\text{H}}(x) = 1$. Otherwise, we have an indifference condition $U^W(x;q) = 0$, which implies $\beta_{\text{H}}(x) = \left(\frac{1-\gamma}{c-W(x;q)}\right)^{1/\alpha}$. Therefore, the given $(\beta, q)$ is the unique equilibrium under the compensation scheme $W$ as claimed. If $1-\gamma\geq c$ then we already have $q = p$, and the platform can set $W = 0$ to achieve the best possible profit: $\Pi^W(q) = \gamma$. Now, suppose that $1-\gamma < c$, we can also restrict our attention to $W$ such that $W(x;q) < c$. Substituting the equilibrium content distribution $q(x) = \left(\frac{1-\gamma}{c-W(x;q
    )}\right)^{1/\alpha}p(x)$ into (\ref{platformrevenue}), we find the platform equilibrium profit to be:
    \begin{equation}\label{pregenaiplatformprofit}
        \Pi^W(q) = \int_{\mathcal{X}}\left(\frac{c\gamma-W(x;q)}{1-\gamma}\right)\left(\frac{1-\gamma}{c-W(x;q)}\right)^{1/\alpha}p(x)dx.
    \end{equation}
    The platform can maximize profit by choosing $W^*(x;q)$ which maximizes the integrand at each $x$. Evidently, we have $W^*(x;q) \in [0,c\gamma]$, and since the integrand is independent of $x$, apart from the factor of $p(x)$, we have that $W^*(x;q) = W^*$ is independent of $x$ and $q$. The sign of the derivative of the first two factors in the integrand is determined by the sign of $c(\gamma-\alpha) - (1-\alpha)W(x;q)$. We conclude that $W^* = \max\left\{0,c\frac{\gamma-\alpha}{1-\alpha}\right\}\leq c\gamma$ and the corresponding optimal platform's profit follows.
\end{proof}

\subsection{Proof of Lemma \ref{xbasedexistencelemma}}
\begin{proof}
    Let us fix a small $\delta > 0$, then we consider the set $\mathcal{S}_\delta := L^2(\mathcal{X},[0,1])\times L^2(\mathcal{X}, [0,\overline{V}_\delta]) \subset L^2(\mathcal{X}, \mathbb{R})^2$, where $\overline{V}_\delta := \frac{1-\gamma}{\delta^\alpha}\sup_{x\in \mathcal{X}}r(x)^\alpha$, and a map $\Phi_\delta : \mathcal{S}_\delta \rightarrow \mathcal{S}_\delta$ given by $\Phi_\delta(\beta, V) = (\tilde{\beta}, \widetilde{V})$ where we have for each $x \in \mathcal{X}$:
    \begin{equation}\label{phimapbetavdefinition}
        \begin{aligned}
            \tilde{\beta}(x) &:= \clip\left(1 + \frac{\int_{\mathcal{X}}\beta(z)p(z)dz}{r(x)}-\left(\frac{1-\gamma}{\max\{V^W(g)+c-W(x),\delta\}}\right)^{1/\alpha},\delta,1\right)\\
            \widetilde{V}(x) &:= (1-\gamma)\left(1-\tilde{\beta}(x) + \frac{\int_{\mathcal{X}}\tilde{\beta}(z)p(z)dz}{r(x)}\right)^{-\alpha}
        \end{aligned},
    \end{equation}
    $V^W(g) := \int_{\mathcal{X}}(V(z) + W(z))g(z)dz$, and $\clip(v,\underline{v},\overline{v}) := \max\{\underline{v}, \min\{\overline{v}, v\}\}$. Endow $\mathcal{S}_\delta$ with the product of weak topologies of $L^2(\mathcal{X};[0,1])$ and $L^2(\mathcal{X};[0, \overline{V}_\delta])$. By Banach–Alaoglu theorem, $\mathcal{S}_\delta$ is a compact convex subset of $L^2(\mathcal{X};\mathbb{R})^2$ which is a locally convex space under the weak topology. We show that $\Phi_\delta$ is continuous as follows. Consider $\{(\beta_i, V_i)\}_{i=1}^\infty$ such that $(\beta_i, V_i) \rightharpoonup (\beta, V) \in \mathcal{S}_\delta$, then we have $\int_{\mathcal{X}}\beta_i(z)p(z)dz \rightarrow \int_{\mathcal{X}}\beta(z)p(z)dz$ and $\int_{\mathcal{X}}V_i(z)g(z)dz \rightarrow \int_{\mathcal{X}}V(z)g(z)dz$ by definition. Let $(\tilde{\beta}_i, \widetilde{V}_i) := \Phi_\delta(\beta_i, V_i)$ and $(\tilde{\beta}, \widetilde{V}) := \Phi_\delta(\beta, V)$, then we can see from the continuity of the RHS of (\ref{phimapbetavdefinition}) in $\int_{\mathcal{X}}\beta(z)p(z)dz$ and $\int_{\mathcal{X}}V(z)g(z)dz$ that we have a point--wise convergence: $\tilde{\beta}_i(x) \rightarrow \tilde{\beta}(x), \widetilde{V}_i(x)\rightarrow \widetilde{V}(x)$, for all $x\in \mathcal{X}$. Since $\tilde{\beta}_i$ and $\widetilde{V}_i$ are uniformly bounded, the point--wise convergence implies strong convergence by dominated convergence theorem, which implies weak convergence: $(\tilde{\beta}_i, \widetilde{V}_i)\rightharpoonup (\tilde{\beta}, \widetilde{V})$, proving that $\Phi_\delta$ is continuous with respect to the weak topology. From Schauder-–Tychonoff fixed--point theorem, there exists a fixed--point $(\beta^*_\delta, V^*_\delta) \in \mathcal{S}_\delta$ of $\Phi_\delta$.

    Recall that $\delta > 0$ is arbitrary, let us consider the sequence of fixed--points $\{(\beta^*_\delta, V^*_\delta)\}_{\delta > 0} \subset L^2(\mathcal{X},[0,1])\times L^2(\mathcal{X}, \mathbb{R}_{\geq 0})$ and the limit $\delta \searrow 0$ in the following. Suppose that $\int_{\mathcal{X}}\beta^*_\delta(z)p(z)dz \searrow 0$, or that we are able to refine to such a subsequence. By further refining a subsequence if necessary, we have $\beta^*_\delta(x)\searrow 0$ for almost every $x\in \mathcal{X}$ as $\delta \searrow 0$. Consequently, we have from (\ref{phimapbetavdefinition}) that $V^*_\delta(x)\rightarrow 1-\gamma$ for almost all $x \in \mathcal{X}$, hence $V^{*W}_\delta(g) + c - W(x) \rightarrow 1-\gamma + c + \mathbb{E}_{y \sim g}W(y)-W(x)$, where $V^{*W}_\delta(g) := \int_{\mathcal{X}}(V^*_\delta(z)+W(z))g(z)dz$. But $W(x) \leq \mathbb{E}_{y \sim g}W(y)$ on some subset $\Delta \subset \mathcal{X}$ of positive measure, so $\frac{1-\gamma}{V^{*W}_\delta(g) + c - W(x)} < \frac{1-\gamma}{1-\gamma+c} < 1$, and hence $\beta^*_\delta(x) \geq 1 - \frac{1-\gamma}{V^{*W}_\delta(g) + c - W(x)} > \frac{c}{1-\gamma+c} > 0$ over $\Delta$ by (\ref{phimapbetavdefinition}) for all sufficiently small $\delta > 0$, contradicting $\int_{\mathcal{X}}\beta^*_\delta(z)p(z)dz\searrow 0$. We conclude that $\int_{\mathcal{X}}\beta^*_\delta(z)p(z)dz$ is bounded away from $0$ for all sufficiently small $\delta > 0$. Consequently, $V^*_\delta$ is uniformly bounded by some $\overline{V} > 0$ for all sufficiently small $\delta > 0$. Hence, $\{(\beta^*_\delta, V^*_\delta)\}_{\delta > 0}$ is a sequence of fixed--points of $\Phi_\delta$ in a weakly compact set $\mathcal{S} := L^2(\mathcal{X};[0,1])\times L^2(\mathcal{X};[0,\overline{V}])\subset L^2(\mathcal{X};\mathbb{R})^2$. 
    
    Note that the definition (\ref{phimapbetavdefinition}) also works at $\delta=0$ to define $\Phi_0 := \Phi_{\delta = 0}$ at any $(\beta, V)$ such that $\widetilde{V}(x)$ is well--defined for all $x$, where we set $\tilde{\beta}(x) = 0$ if $W(x) \geq V^W(g) + c$. By extracting a weakly convergence subsequence if needed, let us assume that $(\beta^*_\delta, V^*_\delta) \rightharpoonup (\beta^*, V^*) \in \mathcal{S}$. We claim that $(\tilde{\beta}^*_\delta, \widetilde{V}^*_\delta) := \Phi_{\delta}(\beta^*_\delta, V^*_\delta) \rightharpoonup (\tilde{\beta}^*, \widetilde{V}^*) := \Phi_0(\beta^*, V^*)$ as $\delta \searrow 0$. Note that $(\tilde{\beta}^*_\delta, \widetilde{V}^*_\delta) = \Phi_\delta(\beta^*_\delta, V^*_\delta) = (\beta^*_\delta, V^*_\delta)$ for all $\delta > 0$ by the fixed--point property, but we denote them differently to conceptually distinguishes the LHS of (\ref{phimapbetavdefinition}) from the RHS. To see the convergence, we note that: $\int_{\mathcal{X}}\beta^*_\delta(z)p(z)dz \rightarrow \int_{\mathcal{X}}\beta^*(z)p(z)dz$ and $\int_{\mathcal{X}}V^*_\delta(z)g(z)dz \rightarrow \int_{\mathcal{X}}V^*(z)g(z)dz$, and that $\tilde{\beta}^*_\delta(x) = \delta$ for all sufficiently small $\delta > 0$ if $W(x) \geq \int_{\mathcal{X}}(V^*_\delta(z) + W(z))g(z)dz + c$. This establishes the point--wise convergence: $\tilde{\beta}^*_\delta(x) \rightarrow \tilde{\beta}^*(x)$ and $\widetilde{V}^*_\delta(x)\rightarrow \widetilde{V}^*(x)$, for all $x \in \mathcal{X}$, which implies $\tilde{\beta}^*(x) \in [0,1], \widetilde{V}^*(x) \in [0,\overline{V}]$ because $\tilde{\beta}^*_\delta(x) = \beta^*_\delta(x) \in [0,1], \widetilde{V}^*_\delta(x) = V^*_\delta(x) \in [0,\overline{V}]$ for all $x\in \mathcal{X}$ and $\delta > 0$. Hence $(\tilde{\beta}^*_\delta, \widetilde{V}^*_\delta)\rightharpoonup (\tilde{\beta}^*, \widetilde{V}^*)$ as claimed. It follows that $(\beta^*_\delta, V^*_\delta) \rightharpoonup \Phi_0(\beta^*, V^*) \in \mathcal{S}$, hence $(\beta^*, V^*) = \Phi_0(\beta^*, V^*)$ by the uniqueness of the limit. Finally, we recognize from (\ref{phimapbetavdefinition}) with $\delta = 0$ that if $(\beta^*, V^*) = \Phi_0(\beta^*, V^*)$ then $(\beta^*, q^*)$ is an equilibrium under $W$, where $q^*(x) = (1-\beta^*(x))p(x) + g(x)\int_{\mathcal{X}}\beta^*(z)p(z)dz$. 
\end{proof}

\subsection{Proof of Lemma \ref{rationalexpectationlemma}}
\begin{proof}
    The proof follows by checking the definition. Note that $V^W(x;q) = V(x;q) + \widetilde{W}(x;\tilde{q})$, therefore we have at $q = \tilde{q}$ that: $V^W(x;\tilde{q}) = V^{\widetilde{W}}(x;\tilde{q})$ and $V^W(g;\tilde{q}) = V^{\widetilde{W}}(g;\tilde{q})$. It follows that $\pi^W(x,\beta;\tilde{q}) = \pi^{\widetilde{W}}(x,\beta; \tilde{q})$ which means $\tilde{\beta}(x) \in \argmax_{\beta \in \Delta \{\text{H}, \text{AI}, \text{O}\}}\pi^{\widetilde{W}}(x, \beta; \tilde{q}) = \argmax_{\beta \in \Delta \{\text{H}, \text{AI}, \text{O}\}}\pi^W(x, \beta; \tilde{q})$, and lastly we have $\tilde{q}(x) = \tilde{\beta}_{\text{H}}(x)p(x) + g(x)\int_{\mathcal{X}}\tilde{\beta}_{\text{AI}}(z)p(z)dz$ by definition, proving that $(\tilde{\beta}, \tilde{q})$ is an equilibrium under $W$.
\end{proof}

\subsection{Proof of Lemma \ref{decisionundercompensationlemma}}
\begin{proof}
    Suppose that $\underline{v} \leq V^{\underline{v}, w}(g;q) + c\leq \underline{v}+w$. Then for any $x\in \mathcal{X}$ such that $V(x;q) \geq \underline{v}$, we have $U^{\underline{v}, w}(x;q) = V(x;q) + w - c \geq \underline{v} + w - c\geq V^{\underline{v}, w}(g;q)$. Note that the inequality $U^{\underline{v}, w}(x;q) \geq V^{\underline{v}, w}(g;q)$ is strict precisely if $V(x;q) > \underline{v}$ or $\underline{v}+w > V^{\underline{v}, w}(g;q)+c$, in other words, the creator $x$ strictly prefers manual creation if $V(x;q) > \underline{v}$ or $\underline{v}+w > V^{\underline{v}, w}(g;q)+c$. Therefore, if $V(x;q) = \underline{v}$, and $\underline{v}+w = V^{\underline{v}, w}(g;q)+c$ then $U^{\underline{v}, w}(x;q) = V^{\underline{v}, w}(g;q)$, so $x$ is indifferent between manual creation and the use of GenAI. If $V(x;q) < \underline{v}$ then $U^{\underline{v}, w}(x;q) = V(x;q)-c < \underline{v} - c \leq V^{\underline{v},w}(g;q)$ then $x$ strictly prefers to use GenAI, as claimed.

    Suppose that $\underline{v} + w < V^{\underline{v}, w}(g;q) + c$ then $\tilde{\underline{v}} := V^{\underline{v}, w}(g;q) + c - w > \underline{v}$ is the indifference threshold under $(\underline{v}, w)$ for the given common belief that the content distribution density is $q$. For all $\tilde{w} \in [0,w]$ we have: $V^{\tilde{\underline{v}}, \tilde{w}}(g;q) + c \geq \tilde{\underline{v}}$ because $V^{\underline{v}, w}(g;q) - w = V(g;q) - w\int_{\mathcal{X}}\mathbbm{1}[V(y;q)<\underline{v}]\cdot g(y)dy \leq V(g;q) \leq V^{\tilde{\underline{v}}, \tilde{w}}(g;q)$. Additionally, we have $\tilde{\underline{v}} + w = V^{\underline{v}, w}(g;q)+c \geq V^{\tilde{\underline{v}}, w}(g;q) + c$ because we stop paying compensation $w$ to any contents $y \in \mathcal{X}$ with $V(y;q) \in [\underline{v}, \tilde{\underline{v}})$. Since $V^{\tilde{\underline{v}}, \tilde{w}}(g;q)+c$ is linear in $\tilde{w} \in [0,w]$ with slope $\int_{\mathcal{X}}\mathbbm{1}[V(y;q)\geq \tilde{\underline{v}}]\cdot g(y)dy \leq 1$ which is less than the slope of $\tilde{\underline{v}}+\tilde{w}$, while $\tilde{\underline{v}}+w\geq V^{\tilde{\underline{v}}, w}(g;q)$ and $\tilde{\underline{v}} \leq V^{\tilde{\underline{v}}, \tilde{w}=0}(g;q) = V(g;q)$, therefore we can lower $\tilde{w} \in [0,w]$ until the two quantity intercepts, at which point we have: $\tilde{\underline{v}} + \tilde{w} = V^{\tilde{\underline{v}}, \tilde{w}}(g;q)+c \geq \tilde{\underline{v}}$.

    Suppose that $\underline{v} > V^{\underline{v}, w}(g;q) + c$ and consider any $x\in \mathcal{X}$ whose strictly prefers GenAI when there is no compensation: $U(x;q) = V(x;q) - c < V(g;q)$. Then $V(x;q) < \underline{v}$ and so $U^{\underline{v}, w}(x;q) = V^{\underline{v}, w}(x;q) - c = V(x;q)-c < V(g;q) \leq V^{\underline{v}, w}(g;q)$, thus, $x$ also strictly prefers GenAI under $(\underline{v}, w)$. Additionally, if none of $x \in \mathcal{X}$ strictly prefers manual creation, then $V(x;q) \leq V^{\underline{v}, w}(x;q) \leq V^{\underline{v}, w}(g;q)+c < \underline{v}$ for all $x\in \mathcal{X}$. Therefore, $V^{\underline{v}, w}(g;q) = V(g;q)$ and so $\underline{v}_0 = V^{\underline{v}, w}(g;q) + c$ is also the indifference threshold under $(\underline{v}, w)$.
\end{proof}

\subsection{Proof of Proposition \ref{rthresholdschemeproposition}}

\begin{proof}
    To simplify the notation, we define $M(\beta, \mathcal{X}) := \int_{\mathcal{X}}\beta(z)p(z)dz$.

    \emph{Part 1 (The compensation given to any $x,y\in \mathcal{X}$ with $r(x) = r(y)$ should be the same):} 
    
    Let $(\tilde{\beta}, \tilde{q})$ be an equilibrium under the given platform's compensation scheme $\widetilde{W}$. If $\tilde{\beta}(x) < 1$ for any $x \in \mathcal{X}$ then $V^{\widetilde{W}}(x;\tilde{q}) = V(x;\tilde{q}) + \widetilde{W}(x;\tilde{q}) \geq V^{\widetilde{W}}(g;\tilde{q}) + c$. Then it is safe to assume that $V^{\widetilde{W}}(x;\tilde{q}) = V^{\widetilde{W}}(g;\tilde{q}) + c$, otherwise lowering the compensation $\widetilde{W}(x;\tilde{q})$ will not change the creator $x$'s behavior while saving compensation costs for the platform. We denote by $\widetilde{\mathcal{X}}_{\text{IN}} \subset \mathcal{X}$ the set of creators who are indifferent between manual creation and using GenAI under $\widetilde{W}$. Similarly, if $x$ strictly prefers to use GenAI under $\widetilde{W}$, then $V^{\widetilde{W}}(x;\tilde{q}) = V(x;\tilde{q}) + \widetilde{W}(x;\tilde{q}) < V^{\widetilde{W}}(g;\tilde{q}) + c$ and it is safe to assume that $\widetilde{W}(x;\tilde{q}) = 0$, otherwise lowering $\widetilde{W}(x;\tilde{q})$ to zero will not change $\tilde{\beta}$ nor $\tilde{q}$ but will save the compensation costs for the platform. We denote by $\widetilde{\mathcal{X}}_{\text{AI}}\subset \mathcal{X}$ the set of creators who strictly prefers GenAI under $\widetilde{W}$. We obtain the decomposition $\mathcal{X} = \widetilde{\mathcal{X}}_{\text{AI}}\sqcup \widetilde{\mathcal{X}}_{\text{IN}}$, just as in \S\ref{equilibriumanalysissection}. Using the general form (\ref{generaldecomposedformwitho}) of $\tilde{q}$, for any $x \in \widetilde{\mathcal{X}}_{\text{IN}}$ we have
    \begin{equation}\label{qtildeindifferencecondition}
        \frac{\tilde{q}(x)}{p(x)} = 1-\tilde{\beta}(x) + \frac{M(\tilde{\beta}, \mathcal{X})}{r(x)} = \left(\frac{1-\gamma}{V^{\widetilde{W}}(g;\tilde{q}) + c - \widetilde{W}(x;\tilde{q})}\right)^{1/\alpha},
    \end{equation}
    and for any $x\in \mathcal{X}_{\text{AI}}$ we have $\tilde{q}(x)/p(x) = M(\tilde{\beta}, \mathcal{X})/r(x)$.
    Let us choose an arbitrary small $\delta > 0$ and any $r_0 \in [\essinf_{x\in \mathcal{X}}r(x), \esssup_{x\in \mathcal{X}}r(x)] \subset \mathbb{R}_{>0}$, then define $\Delta_0 := r^{-1}[r_0-\delta/2, r_0 + \delta/2]\cap \widetilde{\mathcal{X}}_{\text{IN}} \subset \mathcal{X}$. By the atomless assumption on the distribution of $r(x)$, we have $\int_{\Delta_0}p(x)dx = O(\delta) > 0$. Let us define:
    \begin{equation}\label{phihelperfuncdefinition}
    \phi^{\widetilde{W}}(w;\tilde{q}) := \left(\frac{1-\gamma}{V^{\widetilde{W}}(g;\tilde{q}) + c - w}\right)^{1/\alpha}\left(\frac{\gamma V^{\widetilde{W}}(g;\tilde{q})+\gamma c - w}{1-\gamma}\right),
    \end{equation}
    then from (\ref{platformrevenue}) and (\ref{qtildeindifferencecondition}) we can express the platform's profit as follows:
    \begin{multline*}\label{profitbeforeaverageperturbation}
        \Pi^{\widetilde{W}}(\tilde{q}) = \gamma\int_{\mathcal{X}}\left(\frac{\tilde{q}(x)}{p(x)}\right)^{1-\alpha}p(x)dx - \int_{\mathcal{X}}\widetilde{W}(x;\tilde{q})\left(\frac{\tilde{q}(x)}{p(x)}\right)p(x)dx\\
        = \gamma M(\tilde{\beta},\mathcal{X})^{1-\alpha}\int_{\mathcal{X}_{\text{AI}}}\frac{p(x)}{r(x)^{1-\alpha}}dx + \int_{\Delta_0}\phi^{\widetilde{W}}(\widetilde{W}(x;\tilde{q});\tilde{q})p(x)dx\\
        + \gamma\int_{\mathcal{X}_{\text{IN}}\setminus\Delta_0}\left(\frac{\tilde{q}(x)}{p(x)}\right)^{1-\alpha}p(x)dx - \gamma\int_{\mathcal{X}_{\text{IN}}\setminus\Delta_0}\widetilde{W}(x;\tilde{q})\left(\frac{\tilde{q}(x)}{p(x)}\right)p(x)dx.\numberthis
    \end{multline*}
    If $\gamma \leq \frac{2\alpha}{\alpha+1}$ then it can be shown that $w\mapsto \phi^{\widetilde{W}}(w;\tilde{q})$ is concave for all $w \in [0,V^{\widetilde{W}}(x;\tilde{q}) + c)$, thus, we have by Jensen's inequality that:
    \begin{multline*}\label{jensenineqonphi}
        \int_{\Delta_0}\phi^{\widetilde{W}}(\widetilde{W}(x;\tilde{q});\tilde{q})p(x)dx = \mathbb{E}_{x\sim p}\left[\phi^{\widetilde{W}}(\widetilde{W}(x;\tilde{q});\tilde{q})|\Delta_0\right]\int_{\Delta_0}p(x)dx\\
        \leq \phi^{\widetilde{W}}\left(\mathbb{E}_{x\sim p}\left[\widetilde{W}(x;\tilde{q})|\Delta_0\right];\tilde{q}\right)\int_{\Delta_0}p(x)dx.\numberthis
    \end{multline*}
    Consider a new compensation scheme $W : \mathcal{X}\times \mathcal{D}(\mathcal{X}) \rightarrow \mathbb{R}_{\geq 0}$ given by $W(x;q) = W(x) := (\widetilde{W}(x;\tilde{q}) + O(\delta))\cdot \mathbbm{1}[x\in \mathcal{X}\setminus \Delta_0] + \mathbb{E}_{x\sim p}\left[\widetilde{W}(x;\tilde{q})|\Delta_0\right]\cdot \mathbbm{1}[x\in \Delta_0]$, where the $\delta$--order part will soon be discussed. Note that $W(x;q)$ is independent of $q$, it depends only on $x$ (and the given $\tilde{q}$) over $\mathcal{X}\setminus \Delta_0$ and it is constant on $\Delta_0$. We will find an equilibrium $(\beta, q)$ under $W$ such that $\beta(x) = \tilde{\beta}(x) + O(\delta)$ for all $x \in \mathcal{X}\setminus \Delta_0$, while we allow for $\beta(x) - \tilde{\beta}(x)$ to have order $O(1)$ over $\Delta_0$ which has a measure $O(\delta)$. Therefore, let us assume for now that $V^W(g;q) = V^{\widetilde{W}}(g;\tilde{q}) + O(\delta)$. This also implies that $\phi^{W}(w;q) = \phi^{\widetilde{W}}(w;\tilde{q}) + O(\delta)$. Let the corresponding decomposition under $W$ be $\mathcal{X}=\mathcal{X}_{\text{AI}}\sqcup \mathcal{X}_{\text{IN}}$. The indifference condition for $x \in \mathcal{X}_{\text{IN}}$ is given by:
    \begin{equation}\label{qindifferencecondition}
        \frac{q(x)}{p(x)} = 1 - \beta(x) + \frac{M(\beta,\mathcal{X})}{r(x)} = \left(\frac{1-\gamma}{V^W(g;q) + c - W(x;q)}\right)^{1/\alpha}.
    \end{equation} 
    Since $\inf_{x\in \Delta_0}\widetilde{W}(x;\tilde{q}) \leq \mathbb{E}_{x\sim p}\left[\widetilde{W}(x;\tilde{q})|\Delta_0\right]\leq \sup_{x\in \Delta_0}\widetilde{W}(x;\tilde{q})$ while $r(x) = r_0 + O(\delta)$ for $x \in \Delta_0$, we have that $x \in \Delta_0$ are indifferent between manual creation and using GenAI under $W$ as well as under $\widetilde{W}$, up to $O(\delta)$. In particular, using $M(\beta,\mathcal{X}) = M(\tilde{\beta},\mathcal{X}) + O(\delta)$, $V^W(g;q) = V^{\widetilde{W}}(g;\tilde{q}) + O(\delta)$, $r(x) = r_0 + O(\delta)$, and $W(x;q)=\mathbb{E}_{x\sim p}\left[\widetilde{W}(x;\tilde{q})|\Delta_0\right]$, we can solve (\ref{qindifferencecondition}) to find $\beta(x) = \beta(\Delta_0) + O(\delta)$ for some constant $\beta(\Delta_0) \in [0,1]$, for all $x \in \Delta_0$. By the atomless assumption, the symmetric difference between $\mathcal{X}_{\text{AI}}$ and $\widetilde{\mathcal{X}}_{\text{AI}}$ has measure $O(\delta)$, while all creators in the symmetric difference set are profit--wise indifferent up to order $O(\delta)$. This means any difference in revenue contribution from the symmetric difference will be of order $O(\delta^2)$, therefore, we will not make any distinction between $\mathcal{X}_{\text{AI}}$ and $\widetilde{\mathcal{X}}_{\text{AI}}$ (or $\mathcal{X}_{\text{IN}}$ and $\widetilde{\mathcal{X}}_{\text{IN}}$) for the remainder of this part of the proof. In particular, we have $\Delta_0\subset \mathcal{X}_{\text{IN}}\cap \widetilde{\mathcal{X}}_{\text{IN}}$. Integrating the second equality in the condition (\ref{qindifferencecondition}) against $p$ over $\Delta_0$ gives:
    \begin{multline*}
        \int_{\mathcal{X}\setminus\Delta_0}g(z)dz\cdot M(\beta,\Delta_0) = \int_{\Delta_0}p(z)dz + M(\tilde{\beta},\mathcal{X}\setminus \Delta_0)\cdot \int_{\Delta_0}g(z)dz\\
        - \left(\frac{1-\gamma}{V^{\widetilde{W}}(g;\tilde{q}) + c - \mathbb{E}_{x\sim p}\left[\widetilde{W}(x;\tilde{q})|\Delta_0\right]}\right)^{1/\alpha}\int_{\Delta_0} p(z)dz + O(\delta^2)\\
        \geq \int_{\Delta_0}p(z)dz + M(\tilde{\beta},\mathcal{X}\setminus \Delta_0)\cdot \int_{\Delta_0}g(z)dz\\
        - \int_{\Delta_0}\left(\frac{1-\gamma}{V^{\widetilde{W}}(g;\tilde{q}) + c - \widetilde{W}(z;\tilde{q})}\right)^{1/\alpha} p(z)dz + O(\delta^2) = \int_{\mathcal{X}\setminus\Delta_0}g(z)dz\cdot M(\tilde{\beta},\Delta_0) + O(\delta^2).
    \end{multline*}
    Where we used $M(\beta,\mathcal{X}) = M(\tilde{\beta},\mathcal{X}) + O(\delta)$, $V^W(g;q) = V^{\widetilde{W}}(g;\tilde{q}) + O(\delta)$ and $O(\delta)\cdot \int_{\Delta_0}p(z)dz = O(\delta^2)$ in the first equality, we used Jensen's inequality and the fact that $w\mapsto 1/(A-w)^{1/\alpha}$ is convex to obtain the inequality, and we recognized the last equality by integrating the second equality in (\ref{qtildeindifferencecondition}) against $p$ over $\Delta_0$. It follows that $M(\beta, \Delta_0) \geq M(\tilde{\beta},\Delta_0) + O(\delta^2)$.
    
    Next, let us discuss how to choose $W(x;q)$ for $x \in \mathcal{X}\setminus \Delta_0$. Since we can set $W(x;q) = 0$ for $x\in \mathcal{X}_{\text{AI}}$, we will only need to consider $x\in \mathcal{X}_{\text{IN}}\setminus \Delta_0$. We will show that we can choose $W(x;q)$ such that $q(x)/p(x) = \tilde{q}(x)/p(x)$ for all $x \in \mathcal{X}_{\text{IN}}\setminus \Delta_0$.
    Next, we integrate the first equality of (\ref{qindifferencecondition}), with $\tilde{q}(x)/p(x)$ replacing $q(x)/p(x)$, against $p$ over $\mathcal{X}_{\text{IN}}\setminus \Delta_0$ to get:
    \begin{multline*}\label{genaiusagequantitybeforeafter}
        M(\beta, \mathcal{X}_{\text{IN}}\setminus\Delta_0) = \frac{\int_{\mathcal{X}_{\text{IN}}\setminus \Delta_0}p(z)dz - \int_{\mathcal{X}_{\text{IN}}\setminus \Delta_0}\tilde{q}(z)dz + (M(\beta,\Delta_0)+M(\beta, \mathcal{X}_{\text{AI}}))\int_{\mathcal{X}_{\text{IN}}\setminus\Delta_0}g(z)dz}{1 - \int_{\mathcal{X}_{\text{IN}}\setminus \Delta_0}g(z)dz}\\
        \geq \frac{\int_{\mathcal{X}_{\text{IN}}\setminus \Delta_0}p(z)dz - \int_{\mathcal{X}_{\text{IN}}\setminus \Delta_0}\tilde{q}(z)dz + (M(\tilde{\beta},\Delta_0) + M(\tilde{\beta}, \mathcal{X}_{\text{AI}}))\int_{\mathcal{X}_{\text{IN}}\setminus\Delta_0}g(z)dz}{1 - \int_{\mathcal{X}_{\text{IN}}\setminus \Delta_0}g(z)dz} + O(\delta^2)\\
        = M(\tilde{\beta},\mathcal{X}\setminus \Delta_0) + O(\delta^2),\numberthis
    \end{multline*}
    where $M(\beta, \mathcal{X}_{\text{AI}}) = \int_{\mathcal{X}_{\text{AI}}}p(z)dz = M(\tilde{\beta}, \mathcal{X}_{\text{AI}}) + O(\delta^2)$. From the first line above, we obtain the expression for $M(\beta, \mathcal{X}) = M(\beta, \mathcal{X}_{\text{AI}}) + M(\beta, \mathcal{X}_{\text{IN}}\setminus \Delta_0) + M(\beta, \Delta_0)\geq M(\tilde{\beta}, \mathcal{X}_{\text{AI}}) + M(\tilde{\beta}, \mathcal{X}_{\text{IN}}\setminus \Delta_0) + M(\tilde{\beta}, \Delta_0) + O(\delta^2) = M(\tilde{\beta}, \mathcal{X}) + O(\delta^2)$, substituting this back into the first equality of (\ref{qindifferencecondition}) with $q(x)/p(x) = \tilde{q}(x)/p(x)$, we obtain the needed $\beta(x)$ for $x \in \mathcal{X}_{\text{IN}}\setminus \Delta_0$. Furthermore, we have
    \begin{multline*}\label{genairevenuebeforeafter}
        V^W(g;q)+c = \frac{\frac{1-\gamma}{M(\beta,\mathcal{X})^{\alpha}}\int_{\mathcal{X}_{\text{AI}}}p(z)/r(z)^{1-\alpha}dz + c}{\int_{\mathcal{X}_{\text{AI}}}p(z)/r(z)dz}\\
        \leq \frac{\frac{1-\gamma}{M(\tilde{\beta},\mathcal{X})^{\alpha}}\int_{\mathcal{X}_{\text{AI}}}p(z)/r(z)^{1-\alpha}dz + c}{\int_{\mathcal{X}_{\text{AI}}}p(z)/r(z)dz}+O(\delta^2) = V^{\widetilde{W}}(g;\tilde{q}) + O(\delta^2),\numberthis
    \end{multline*}
    from this and (\ref{qindifferencecondition}) we can determine the needed compensation for $x \in \mathcal{X}_{\text{IN}}\setminus \Delta_0$ to be:
    \begin{multline*}\label{compensationbeforeafter}
        W(x;q) = V^W(g;q) + c - (1-\gamma)\left(\frac{p(x)}{q(x)}\right)^{\alpha}\\
        \leq V^{\widetilde{W}}(g;\tilde{q}) + c - (1-\gamma)\left(\frac{p(x)}{\tilde{q}(x)}\right)^{\alpha} + O(\delta^2) = \widetilde{W}(x;\tilde{q}) + O(\delta^2),\numberthis
    \end{multline*}
    where the last equality followed from (\ref{qtildeindifferencecondition}). In other words, $W(x;q) = \widetilde{W}(x;\tilde{q}) + O(\delta)$ where the $\delta$--order term is negative. Assembling all the results, we can write the platform's profit under $W$ as:
    \begin{multline*}\label{profitafteraverageperturbation}
        \Pi^{W}(q)
        = \gamma M(\beta,\mathcal{X})^{1-\alpha}\int_{\mathcal{X}_{\text{AI}}}\frac{p(x)}{r(x)^{1-\alpha}}dx + \int_{\Delta_0}\phi^{W}(W(x;q);q)p(x)dx\\
        + \gamma\int_{\mathcal{X}_{\text{IN}}\setminus\Delta_0}\left(\frac{q(x)}{p(x)}\right)^{1-\alpha}p(x)dx - \gamma\int_{\mathcal{X}_{\text{IN}}\setminus\Delta_0}W(x;q)\left(\frac{q(x)}{p(x)}\right)p(x)dx + O(\delta^2)\\
        \geq \Pi^{\widetilde{W}}(\tilde{q}) + O(\delta^2),\numberthis
    \end{multline*}
    where the inequality follows by comparing each term of $\Pi^{W}(q)$ in (\ref{profitafteraverageperturbation}) to the counter--part of $\Pi^{\widetilde{W}}(\tilde{q})$ in (\ref{profitbeforeaverageperturbation}). The inequality for the first term follows from (\ref{genaiusagequantitybeforeafter}), the inequality for the second term follows from (\ref{jensenineqonphi}) and the fact that $\phi^{W}(w;q) = \phi^{\widetilde{W}}(w;\tilde{q}) + O(\delta)$ due to $V^{W}(w;q) = V^{\widetilde{W}}(w;\tilde{q}) + O(\delta)$, the third term remains the same as $q(x)/p(x) = \tilde{q}(x)/p(x)$ for $x\in \mathcal{X}_{\text{IN}}\setminus\Delta_0$, and the inequality for the fourth term follows from (\ref{compensationbeforeafter}). This completes this part of the proof.

    \emph{Part 2 (If $x\in \mathcal{X}$ receives compensation then so should any $y \in \mathcal{X}$ with $r(y) > r(x)$):}
    
    Let $(\tilde{\beta}, \tilde{q})$ be an equilibrium under the given platform's compensation scheme $\widetilde{W}$ and let $\mathcal{X} = \widetilde{\mathcal{X}}_{\text{AI}}\sqcup\widetilde{\mathcal{X}}_{\text{IN}}$ be the resulting decomposition. Define: $\overline{r} := \esssup_{x\in \mathcal{X}_{\text{AI}}}r(x)$ and $\underline{r} := \essinf_{x\in \mathcal{X}_{\text{IN}}}r(x)$, and suppose that $\overline{r} > \underline{r}$. Since the distribution of $r(x)$ under $p$ is atomless, the map $\delta\mapsto \int_{r^{-1}[\overline{r} - \delta,\overline{r}]}p(x)dx$ and $\delta\mapsto \int_{r^{-1}[\underline{r}, \underline{r}+\delta]}p(x)dx$ are monotonically increasing and continuous in $\delta \geq 0$. Let us choose a small $\delta > 0$ such that $\int_{r^{-1}[\overline{r} - \delta,\overline{r}]}p(x)dx = O(\delta) > 0$, then there exists a continuous $\kappa:\mathbb{R}_{\geq 0}\rightarrow \mathbb{R}_{\geq 0}$ with $\kappa(0) = 0$ such that $\int_{r^{-1}[\underline{r}, \underline{r}+\kappa(\delta)]}p(x)dx = \int_{r^{-1}[\overline{r} - \delta,\overline{r}]}p(x)dx$. We shall assume that $\underline{r} + \kappa(\delta) < \overline{r} - \delta$, or we can always choose a smaller $\delta > 0$ otherwise. Fix the chosen $\delta > 0$, then to shorten the notation, we define: $\overline{\Delta} := r^{-1}[\overline{r}-\delta, \overline{r}]\cap\widetilde{\mathcal{X}}_{\text{AI}}\subset \mathcal{X}$ and $\underline{\Delta} := r^{-1}[\underline{r}, \underline{r}+\kappa(\delta)]\cap \widetilde{\mathcal{X}}_{\text{IN}} \subset \mathcal{X}$. 
    
    It is safe to assume that $\widetilde{W}(x;\tilde{q})=0$ for $x \in \overline{\Delta}$, and $V^{\widetilde{W}}(x;\tilde{q}) = V^{\widetilde{W}}(g;\tilde{q})+c$ for $x\in \underline{\Delta}$. Moreover, from the first part, we can also assume that $\widetilde{W}(x;\tilde{q})$ is constant over $\underline{\Delta}$, and let us denote this constant by $\widetilde{W}(\underline{\Delta})$. Consequently, we can assume that $\tilde{\beta}(x) = \tilde{\beta}(\underline{\Delta}) + O(\delta)$ for some constant $\tilde{\beta}(\underline{\Delta}) \in [0,1]$, for all $x \in \underline{\Delta}$. We must have $\widetilde{W}(\underline{\Delta}) > 0$ since $V^{\widetilde{W}}(x;\tilde{q}) = V(x;\tilde{q}) + \widetilde{W}(\underline{\Delta}) = V^{\widetilde{W}}(g;\tilde{q})+c$ but 
    \begin{multline*}\label{bothdeltasuseai}
        V(x;\tilde{q}) = (1-\gamma)\left(1-\tilde{\beta}(x)+\frac{M(\tilde{\beta}, \mathcal{X})}{r(x)}\right)^{-\alpha} \leq (1-\gamma)\left(\frac{r(x)}{M(\tilde{\beta}, \mathcal{X})}\right)^{\alpha}\\
        < (1-\gamma)\left(\frac{r(y)}{M(\tilde{\beta}, \mathcal{X})}\right)^{\alpha} = V(y;\tilde{q}) = V^{\widetilde{W}}(y;\tilde{q}) < V^{\widetilde{W}}(g;\tilde{q})+c,\numberthis
    \end{multline*}
    for any $x \in \underline{\Delta}, y\in \overline{\Delta}$. To study what happen when we remove the compensation from creators in $\underline{\Delta}$ and give compensation to creators in $\overline{\Delta}$ instead, let us start from $\widetilde{\widetilde{W}} : \mathcal{X}\times\mathcal{D}(\mathcal{X})\rightarrow \mathbb{R}_{\geq 0}$ given by $\widetilde{\widetilde{W}}(x;\tilde{\tilde{q}}) := \widetilde{W}(x;\tilde{q})\cdot \mathbbm{1}[x \notin \underline{\Delta}]$, where neither $\overline{\Delta}$ nor $\underline{\Delta}$ receive any compensation. Let $(\tilde{\tilde{\beta}}, \tilde{\tilde{q}})$ and $\mathcal{X} = \widetilde{\widetilde{\mathcal{X}}}_{\text{AI}}\sqcup \widetilde{\widetilde{\mathcal{X}}}_{\text{IN}}$ denotes the corresponding equilibrium and decomposition, respectively. In particular, we have $\overline{\Delta}\cup \underline{\Delta}\subset \widetilde{\widetilde{\mathcal{X}}}_{\text{AI}}$ as we have argued in (\ref{bothdeltasuseai}). Now, we consider a new compensation scheme $W : \mathcal{X}\times \mathcal{D}(\mathcal{X}) \rightarrow \mathbb{R}_{\geq 0}$ given by raising compensation to creators in $\overline{\Delta}$: $W(x;q) := (\widetilde{\widetilde{W}}(x;\tilde{\tilde{q}}) + O(\delta))\cdot \mathbbm{1}[x \notin \overline{\Delta}] + W(\overline{\Delta})\cdot \mathbbm{1}[x\in \overline{\Delta}]$, where $W(\overline{\Delta}) \geq 0$ is a parameter to be chosen, and the $O(\delta)$ part will be discussed. Let $(\beta, q)$ denotes the corresponding equilibrium. We have: 
    \begin{equation}\label{betaoverlinedelta}
        \beta(y) = 1 + \left[\frac{M(\tilde{\tilde{\beta}}, \mathcal{X})}{\overline{r}} - \left(\frac{1-\gamma}{V^{\widetilde{\widetilde{W}}}(g;\tilde{\tilde{q}}) + c - W(\overline{\Delta})}\right)^{1/\alpha} + O(\delta)\right]\cdot \mathbbm{1}[W(\overline{\Delta}) \geq W^*(\overline{\Delta})],
    \end{equation}
    for any $y \in \overline{\Delta}$, where $W^*(\overline{\Delta})$ is the minimum compensation for $y$ to be indifferent between manual creation and using GenAI under $W$; $r(y) = \overline{r} +O(\delta)$ by construction, and $V^W(g;q) = V^{\widetilde{\widetilde{W}}}(g;\tilde{\tilde{q}}) + O(\delta)$ due to the changes in compensation over $\overline{\Delta}$ which has measure $O(\delta)$. More precisely, we can write:
    \begin{multline*}\label{vwgqasafunctionofwoverlinedelta}
        V^W(g;q)+c = \\
        \frac{\frac{1-\gamma}{M(\beta, \mathcal{X})^\alpha}\int_{\widetilde{\widetilde{\mathcal{X}}}_{\text{AI}}\setminus \overline{\Delta}}\frac{p(z)}{r(z)^{1-\alpha}} dz + \int_{\overline{\Delta}}\frac{1-\gamma}{M(\beta, \mathcal{X})^\alpha}\frac{p(z)}{r(z)^{1-\alpha}}+\frac{W(\overline{\Delta})}{r(z)}p(z) dz\cdot \mathbbm{1}[W(\overline{\Delta}) < W^*(\overline{\Delta})] + c}{\int_{\widetilde{\widetilde{\mathcal{X}}}_{\text{AI}}\setminus \overline{\Delta}}p(z)/r(z) dz + \int_{\overline{\Delta}}p(z)/r(z) dz\cdot \mathbbm{1}[W(\overline{\Delta}) < W^*(\overline{\Delta})]}\\
        + O(\delta^2).\numberthis
    \end{multline*}
    It is true that not all $y \in \overline{\Delta}$ will require the same $W(\overline{\Delta})$ to be indifferent, due to the variation of order $O(\delta)$ of $r(y)$ in $\overline{\Delta}$. However, taking this into account give $O(\delta)$ correction to $W(\overline{\Delta})$ and $\beta(y)$ over the $\overline{\Delta}$ which has $O(\delta)$ measure, thus, only contribute $O(\delta^2)$ to the rest of the calculation and can be ignored. We will compare $W$ to $\widetilde{W}':\mathcal{X}\times \mathcal{D}(\mathcal{X}) \rightarrow \mathbb{R}_{\geq 0}$ given by: $\widetilde{W}'(x;\tilde{q}) := (\widetilde{\widetilde{W}}(x;\tilde{\tilde{q}}) + O(\delta))\cdot \mathbbm{1}[x \notin \underline{\Delta}] + \widetilde{W}'(\underline{\Delta})\cdot \mathbbm{1}[x \in \underline{\Delta}]$. Let $(\tilde{\beta}', \tilde{q}')$ denotes the corresponding equilibrium. Note that if $\widetilde{W}'(\underline{\Delta}) = \widetilde{W}(\underline{\Delta})$ then we have $\widetilde{W} = \widetilde{W}'$ and we recover the original equilibrium: $(\tilde{\beta}',\tilde{q}') = (\tilde{\beta},\tilde{q})$. For other choices of $\widetilde{W}'(\underline{\Delta})$, we have:
    \begin{equation}\label{betaunderlinedelta}
        \tilde{\beta}'(x) = 1 + \left[\frac{M(\tilde{\tilde{\beta}}, \mathcal{X})}{\underline{r}} - \left(\frac{1-\gamma}{V^{\widetilde{\widetilde{W}}}(g;\tilde{\tilde{q}}) + c - \widetilde{W}'(\underline{\Delta})}\right)^{1/\alpha} + O(\delta)\right]\cdot \mathbbm{1}[\widetilde{W}'(\underline{\Delta}) \geq \widetilde{W}'^*(\underline{\Delta})]
    \end{equation}
    for any $x \in \underline{\Delta}$, where $\widetilde{W}'^*(\underline{\Delta}) \leq \widetilde{W}(\underline{\Delta})$ is the minimum compensation for $x$ to be indifferent between manual creation and using GenAI under $\widetilde{W}'$; $r(x) = \underline{r} + O(\delta)$, and $V^{\widetilde{W}}(g;\tilde{q}) = V^{\widetilde{\widetilde{W}}}(g;\tilde{\tilde{q}}) + O(\delta)$. Similarly to what we had previously, we can write:
    \begin{multline*}\label{vtildewgqasafunctionofwoverlinedelta}
        V^{\widetilde{W}'}(g;\tilde{q})+c = \\
        \frac{\frac{1-\gamma}{M(\tilde{\beta}, \mathcal{X})^\alpha}\int_{\widetilde{\widetilde{\mathcal{X}}}_{\text{AI}}\setminus \underline{\Delta}}\frac{p(z)}{r(z)^{1-\alpha}} dz + \int_{\underline{\Delta}}\frac{1-\gamma}{M(\tilde{\beta}, \mathcal{X})^\alpha}\frac{p(z)}{r(z)^{1-\alpha}}+\frac{\widetilde{W}'(\underline{\Delta})}{r(z)}p(z) dz\cdot \mathbbm{1}[\widetilde{W}'(\underline{\Delta}) < \widetilde{W}'^*(\underline{\Delta})] + c}{\int_{\widetilde{\widetilde{\mathcal{X}}}_{\text{AI}}\setminus \underline{\Delta}}p(z)/r(z) dz + \int_{\underline{\Delta}}p(z)/r(z) dz\cdot \mathbbm{1}[\widetilde{W}'(\underline{\Delta}) < \widetilde{W}'^*(\underline{\Delta})]}\\
        + O(\delta^2).\numberthis
    \end{multline*} 
    If $W(\overline{\Delta}) \in [0,W^*(\overline{\Delta}))$ then $(\beta,q) = (\tilde{\tilde{\beta}}, \tilde{\tilde{q}})$ and $M(\beta, \mathcal{X}) = M(\tilde{\tilde{\beta}}, \mathcal{X})$, in other words, any compensation $W(\overline{\Delta}) \in [0,W^*(\overline{\Delta}))$ is deemed too low to be effective in changing any creators' decision from the one at equilibrium under $\widetilde{\widetilde{W}}$. For $W(\overline{\Delta}) \in [0,W^*(\overline{\Delta}))$ we can see from (\ref{betaunderlinedelta}) that $V^W(g;q)$ increases linearly in $W(\overline{\Delta})$ from $V^{\widetilde{\widetilde{W}}}(g;\tilde{\tilde{q}})$ with a slope $\int_{\overline{\Delta}}p(z)/\overline{r}dz/\int_{\widetilde{\widetilde{\mathcal{X}}}_{\text{AI}}}p(z)/r(z)dz$, up to the $O(\delta^2)$ order. Similarly, if $\widetilde{W}'(\underline{\Delta}) \in [0, \widetilde{W}'^*(\underline{\Delta}))$ then we have $(\tilde{\beta}',\tilde{q}') = (\tilde{\tilde{\beta}}, \tilde{\tilde{q}})$, $M(\tilde{\beta}, \mathcal{X}) = M(\tilde{\tilde{\beta}}, \mathcal{X})$ and $V^{\widetilde{W}}(g;\tilde{q})$ increases linearly in $W(\underline{\Delta})$ with a slope $\int_{\underline{\Delta}}p(z)/\underline{r}dz/\int_{\widetilde{\widetilde{\mathcal{X}}}_{\text{AI}}}p(z)/r(z)dz$, up to the $O(\delta^2)$ order. Since $\overline{r} > \underline{r}$, the slope of $V^W(g;q)$ is less than the slope of $V^{\widetilde{W}}(g;\tilde{q})$. But $\overline{r} > \underline{r}$ also means that $V(y;\tilde{\tilde{q}}) = (1-\gamma)\left(\frac{\overline{r}}{M(\tilde{\tilde{\beta}},\mathcal{X})}\right)^{\alpha} > (1-\gamma)\left(\frac{\underline{r}}{M(\tilde{\tilde{\beta}},\mathcal{X})}\right)^\alpha = V(x;\tilde{\tilde{q}})$ for all $x\in \underline{\Delta}, y\in \overline{\Delta}$. Therefore $W^*(\overline{\Delta}) < \widetilde{W}'^*(\underline{\Delta})$, and we have:
    \begin{equation}\label{vwvwtildeatstar}
        V^W(g;q)|_{W(\overline{\Delta}) = W^*(\overline{\Delta})} < V^{\widetilde{W}'}(g;\tilde{q}')|_{\widetilde{W}'(\underline{\Delta}) = \widetilde{W}'^*(\underline{\Delta})}.
    \end{equation}

    We will choose $W$ such that the equilibrium $(\beta, q)$ satisfies: $\beta(y) := \beta(\overline{\Delta})+O(\delta) = \tilde{\beta}(\underline{\Delta}) + O(\delta)$ for all $y \in \overline{\Delta}$, $\beta(x) = 1$ for $x \in \widetilde{\widetilde{\mathcal{X}}}_{\text{AI}}\setminus \overline{\Delta}$, and $\beta(x) = \tilde{\beta}(x) = \tilde{\tilde{\beta}}(x)+O(\delta)$ for $x \in \widetilde{\widetilde{\mathcal{X}}}_{\text{IN}}$. Under these conditions, we would have $M(\beta, \mathcal{X}) = M(\tilde{\beta}, \mathcal{X}) + O(\delta^2)$ and $q(x)/p(x) = 1-\beta(x) + \frac{M(\beta, \mathcal{X})}{r(x)} = 1-\tilde{\beta}(x) + \frac{M(\tilde{\beta},\mathcal{X})}{r(x)} + O(\delta^2) = \tilde{q}(x)/p(x)+O(\delta^2)$ for all $x \in \mathcal{X}\setminus (\overline{\Delta}\cup \underline{\Delta})$. From (\ref{vwgqasafunctionofwoverlinedelta}) we can see that $V^W(g;q)$ is linear in $(1-\gamma)M(\beta, \mathcal{X})^{-\alpha}$ for $W(\overline{\Delta}) \geq W^*(\overline{\Delta})$ with slope $\int_{\widetilde{\widetilde{\mathcal{X}}}_{\text{AI}}\setminus \overline{\Delta}}p(z)/r(z)^{1-\alpha}dz/\int_{\widetilde{\widetilde{\mathcal{X}}}_{\text{AI}}\setminus \overline{\Delta}}p(z)/r(z)dz$. Similarly, from (\ref{vtildewgqasafunctionofwoverlinedelta}) we can see that $V^{\widetilde{W}'}(g;\tilde{q}')$ is linear in $(1-\gamma)M(\tilde{\beta}',\mathcal{X})^{-\alpha}$ for $\widetilde{W}'(\underline{\Delta}) \geq \widetilde{W}'^*(\underline{\Delta})$, including at $\widetilde{W}'(\underline{\Delta}) = \widetilde{W}(\underline{\Delta})$ where $(\tilde{\beta}', \tilde{q}') = (\tilde{\beta}, \tilde{q})$, with slope $\int_{\widetilde{\widetilde{\mathcal{X}}}_{\text{AI}}\setminus \underline{\Delta}}p(z)/r(z)^{1-\alpha}dz/\int_{\widetilde{\widetilde{\mathcal{X}}}_{\text{AI}}\setminus \underline{\Delta}}p(z)/r(z)dz$. Note that both slopes are equal up to the $O(\delta)$ order, while $M(\beta, \mathcal{X}) - M(\tilde{\tilde{\beta}}, \mathcal{X}) = M(\tilde{\beta}, \mathcal{X})-M(\tilde{\tilde{\beta}}, \mathcal{X}) + O(\delta^2)$ has order $O(\delta)$, therefore, we have from (\ref{vwvwtildeatstar}) that $V^W(g;q) \leq V^{\widetilde{W}}(g;\tilde{q}) + O(\delta^2)$.
    
    By comparing (\ref{betaoverlinedelta}) and (\ref{betaunderlinedelta}), to get $\beta(y) = \tilde{\beta}(\underline{\Delta}) + O(\delta)$ for $y \in \overline{\Delta}$, it follows that we must choose $W(\overline{\Delta}) \in [W^*(\overline{\Delta}), \widetilde{W}(\underline{\Delta})]$. For $x \in \widetilde{\widetilde{\mathcal{X}}}_{\text{AI}}\setminus \overline{\Delta}$ we set $W(x;q) = 0$. Finally, for $x\in \widetilde{\widetilde{\mathcal{X}}}_{\text{IN}}$ we find $W(x;q)$ from the indifference condition given by the second equality in (\ref{qindifferencecondition}) with $\beta(x) = \tilde{\beta}(x)$. But $\tilde{\beta}(x)$ also satisfies the indifference condition given by the second equality of (\ref{qtildeindifferencecondition}), while $M(\beta, \mathcal{X}) = M(\tilde{\beta}, \mathcal{X}) + O(\delta^2)$ and $V^W(g;q) \leq V^{\widetilde{W}}(g;\tilde{q}) + O(\delta^2)$, which means $W(x;q) \leq \widetilde{W}(x;\tilde{q}) + O(\delta^2)$ for all $x \in \widetilde{\widetilde{\mathcal{X}}}_{\text{IN}}$. Finally, using (\ref{platformrevenue}) we compare the platform's profit at equilibrium under $W$ and $\widetilde{W}$:
    \begin{multline*}
        \Pi^W(q) - \Pi^{\widetilde{W}}(\tilde{q}) = \gamma\left[\left(\frac{M(\beta, \mathcal{X})}{\underline{r}}\right)^{1-\alpha} - \left(1-\tilde{\beta}(\underline{\Delta}) + \frac{M(\tilde{\beta}, \mathcal{X})}{\underline{r}}\right)^{1-\alpha}\right]\int_{\underline{\Delta}}p(x)dx\\
        + \gamma\left[\left(1-\beta(\overline{\Delta}) + \frac{M(\beta, \mathcal{X})}{\overline{r}}\right)^{1-\alpha} - \left(\frac{M(\tilde{\beta}, \mathcal{X})}{\overline{r}}\right)^{1-\alpha}\right]\int_{\overline{\Delta}}p(x)dx\\
        - W(\overline{\Delta})\left(1-\beta(\overline{\Delta}) + \frac{M(\beta, \mathcal{X})}{\overline{r}}\right)\int_{\overline{\Delta}}p(x)dx + \widetilde{W}(\underline{\Delta})\left(1-\tilde{\beta}(\underline{\Delta}) + \frac{M(\tilde{\beta}, \mathcal{X})}{\underline{r}}\right)\int_{\underline{\Delta}}p(x)dx\\
        - \int_{\widetilde{\widetilde{\mathcal{X}}}_{\text{IN}}}\left(W(x;q)\frac{q(x)}{p(x)} - \widetilde{W}(x;q)\frac{\tilde{q}(x)}{p(x)}\right)p(x)dx + O(\delta^2) \geq 0,
    \end{multline*}
    keeping in mind that $\int_{\overline{\Delta}}p(x)dx = \int_{\underline{\Delta}}p(x)dx$, $\beta(\overline{\Delta}) = \tilde{\beta}(\underline{\Delta}) + O(\delta)$, $M(\beta, \mathcal{X}) = M(\tilde{\beta}, \mathcal{X}) + O(\delta^2)$, $q(x)/p(x) = \tilde{q}(x)/p(x) + O(\delta^2)$ for $x \in \mathcal{X}\setminus (\overline{\Delta}\cup \underline{\Delta})$, and that $W(x;q) = \widetilde{W}(x;\tilde{q}) = 0$ for $x\in \widetilde{\widetilde{\mathcal{X}}}_{\text{AI}}\setminus (\overline{\Delta}\cup \underline{\Delta})$. To the first order of $\delta$, the sum of the first two terms is positive due to concavity of $r\mapsto r^{1-\alpha}$, the fourth term is greater than the third term since $W(\overline{\Delta})\leq \widetilde{W}(\underline{\Delta})$ and $1/\overline{r}\leq 1/\underline{r}$, finally, the last term is positive since $q(x)/p(x) = \tilde{q}(x)/p(x) + O(\delta^2)$ and $W(x;q) \leq \widetilde{W}(x;\tilde{q}) + O(\delta^2)$. This completes this part of the proof. 

    \emph{Part 3 (Finalize):}

    Start from the given compensation scheme $\widetilde{W} : \mathcal{X}\times \mathcal{D}(\mathcal{X})\rightarrow \mathbb{R}_{\geq 0}$ with an equilibrium $(\tilde{\beta}, \tilde{q})$. By repeated application of Part 1 and Part 2, we can find $W$ with an equilibrium $(\beta, q)$, where $W$ is characterized by some threshold $\underline{r} > 0$ such that $W(x;q) = W(x) > 0$ if $r(x) \geq \underline{r}$, and $W(x) = 0$ otherwise, such that $\Pi^W(q) \geq \Pi^{\widetilde{W}}(\tilde{q})$. We have the decomposition $\mathcal{X} = \mathcal{X}_{\text{AI}}\sqcup \mathcal{X}_{\text{IN}}$ where $\mathcal{X}_{\text{AI}} := r^{-1}[0,\underline{r})$ and $\mathcal{X}_{\text{IN}} := r^{-1}[\underline{r}, \infty)$. If there exists $\overline{\Delta} \subset \mathcal{X}_{\text{IN}}$ and $\underline{\Delta}\subset \mathcal{X}_{\text{IN}}$ with positive measure $\int_{\overline{\Delta}}p(z)dz = \int_{\underline{\Delta}}p(z)dz = O(\delta) > 0$ such that $\inf_{y\in \overline{\Delta}}\beta(y) > \sup_{x \in \underline{\Delta}}\beta(x)$ but $\inf_{y\in \overline{\Delta}}r(y) > \sup_{x\in \underline{\Delta}}r(x)$ then we repeat the similar argument to Part 2 to transfer an appropriate amount of compensation from creators in $\underline{\Delta}$ to those in $\overline{\Delta}$ while improving the platform's profit. Therefore, we can assume that under $W$, we have $W(y) \geq W(x)$ if $r(y) \geq r(x)$ for all $x,y\in \mathcal{X}$. Lastly, we can apply the similar argument to Part 1 to any $\Delta_0\subset \mathcal{X}_{\text{IN}}$ with measure $\int_{\Delta_0}p(z)dz = O(\delta) > 0$ to improve the platform's profit by replacing $W$ over $\Delta_0$ with the constant $\mathbb{E}_{x\sim p}[W(x)|\Delta_0]$. By repeating this process, we monotonically improving the platform's profit, and we conclude that there exists $W : \mathcal{X}\times \mathcal{D}(\mathcal{X})\rightarrow \mathbb{R}_{\geq 0}$ given by paying a fixed compensation $w \geq 0$ to all $x \in \mathcal{X}_{\text{IN}} := r^{-1}[\underline{r}, \infty)$ for some threshold $\underline{r} \geq 0$, such that $\Pi^W(q) \geq \Pi^{\widetilde{W}}(\tilde{q})$. 
\end{proof}

\subsection{Proof of Lemma \ref{xhisemptylemma}}
\begin{proof}
    If $x \in \mathcal{X}_{\text{H}}$ then $U^{\underline{v}, w}(x;q) > V^{\underline{v}, w}(g;q)$ and $\beta(x) = 0$. It follows from (\ref{generaldecomposedform}) that $q(x) = p(x) + g(x)\cdot \int_{\mathcal{X}}\beta(y)p(y)dy \geq p(x)$, therefore $V(x;q) \leq 1-\gamma$ for all $x \in \mathcal{X}_{\text{H}}$. Since $U^{\underline{v}, w}(x;q) < U^{\underline{v}, w}(y;q)$ if and only if $V(x;q) < V(y;q)$, we have for any $x\in \mathcal{X}\setminus \mathcal{X}_{\text{H}}$ and $y \in \mathcal{X}_{\text{H}}$ that
    \begin{equation*}
        (1-\gamma)\left(\frac{p(x)}{q(x)}\right)^\alpha = V(x;q) < V(y;q) \leq 1-\gamma \qquad \implies \qquad q(x) > p(x).
    \end{equation*}
    since $U^{\underline{v},w}(x;q) \leq V^{\underline{v},w}(g;q) < U^{\underline{v}, w}(y;q)$. This implies that $q(x) \geq p(x)$ for all $x\in \mathcal{X}$ where the inequality is strict for $x \in \mathcal{X}\setminus \mathcal{X}_{\text{H}}$. Therefore, $\int_{\mathcal{X}}q(x)dx > \int_{\mathcal{X}}p(x)dx = 1$, contradicting the normalization of $q$, unless either $\mathcal{X}\setminus\mathcal{X}_{\text{H}}$ has measure zero or our initial assumption of $x\in \mathcal{X}_{\text{H}}$ is incorrect and we have $\mathcal{X}_{\text{H}} = \emptyset$. Let us analyze the first possibility: $\mathcal{X}\setminus\mathcal{X}_{\text{H}}$ has measure zero, this is the same as $\beta(x) = 0$ for almost all $x$, so we have $q = p$ almost everywhere. Then $V(x;q) = 1-\gamma$, which gives $V^{\underline{v}, w}(x;q) = 1-\gamma+w\cdot \mathbbm{1}[1-\gamma\geq \underline{v}]$ for almost all $x \in \mathcal{X}$, after compensation. On the other hand, the expected revenue from GenAI becomes:
    \begin{multline*}
        V^{\underline{v}, w}(g;q) = \int_{\mathcal{X}}\left(1-\gamma+w\cdot \mathbbm{1}[1-\gamma\geq \underline{v}]\right)g(y)dy = 1-\gamma+w\cdot \mathbbm{1}[1-\gamma\geq \underline{v}]\\
        = V^{\underline{v}, w}(x;q) = U^{\underline{v}, w}(x;q) + c > U^{\underline{v}, w}(x;q).
    \end{multline*}
    for almost all $x\in \mathcal{X}$. Therefore it is profitable for almost all $x \in \mathcal{X}_{\text{H}}$ to deviate to use GenAI, which is a contradiction. This leaves us with the only conclusion: $\mathcal{X}_{\text{H}} = \emptyset$, which completes the proof of the first claim.

    Now, suppose that $\mathcal{X}_{\text{AI}}$ has a zero measure, then combining with $\mathcal{X}_{\text{H}}=\emptyset$ we have just proven, we have $x\in\mathcal{X}_{\text{IN}}$ for almost every $x\in \mathcal{X}$. Then we have $U^{\underline{v}, w}(x;q) = V^{\underline{v}, w}(g;q)$ for almost every $x \in \mathcal{X}$, but this means
    \begin{multline*}
        V^{\underline{v}, w}(g;q) = \int_{\mathcal{X}}V^{\underline{v}, w}(y;q)g(y)dy = \int_{\mathcal{X}_{\text{IN}}}V^{\underline{v}, w}(y;q)g(y)dy\\
        = (V^{\underline{v}, w}(g;q)+c)\int_{\mathcal{X}_{\text{IN}}}g(y)dy = (V^{\underline{v}, w}(g;q)+c)\int_{\mathcal{X}}g(y)dy = V^{\underline{v}, w}(g;q)+c
    \end{multline*}
    which is false since $c > 0$. Thus, we conclude that $\mathcal{X}_{\text{AI}}$ must have a positive measure.
\end{proof}

\subsection{Proof of Lemma \ref{aislopresult}}
\begin{proof}
    Suppose that $q = p$ and $g\neq p$. Recall from Lemma \ref{xhisemptylemma} that we have a decomposition $\mathcal{X} = \mathcal{X}_{\text{AI}}\sqcup \mathcal{X}_{\text{IN}}$, then it must be the case that $\mathcal{X}_{\text{IN}}\neq \emptyset$, otherwise we simply have $q = g\neq p$. Since $q = p$ implies that $V(x;q) = 1-\gamma$ for all $x \in \mathcal{X}$, this means the compensated revenue $V^{\underline{v}, w}(x;q)$ is also independent of $x \in \mathcal{X}$. But then, given any $x \in \mathcal{X}_{\text{IN}}$ we have $V^{\underline{v}, w}(g;q) = \int_{\mathcal{X}}V^{\underline{v}, w}(y;q)g(y)dy = V^{\underline{v}, w}(x;q) > V^{\underline{v}, w}(x;q) - c = U^{\underline{v}, w}(x;q)$, so it would be profitable for any $x\in \mathcal{X}_{\text{IN}}$ to deviates to strictly using GenAI, contradicting $\mathcal{X}_{\text{IN}}\neq \emptyset$. 
\end{proof}

\subsection{Proof of Proposition \ref{generalequilibriumproposition}}
\begin{proof}
    First, we argue that the condition $1-\gamma < \underline{v} < (1-\gamma)\sup_{x\in\mathcal{X}}r(x)^\alpha$ implies that $r^{-1}[0, \underline{r}(\underline{v}))$ has a positive measure and $r^{-1}[\underline{r}(\underline{v}), \infty)$ is non--empty. The positive measure of $r^{-1}[0, \underline{r}(\underline{v}))$ is essential to ensure that the denominator of $\widetilde{V}^{\underline{v}}(g)$ is positive, so that it is well--defined. Suppose that $r^{-1}[0, \underline{r}(\underline{v}))$ has a zero measure then the LHS of (\ref{definingbarrbarv}) would be exactly $1$ while the RHS is strictly greater than $1$ since $\underline{v} > 1-\gamma$, a contradiction. Next, suppose that $r^{-1}[\underline{r}(\underline{v}), \infty) = \emptyset$ then $r(x) < \underline{r}(\underline{v})$ for all $x\in \mathcal{X}$ so (\ref{definingbarrbarv}) becomes:
    \begin{equation*}
        \underline{r}(\underline{v}) = \int_{\mathcal{X}}\frac{\underline{r}(\underline{v})}{r(y)}p(y)dy =  \left(\frac{\underline{v}}{1-\gamma}\right)^{1/\alpha} < \sup_{x \in \mathcal{X}}r(x).
    \end{equation*}
    But $\underline{r}(\underline{v}) > r(x)$ for all $x\in \mathcal{X}$ also implies that $\underline{r}(\underline{v})\geq \sup_{x\in \mathcal{X}}r(x)$, which is a contradiction.

    It is straightforward to verify that $q$ as given in (\ref{generalequilibriumform}) is normalized by noting that (\ref{definingbarrbarv}) can also be written as:
    \begin{equation}\label{mgformula}
        \underline{r}(\underline{v})\int_{r^{-1}[0, \underline{r}(\underline{v}))}g(y)dy + \int_{r^{-1}[\underline{r}(\underline{v}), \infty)}p(y)dy = \left(\frac{\underline{v}}{1-\gamma}\right)^{1/\alpha}
    \end{equation}
    and that $q$ as given in (\ref{generalequilibriumform}) can be obtained by substituting $\beta(x)$ as given in (\ref{generalequilibriumform}) into the general form (\ref{generaldecomposedform}). Next, we observe that $\beta(x) \in [0,1]$ for all $x\in \mathcal{X}$. This is clear if $x \in r^{-1}[0,\underline{r}(\underline{v}))$ since we would simply have $\beta(x) = 1$, but if $x \in r^{-1}[\underline{r}(\underline{v}), \infty)$ then $1-\underline{r}(\underline{v})/r(x) \in [0,1]$, and since $\underline{v} > 1-\gamma$ we have
    \begin{equation}\label{betain01inequality}
        \beta(x) = 1 - \left(\frac{1-\gamma}{\underline{v}}\right)^{1/\alpha}\left(1 - \frac{\underline{r}(\underline{v})}{r(x)}\right) \in [0,1],\numberthis
    \end{equation}
    as needed. In particular, it is valid to interpret $q(x)$ given in (\ref{generalequilibriumform}) as a probability distribution of contents obtained from the GenAI usage propensity probability $\beta:\mathcal{X}\rightarrow [0,1]$. 
    
    Using $q(x)$ as given in  (\ref{generalequilibriumform}), we find that, for all $x \in r^{-1}[\underline{r}(\underline{v}), \infty)$: \begin{equation}\label{indifferencecondition}
        V(x;q) = (1-\gamma)\left(\frac{p(x)}{q(x)}\right)^{\alpha} = \underline{v},
    \end{equation}
    and for all $x \in r^{-1}[0, \underline{r}(\underline{v}))$:
    \begin{equation}\label{prefereneforaicondition}
        V(x;q) = (1-\gamma)\left(\frac{p(x)}{q(x)}\right)^{\alpha} = \underline{v}\left(\frac{r(x)}{\underline{r}(\underline{v})}\right)^\alpha < \underline{v}.
    \end{equation}
    It remains for us to verify that (\ref{indifferencecondition}) and (\ref{prefereneforaicondition}) in fact reinforce the indifference condition and the strict preference for GenAI condition, respectively, under the compensation scheme $(\underline{v}, w)$ where $w := \widetilde{V}^{\underline{v}}(g) + c - \underline{v}$. We have by definition and from (\ref{indifferencecondition}) that 
    \begin{equation*}
        V^{\underline{v}, w}(x;q) = V(x;q) + W^{\underline{v}, w}(x;q) = \underline{v} + \widetilde{V}^{\underline{v}}(g) + c - \underline{v} = \widetilde{V}^{\underline{v}}(g) + c,
    \end{equation*}
    for all $x \in r^{-1}[\underline{r}(\underline{v}), \infty)$, while $V^{\underline{v}, w}(x;q) = V(x;q)$, for all $x \in r^{-1}[0, \underline{r}(\underline{v}))$. We can compute the profit from GenAI assuming the content distribution density $q$ using (\ref{manualandaiutility}):
    \begin{multline*}\label{generaldecomposedformsimplified}
        V^{\underline{v}, w}(g;q) = \int_{\mathcal{X}}V^{\underline{v}, w}(y;q)g(y)dy\\
        = (1-\gamma)\int_{r^{-1}[0, \underline{r}(\underline{v}))}\left(\frac{p(y)}{q(y)}\right)^\alpha g(y)dy + (\widetilde{V}^{\underline{v}}(g)+c)\int_{r^{-1}[\underline{r}(\underline{v}),\infty)}g(y)dy\\
        = \frac{\underline{v}}{\underline{r}(\underline{v})^\alpha}M_g(\underline{v}) + (\widetilde{V}^{\underline{v}}(g)+c)\left(1 - \left(\frac{\underline{v}}{1-\gamma}\right)^{1/\alpha}\frac{1 - M_p(\underline{v})}{\underline{r}(\underline{v})}\right)\\
        = (\widetilde{V}^{\underline{v}}(g) + c)\left(\frac{\underline{v}}{1-\gamma}\right)^{1/\alpha}\frac{1-M_p(\underline{v})}{\underline{r}(\underline{v})} - c + (\widetilde{V}^{\underline{v}}(g)+c)\left(1 - \left(\frac{\underline{v}}{1-\gamma}\right)^{1/\alpha}\frac{1 - M_p(\underline{v})}{\underline{r}(\underline{v})}\right)\\
        = \widetilde{V}^{\underline{v}}(g).
    \end{multline*}
    Where we have used (\ref{mgformula}) for the second term the second equality. It follows that, $U^{\underline{v}, w}(x;q) = V^{\underline{v}, w}(x;q) - c = V^{\underline{v}, w}(g;q)$ for all $x \in r^{-1}[\underline{r}(\underline{v}), \infty)$ which is consistent with $\beta(x) \in [0,1]$, and $U^{\underline{v}, w}(x;q) = V(x;q) - c < \underline{v}-c \leq V^{\underline{v}, w}(g;q)$ for all $x \in r^{-1}[0, \underline{r}(\underline{v}))$ which is consistent with $\beta(x) = 1$. Therefore, $(\beta, q)$ is an equilibrium under $(\underline{v}, w)$, and we have $\mathcal{X}_{\text{IN}} = r^{-1}[\underline{r}(\underline{v}), \infty)$ and $\mathcal{X}_{\text{AI}} = r^{-1}[0,\underline{r}(\underline{v}))$, as claimed. 

    Conversely, consider an equilibrium $(\beta, q)$ under the platform compensation scheme $(\underline{v}, w)$ characterized by the decomposition $\mathcal{X} = \mathcal{X}_{\text{AI}}\sqcup \mathcal{X}_{\text{IN}}$ with $\mathcal{X}_{\text{IN}} \neq \emptyset$. From Lemma \ref{decisionundercompensationlemma}, the equilibrium content distribution $q$ satisfies the indifference condition: $V(x;q) = (1-\gamma)(p(x)/q(x))^\alpha = \underline{v}$ for all $x \in \mathcal{X}_{\text{IN}}$, and satisfies the strict preference for GenAI condition: $V(x;q) < \underline{v}$ for all $x\in \mathcal{X}_{\text{AI}}$. The first condition shows that $q(x) = \left(\frac{1-\gamma}{\underline{v}}\right)^{1/\alpha}p(x)$ for all $x\in \mathcal{X}_{\text{IN}}$, while we know from (\ref{generaldecomposedform}) that $q(x) = g(x)\cdot \int_{\mathcal{X}}\beta(y)p(y)dy$ for all $x \in \mathcal{X}_{\text{AI}}$. Using the normalization condition $\int_{\mathcal{X}}q(y)dy = 1$, we can determine 
    \begin{equation}\label{integralbetap}
    \int_{\mathcal{X}}\beta(y)p(y)dy = \frac{1 - \left(\frac{1-\gamma}{\underline{v}}\right)^{1/\alpha}\int_{\mathcal{X}_{\text{IN}}}p(y)dy}{\int_{\mathcal{X}_{\text{AI}}}g(y)dy}
\end{equation}
    which in turn give us:
\begin{equation}\label{generaldecomposedformsimplified}
    \begin{aligned}
    q(x) &= \left(\frac{1-\gamma}{\underline{v}}\right)^{1/\alpha}p(x) \cdot \mathbbm{1}[x\in \mathcal{X}_{\text{IN}}] + \frac{1 - \left(\frac{1-\gamma}{\underline{v}}\right)^{1/\alpha}\int_{\mathcal{X}_{\text{IN}}}p(y)dy}{\int_{\mathcal{X}_{\text{AI}}}g(y)dy}\cdot g(x)\cdot \mathbbm{1}[x\in \mathcal{X}_{\text{AI}}]\\
    \beta(x) &= \left[1 - \left(\frac{1-\gamma}{\underline{v}}\right)^{1/\alpha} + \frac{1 - \left(\frac{1-\gamma}{\underline{v}}\right)^{1/\alpha}\int_{\mathcal{X}_{\text{IN}}}p(y)dy}{\int_{\mathcal{X}_{\text{AI}}}g(y)dy}\cdot \frac{g(x)}{p(x)}\right]\cdot \mathbbm{1}[x\in \mathcal{X}_{\text{IN}}] + \mathbbm{1}[x\in \mathcal{X}_{\text{AI}}]
    \end{aligned}.
\end{equation}
Where $\beta(x)$ can be determined for $x \in \mathcal{X}_{\text{IN}}$ from $\left(\frac{1-\gamma}{\underline{v}}\right)^{1/\alpha}p(x) = q(x) = (1-\beta(x))p(x) + g(x)\int_{\mathcal{X}}\beta(y)p(y)dy$. It remains for us to show that $\mathcal{X}_{\text{AI}}$ and $\mathcal{X}_{\text{IN}}$ coincide with their counter--parts given by the solution $\underline{r}(\underline{v})$ to (\ref{definingbarrbarv}), then it would follow from (\ref{mgformula}) that (\ref{generaldecomposedformsimplified}) coincide with (\ref{generalequilibriumform}). For any $x \in \mathcal{X}_{\text{AI}}$ and any $x' \in \mathcal{X}_{\text{IN}}$ we have
\begin{equation}\label{separatinginequality}
    \frac{p(x)}{g(x)}\int_{\mathcal{X}_{\text{AI}}}g(y)dy + \int_{\mathcal{X}_{\text{IN}}}p(y)dy < \left(\frac{\underline{v}}{1-\gamma}\right)^{1/\alpha}\leq \frac{p(x')}{g(x')}\int_{\mathcal{X}_{\text{AI}}}g(y)dy + \int_{\mathcal{X}_{\text{IN}}}p(y)dy.
\end{equation}
The first inequality followed by expending $V(x;q) = (1-\gamma)(p(x)/q(x))^\alpha < \underline{v}$ using (\ref{generaldecomposedformsimplified}), while the second inequality followed using (\ref{generaldecomposedformsimplified}) and the necessary fact that $\beta(x') \in [0,1]$ at any equilibrium. Let $\underline{r}(\underline{v})$ be the solution to (\ref{definingbarrbarv}). Let $\overline{r} := \inf_{x\in \mathcal{X}_{\text{IN}}}r(x)$ and $\underline{r} := \sup_{x\in \mathcal{X}_{\text{AI}}}r(x)$, then $r(x) \leq \underline{r}\leq \overline{r}\leq r(x')$ for any $x\in \mathcal{X}_{\text{AI}}$ and $x'\in \mathcal{X}_{\text{IN}}$, hence it follows from (\ref{separatinginequality}) that:
\begin{equation*}
    \int_{\mathcal{X}}\max\left\{\frac{\underline{r}}{r(y)},1\right\}p(y)dy \leq \int_{\mathcal{X}}\max\left\{\frac{\underline{r}(\underline{v})}{r(y)},1\right\}p(y)dy \leq \int_{\mathcal{X}}\max\left\{\frac{\overline{r}}{r(y)},1\right\}p(y)dy.
\end{equation*}
Therefore, we have $\underline{r} \leq \underline{r}(\underline{v})\leq \overline{r}$, which means $\mathcal{X}_{\text{AI}} = r^{-1}[0,\underline{r}(\underline{v}))$ and $\mathcal{X}_{\text{IN}} = r^{-1}[\underline{r}(\underline{v}),\infty)$, as needed. The claim $V^{\underline{v}, w}(g;q) = \widetilde{V}^{\underline{v}}(g)$ follows from an explicit computation. Finally, we have from (\ref{separatinginequality}) that 
\begin{equation}\label{pureaibetacondition}
        1 < \left(\frac{\underline{v}}{1-\gamma}\right)^{1/\alpha} \leq \int_{\mathcal{X}}\max\left\{\frac{\overline{r}}{r(y)},1\right\}p(y)dy\\
        \leq \sup_{x\in \mathcal{X}}r(x).\numberthis
    \end{equation}
    where the first inequality is strict because we know from Lemma \ref{xhisemptylemma} that $\mathcal{X}_{\text{AI}} = r^{-1}[0,\underline{r}(\underline{v}))$ has a positive measure at any equilibrium.
\end{proof} 

\subsection{Proof of Lemma \ref{solutionv0lemma}}
\begin{proof}

    \emph{Part 1:} It is clear from the definition that, as functions of $\underline{v}$: $M_g(\underline{v})$, $M_p(\underline{v})$ are continuous everywhere at except at the point masses, while $\underline{r}(\underline{v})$ is continuous. Therefore, focus on the domain $\underline{v}\in (1-\gamma, (1-\gamma)\overline{\overline{r}}^\alpha)$ we have that $\widetilde{V}^{\underline{v}}(g)$ is continuous everywhere except at the point masses. To derive (\ref{discontinuityofvvbar}), note that we have from (\ref{definingbarrbarv}) that $\underline{r}(\underline{v})\nearrow \overline{\overline{r}}$ as $\underline{v}\nearrow (1-\gamma)\overline{\overline{r}}^\alpha$ and so we have from (\ref{mgmpvbarvdefinition}) that:
    \begin{multline*}
        \lim_{\underline{v}\nearrow (1-\gamma)\overline{\overline{r}}^\alpha}\widetilde{V}^{\underline{v}}(g) = \frac{(1-\gamma)\int_{r^{-1}[0,\overline{\overline{r}})}r(y)^\alpha g(y)dy + c\int_{r^{-1}[\overline{\overline{r}},\infty)}g(y)dy}{\int_{r^{-1}[0,\overline{\overline{r}})}g(y)dy}\\
        = \frac{(1-\gamma)\mathbb{E}_{x\sim g}r(x)^\alpha - (1-\gamma)\overline{\overline{r}}^\alpha\int_{r^{-1}[\overline{\overline{r}}, \infty)}g(y)dy + c\int_{r^{-1}[\overline{\overline{r}},\infty)}g(y)dy}{\int_{r^{-1}[0,\overline{\overline{r}})}g(y)dy}\\
        = (1-\gamma)\mathbb{E}_{x\sim g}r(x)^\alpha - \frac{\int_{r^{-1}[\overline{\overline{r}},\infty)}g(y)dy}{\int_{r^{-1}[0,\overline{\overline{r}})}g(y)dy}\left((1-\gamma)\overline{\overline{r}}^\alpha - (1-\gamma)\mathbb{E}_{x\sim g}r(x)^\alpha - c\right)
    \end{multline*}
    Rearranging this yields (\ref{discontinuityofvvbar}).
    
    \emph{Part 2:} Consider some $\underline{v}_1$ such that $\underline{v}_1 \leq \widetilde{V}^{\underline{v}_1}(g) + c$ and let us show that $\widetilde{V}^{\underline{v}}(g)+c$ is a monotonically decreasing function in $\underline{v} \in (1-\gamma, \underline{v}_1]$. To simplify the exposition, we focus on the case where $(1-\gamma, (1-\gamma)\overline{\overline{r}}^\alpha)$ is free of point masses, and we denote the density of $r(x)$ under $p$ by $h$. Let $t := \underline{r}(\underline{v})$, then we can write $\widetilde{V}^{\underline{v}}(g)$ from (\ref{mgmpvbarvdefinition}) as:
    \begin{equation*}
        \widetilde{V}^{\underline{v}}(g)+c = \frac{(1-\gamma)\left(\frac{1}{t}\int_0^\infty \max\left\{\frac{t}{r},1\right\}h(r)dr\right)^\alpha\int_0^t\frac{h(r)}{r^{1-\alpha}}dr + c}{\int_0^t\frac{h(r)}{r}dr} =: V(t).
    \end{equation*}
    Computing the derivative of $V(t)$ we find that $V'(t) \leq 0$ if:
    \begin{equation*}
        D(t) := (1-\gamma)\left(\frac{1}{t}\int_0^\infty \max\left\{\frac{t}{r},1\right\}h(r)dr\right)^\alpha\left(t^\alpha\int_0^t\frac{h(r)}{r}dr - \int_0^t\frac{h(r)}{r^{1-\alpha}}dr\right) - c \leq 0.
    \end{equation*}
    Let $t_1 := \underline{r}(\underline{v}_1)$, then $\underline{v}_1\leq \widetilde{V}^{\underline{v}_1}(g) + c$ implies that:
    \begin{equation*}
        \left(\int_0^\infty \max\left\{\frac{t_1}{r},1\right\}h(r)dr\right)^{\alpha} = \frac{\underline{v}_1}{1-\gamma} \leq \frac{V(t_1)}{1-\gamma}  \quad \implies \quad D(t_1) \leq 0.
    \end{equation*}
    The implication can be seen straightforwardly by expanding the definition of $V(t_1)$ and rearranging the inequality. Meanwhile, we can show that:
    \begin{equation*}
        D'(t) = \alpha (1-\gamma)\frac{t^{1+\alpha}\left(\int_0^t\frac{h(r)}{r}dr\right)^2 + \int_t^\infty h(r)dr\int_0^t\frac{h(r)}{r^{1-\alpha}}dr}{2t^2\left(\frac{1}{t}\int_0^\infty \max\left\{\frac{t}{r},1\right\}h(r)dr\right)^{1-\alpha}} \geq 0,
    \end{equation*}
    therefore, $D(t) \leq D(t_1) \leq 0$ for all $t \in (\inf_{x\in \mathcal{X}}r(x), t_1]$. It follows that $V'(t) \leq 0$ for all $t \in (\inf_{x\in \mathcal{X}}r(x), t_1]$ which is equivalent to $\widetilde{V}^{\underline{v}}(g) + c$ being monotonically decreasing for $\underline{v}\in (1-\gamma, \underline{v}_1]$. The proof can readily be translated to the general case by adding the sum over finite point masses to the density $h$. In particular, $\widetilde{V}^{\underline{v}}(g)$ decreases discontinuously at any point--masses.
    
    \emph{Part 3:} We have already argued in the main text that $\underline{v} \leq \widetilde{V}^{\underline{v}}(g) + c$ for all $\underline{v} > 1-\gamma$ sufficiently close to $1-\gamma$. Suppose that $(1-\gamma)\sup_{x\in \mathcal{X}}r(x)^\alpha > (1-\gamma)\mathbb{E}_{x\sim g}r(x)^\alpha + c$, then we know that (\ref{discontinuityofvvbar}) that:
    \begin{equation*}
        \lim_{\underline{v}\nearrow (1-\gamma)\overline{\overline{r}}^\alpha}(\widetilde{V}^{\underline{v}}(g) + c) < (1-\gamma)\mathbb{E}_{x\sim g}r(x)^\alpha+c < (1-\gamma)\overline{\overline{r}}^\alpha.
    \end{equation*}
    It follows that we have $\widetilde{V}^{\underline{v}}(g) + c < \underline{v}$ for all $\underline{v} < (1-\gamma)\overline{\overline{r}}^\alpha$ sufficiently close to $(1-\gamma)\overline{\overline{r}}^\alpha$. Since we have assume no point masses aside from $\overline{\overline{r}}$, it follows that $\widetilde{V}^{\underline{v}}(g)$ is continuous on $(1-\gamma, (1-\gamma)\overline{\overline{r}}^\alpha)$, and we conclude that there exists $\underline{v}_0 \in (1-\gamma, (1-\gamma)\overline{\overline{r}}^\alpha)$ such that $\underline{v}_0 = V^{\underline{v}_0}(g)+c$. On the other hand, if $(1-\gamma)\mathbb{E}_{x\sim g}r(x)^\alpha + c \geq (1-\gamma)\sup_{x\in \mathcal{X}}r(x)^\alpha$, then repeating the computation above with the inequality reversed, we find that $\underline{v} < \widetilde{V}^{\underline{v}}(g) + c$ for all $\underline{v} < (1-\gamma)\overline{\overline{r}}^\alpha$ sufficiently close to $(1-\gamma)\overline{\overline{r}}^\alpha$. Then it follows from the second part of this Lemma that $\underline{v} < \widetilde{V}^{\underline{v}}(g) + c$ for all $\underline{v} \in (1-\gamma, (1-\gamma)\overline{\overline{r}}^\alpha)$, as claimed.
\end{proof}

\subsection{Proof of Proposition \ref{pureailemma}}
\begin{proof}
    Suppose that the condition (\ref{pureaicondition}) holds and let us consider any equilibrium $(\beta, q)$ under the compensation scheme $(\underline{v}, w)$. Let us first assume that $\underline{v} \leq V^{\underline{v}, w}(g;q)+c \leq \underline{v}+ w$ and $w > 0$. From Lemma \ref{xhisemptylemma} we have the decomposition $\mathcal{X} = \mathcal{X}_{\text{AI}}\sqcup \mathcal{X}_{\text{IN}}$. If $\mathcal{X}_{\text{IN}} = \emptyset$ then $\beta=1,q=g$ and there is nothing to prove. Suppose that $\mathcal{X}_{\text{IN}}\neq \emptyset$, then it follows from Proposition \ref{generalequilibriumproposition} that $\underline{v} \leq (1-\gamma)\sup_{x\in \mathcal{X}}r(x)^\alpha$. If $w > 0$, then from (\ref{pureaicondition}) it must be the case that (see (\ref{separatinginequality})):
    \begin{equation}\label{pureaibetacondition2}
        \left(\frac{\underline{v}}{1-\gamma}\right)^{1/\alpha} = \frac{p(x)}{g(x)}\int_{\mathcal{X}_{\text{AI}}}g(y)dy + \int_{\mathcal{X}_{\text{IN}}}p(y)dy = \sup_{x\in \mathcal{X}}r(x)
    \end{equation}
    for all $x\in \mathcal{X}_{\text{IN}}$, and hence $\beta(x) = 1$ for all $x\in \mathcal{X}$, therefore we have $q = g$.
    
    Next, we consider the no compensation case, i.e. the scheme $(\underline{v}_0, w_0)$ with $\underline{v}_0 = V(g;q) + c$ and $w_0 = 0$. Extra care is needed since the compensation scheme $(\underline{v}, w=0)$ is actually independent of $\underline{v}$. We still have the decomposition $\mathcal{X} = \mathcal{X}_{\text{AI}}\sqcup \mathcal{X}_{\text{IN}}$, and suppose that $\mathcal{X}_{\text{IN}}\neq \emptyset$, then it follows from Proposition \ref{generalequilibriumproposition} that $\underline{v}_0 \leq (1-\gamma)\sup_{x\in \mathcal{X}}r(x)^\alpha$ and that $\underline{v}_0 = \widetilde{V}^{\underline{v}_0}(g)+c$. However, (\ref{pureaicondition}) does not implies that $\underline{v}_0 = V(g;q) + c \geq (1-\gamma)\sup_{x\in \mathcal{X}}r(x)^\alpha$ in this case. But we have from the $(1-\gamma)\mathbb{E}_{x\sim g}r(x)^\alpha+c \geq (1-\gamma)\sup_{x\in \mathcal{X}}r(x)^\alpha$ case of Lemma \ref{solutionv0lemma} that $\widetilde{V}^{\underline{v}}(g)$ is monotonically decreasing in $\underline{v} \in (1-\gamma, (1-\gamma)\overline{\overline{r}}^\alpha)$ and that $\lim_{\underline{v}\nearrow (1-\gamma)\overline{\overline{r}}^\alpha}\widetilde{V}^{\underline{v}}(g) + c\geq (1-\gamma)\sup_{x\in \mathcal{X}}r(x)^\alpha$, which means $\underline{v}_0 = \widetilde{V}^{\underline{v}_0}(g)+c\geq (1-\gamma)\sup_{x\in \mathcal{X}}r(x)^\alpha$. Therefore, we find that $\underline{v}_0$ also satisfies (\ref{pureaibetacondition2}), which implies $\beta(x) = 1$ for all $x\in \mathcal{X}$. Overall, we conclude that $q = g$ is the only possible equilibrium if $\underline{v} \leq V^{\underline{v}, w}(g;q)+c \leq \underline{v}+ w$ for any $w \geq 0$.

    For the case where $\underline{v} + w < V^{\underline{v}, w}(g;q) + c$, we know from Lemma \ref{decisionundercompensationlemma} that there exists a compensation scheme $(\tilde{\underline{v}}, \tilde{w})$ such that $\tilde{\underline{v}} + \tilde{w} = V^{\tilde{\underline{v}}, \tilde{w}}(g;q) + c \geq \tilde{\underline{v}}$, and that $\beta$ is also a best creation response under $(\tilde{\underline{v}}, \tilde{w})$ given the belief $q$. In other words, the given $(\beta, q)$ is also an equilibrium under $(\tilde{\underline{v}}, \tilde{w})$, but we see from our previous analysis, which covers the compensation scheme satisfying $\underline{v} + w \geq V^{\underline{v}, w}(g;q) + c \geq \underline{v}$, that this is not possible, hence a contradiction. The case where $\underline{v} > V^{\underline{v}, w}(g;q) + c$ is similar, since none of creators strictly prefers manual creation at any equilibrium by Lemma \ref{xhisemptylemma}, then by Lemma \ref{decisionundercompensationlemma}: $(\beta, q)$ is also an equilibrium under no compensation $(\underline{v}_0, w_0)$. Thus, we have once again reduced to the case covered in the previous analysis which led to a contradiction.

    We have shown that $q = g$ is the only possible equilibrium. We will now show that $q = g$ is indeed an equilibrium whenever the inequality in (\ref{pureaicondition}) is strict, or $\underline{v} + w \leq (1-\gamma)\mathbb{E}_{x\sim g}r(x)^\alpha + c$. We note from (\ref{pureaicondition}) for any $x\in \mathcal{X}$ that $V(x;q) = (1-\gamma)(p(x)/q(x))^\alpha = (1-\gamma)r(x)^\alpha\leq (1-\gamma)\sup_{x\in \mathcal{X}}r(x)^\alpha \leq \min\left\{V(g;q) + c, \underline{v}\right\}$. If $V(x;q) < \underline{v}$, then $V^{\underline{v}, w}(x;q) = V(x;q) \leq V(g;q) + c \leq V^{\underline{v}, w}(g;q) + c$, hence $x$ weakly prefers to use GenAI. If the inequality in (\ref{pureaicondition}) is not strict and $V(x;q) = \underline{v} \leq V(g;q) + c \leq V^{\underline{v}, w}(g;q) + c$, then $\underline{v}+w\leq (1-\gamma)\mathbb{E}_{x\sim g}r(x)^\alpha + c = V(g;q) + c \leq V^{\underline{v}, w}(g;q) + c$ implies that $x$ either strictly prefers GenAI or indifferent. It follows that $\beta(x) = 1$ for all $x\in \mathcal{X}$ is a valid response decision that sustains the equilibrium content density $q = g$. Of course, the platform obtains the same resulting equilibrium with: $(\underline{v}_0, w_0 = 0)$, therefore the platform weakly prefers $(\underline{v}_0, w_0)$ as it is guaranteed to pays no compensation.

    Conversely, if we know that $(\beta = 1, q = g)$ is an equilibrium under some compensation scheme $(\underline{v}, w)$, then the expected profit from GenAI usage is equal to, or greater than, the profit from manual creation, for all creators. But this is precisely the condition: $(1-\gamma)\mathbb{E}_{x\sim g}r(x)^\alpha \geq (1-\gamma)r(y)^\alpha - c$, for all $y \in \mathcal{X}$, which is equivalent to $(1-\gamma)\mathbb{E}_{x\sim g}r(x)^\alpha+c \geq (1-\gamma)\sup_{x\in \mathcal{X}}r(x)^\alpha$.
\end{proof}

\subsection{Proof of Proposition \ref{equilibriumnonexistenceproposition}}
\begin{proof}
    Suppose that $(\beta, q)$ is an equilibrium under the compensation scheme $(\underline{v}, w)$ satisfying the condition (\ref{admissibleconditionforvbar}) and $\underline{v} + w > \widetilde{V}^{\underline{v}}(g) + c$. From Lemma \ref{xhisemptylemma}, we have the decomposition $\mathcal{X} = \mathcal{X}_{\text{AI}}\sqcup \mathcal{X}_{\text{IN}}$. Let us assume for now that $\mathcal{X}_{\text{IN}} \neq\emptyset$. Suppose that
    $\underline{v} \leq V^{\underline{v}, w}(g;q) + c \leq \underline{v}+w$. Consider the case when the inequality is strict: $\underline{v} \leq V^{\underline{v}, w}(g;q) + c < \underline{v}+w$, otherwise, it follows from Proposition \ref{generalequilibriumproposition} that $V^{\underline{v}, w}(g;q) = \widetilde{V}^{\underline{v}}(g)$, but this contradicts $\underline{v} + w > \widetilde{V}^{\underline{v}}(g) + c$. Then for any creators $x\in \mathcal{X}$ if $V(x;q) < \underline{v}$, we have $V^{\underline{v}, w}(x;q) = V(x;q) < \underline{v} \leq V^{\underline{v}, w}(g;q) + c$ then $x$ strictly prefers GenAI, but if $V(x;q) \geq \underline{v}$ then $V^{\underline{v}, w}(x;q) = V(x;q) + w \geq \underline{v} + w > V^{\underline{v}, w}(g;q) + c$ then $x$ strictly prefers manual creation. This contradicts our assumption that $\mathcal{X}_{\text{IN}} \neq \emptyset$ and contradicts Lemma \ref{xhisemptylemma} that $\mathcal{X}_{\text{H}} = \emptyset$ at equilibrium.
    
    Suppose that $\underline{v} > V^{\underline{v}, w}(g;q) + c$ then $V^{\underline{v}, w}(g;q) = V(g;q)$ since none of the creator strictly prefers manual creation according to Lemma \ref{xhisemptylemma}, meaning that: $V(x;q) \leq V^{\underline{v}, w}(g;q) + c < \underline{v}$. By Lemma \ref{decisionundercompensationlemma}, $(\beta, q)$ is an equilibrium under no compensation $(\underline{v}_0, w_0)$ where $\underline{v}_0 := V(g;q) + c < \underline{v}, w_0 = 0$, and it follows from Proposition \ref{generalequilibriumproposition} that  $\underline{v}_0 = \widetilde{V}^{\underline{v}_0}(g) + c$. But since $\underline{v}_0 < \underline{v} \leq \widetilde{V}^{\underline{v}}(g)+c$, where the last inequality is due to the condition (\ref{admissibleconditionforvbar}), then $\widetilde{V}^{\underline{v}_0}(g) \geq \widetilde{V}^{\underline{v}}(g)$ by Lemma \ref{solutionv0lemma} that $\widetilde{V}^{\underline{v}}(g)$ is monotonically decreasing. But then $\underline{v} > \underline{v}_0 = \widetilde{V}^{\underline{v}_0}(g) + c \geq \widetilde{V}^{\underline{v}}(g) + c$, which contradicts the condition (\ref{admissibleconditionforvbar}).

    Suppose that $\underline{v}+w < V^{\underline{v}, w}(g;q)+c$, then according to Lemma \ref{decisionundercompensationlemma}, $(\beta, q)$ is also an equilibrium under a compensation scheme $(\tilde{\underline{v}}, \tilde{w})$ satisfying $\tilde{\underline{v}} \leq V^{\tilde{\underline{v}}, \tilde{w}}(g;q) + c = \tilde{\underline{v}} + \tilde{w}$ where $\tilde{\underline{v}} > \underline{v}$. Since $\mathcal{X}_{\text{IN}} \neq \emptyset$, the expression for $q$ is given by (\ref{generalequilibriumform}) according to Proposition \ref{generalequilibriumproposition}. Meanwhile, let us choose $w' := \widetilde{V}^{\underline{v}}(g) + c - \underline{v}$. It follows from condition (\ref{admissibleconditionforvbar}) and $\underline{v} + w > \widetilde{V}^{\underline{v}}(g) + c$ that $w' \in [0, w)$. Let $(\beta', q')$ be the equilibrium under $(\underline{v}, w')$ according to Proposition \ref{generalequilibriumproposition}, we have $\widetilde{V}^{\underline{v}}(g) = V^{\underline{v}, w'}(g;q')$. We present the expression for $q$ and $q'$ according to (\ref{generalequilibriumform}) for reference:
    \begin{equation*}
        \begin{aligned}
            q(x) &= \left(\frac{1-\gamma}{\tilde{\underline{v}}}\right)^{1/\alpha}\left(p(x)\cdot \mathbbm{1}[r(x) \geq \underline{r}(\tilde{\underline{v}})] + \underline{r}(\tilde{\underline{v}}) g(x)\cdot \mathbbm{1}[r(x) < \underline{r}(\tilde{\underline{v}})]\right)\\
            q'(x) &= \left(\frac{1-\gamma}{\underline{v}}\right)^{1/\alpha}\left(p(x)\cdot \mathbbm{1}[r(x) \geq \underline{r}(\underline{v})] + \underline{r}(\underline{v}) g(x)\cdot \mathbbm{1}[r(x) < \underline{r}(\underline{v})]\right)
        \end{aligned}.
    \end{equation*}
    We will compare the expected revenue $V^{\underline{v}, w'}(g;q')$ from using GenAI on the platform when the content distribution is at equilibrium $q'$ to $V^{\underline{v}, w'}(g;q)$ when the content distribution is off--equilibrium $q$, under the same compensation scheme $(\underline{v}, w')$.
    On one hand, keeping the equilibrium $(\beta, q)$ fixed, then $V^{\underline{v}, w}(g;q)$ is linear in $w$ with slope less than $1$, therefore lowering $w$ to $w' < w$ we get: $\underline{v} + w' < V^{\underline{v}, w'}(g;q)+c$, hence $V^{\underline{v}, w'}(g;q') < V^{\underline{v}, w'}(g;q)$. On the other hand, $V(x;q) = \tilde{\underline{v}}$ and $V^{\underline{v}, w}(x;q) = V(x;q) + w = V^{\underline{v}, w}(g;q)+c$ for all $x \in r^{-1}[\underline{r}(\tilde{\underline{v}}), \infty)$, since $r^{-1}[\underline{r}(\tilde{\underline{v}}), \infty)$ is the set of indifferent creators under $(\tilde{\underline{v}}, \tilde{w})$ which is the same as under $(\underline{v}, w)$ by construction from Lemma \ref{decisionundercompensationlemma}.  But then $V^{\underline{v}, w'}(x;q) = V(x;q) + w' \leq V^{\underline{v}, w'}(g;q)+c$ for all $x \in r^{-1}[\underline{r}(\tilde{\underline{v}}), \infty)$, since $w' < w$ and $V^{\underline{v}, w'}(g;q)$ is linear in $w'$ with slope less than $1$ for a fixed $(\beta, q)$. Similarly, $r^{-1}[0,\underline{r}(\tilde{\underline{v}}))$ is a set of creators who strictly prefer GenAI under $(\underline{v}, w)$. More precisely, we have that $V^{\underline{v}, w'}(x;q) = V(x;q) = \frac{\tilde{\underline{v}}}{\underline{r}(\tilde{\underline{v}})^\alpha}r(x)^\alpha\leq \frac{\underline{v}}{\underline{r}(\underline{v})^\alpha}r(x)^\alpha = V(x;q') < \underline{v}$ for all $x \in r^{-1}[0, \underline{r}(\underline{v}))$, and that $V^{\underline{v}, w}(x;q) \leq V(x;q) + w < V^{\underline{v}, w}(g;q) + c$ which implies $V(x;q) + w' < V^{\underline{v}, w'}(g;q) + c$, therefore $V^{\underline{v}, w'}(x;q) < V^{\underline{v}, w'}(g;q) + c$ for all $x \in r^{-1}[\underline{r}(\underline{v}), \underline{r}(\tilde{\underline{v}}))$. Note that $\underline{v}/\underline{r}(\underline{v})^\alpha$ is decreasing in $\underline{v}$. It follows that:
    \begin{multline*}
        V^{\underline{v}, w'}(g;q) = \int_{\mathcal{X}}V^{\underline{v}, w'}(y;q)g(y)dy\\
        \leq \int_{r^{-1}[0, \underline{r}(\underline{v}))}V(y;q')g(y)dy + (V^{\underline{v}, w'}(g;q)+c)\int_{r^{-1}[\underline{r}(v), \infty)}g(y)dy\\
        \implies V^{\underline{v}, w'}(g;q)+c \leq \frac{\int_{r^{-1}[0, \underline{r}(\underline{v}))}V(y;q')g(y)dy + c}{\int_{r^{-1}[0,\underline{r}(\underline{v}))}g(y)dy} = V^{\underline{v}, w'}(g;q')+c
    \end{multline*}
    which is a contradiction to the previous strict inequality.

    Let us now turn our attention to the case where $\mathcal{X}_{\text{IN}} = \emptyset$, we have $\beta = 1, q = g$. By condition (\ref{admissibleconditionforvbar}), $\underline{v} \in (1-\gamma, (1-\gamma)\sup_{x\in \mathcal{X}}r(x)^\alpha)$. Then for all creators to strictly prefer GenAI, it must be the case that: $\underline{v} + w \leq (1-\gamma)\sup_{y\in \mathcal{X}}r(y)^\alpha+w < V^{\underline{v}, w}(g;q)+c$, otherwise, we can find $x$ such that $V^{\underline{v}, w}(x;q) = (1-\gamma)r(x)^\alpha+w \geq V^{\underline{v}, w}(g;q) + c$. We can follow a similar argument as in the previous case to force a contradiction, given that we already know $q = g$ so Proposition \ref{generalequilibriumproposition} which requires $\mathcal{X}_{\text{IN}}\neq \emptyset$ can be by--passed.
\end{proof}

\subsection{Proof of Proposition \ref{platformrevennueproposition}}
\begin{proof}
    First, we consider $\underline{v}$ satisfying the condition (\ref{admissibleconditionforvbar}) of Proposition \ref{generalequilibriumproposition}. Let us suppose that $r(x)$ is distributed under $p$ by a density $h$ and finite number of point masses $\sum_{i=1}^nm_i\delta_{r_i}$. Let $t := \underline{r}(\underline{v})$, then we can write $R(\underline{v})$ in terms of $t$ and the distribution of $r(x)$ as:
    \begin{equation*}
        R(t) = \gamma\frac{\int_0^\infty\max\{(t/r)^\alpha,1\}h(r)dr + \sum_{i=1}^nm_i\max\{(t/r_i)^\alpha,1\}}{\left(\int_0^\infty\max\{t/r,1\}h(r)dr + \sum_{i=1}^nm_i\max\{t/r_i,1\}\right)^{\alpha}}.
    \end{equation*}
    It is clear that $R(t)$ is continuous in $t$ and it is differentiable when $t \neq r_i$ for any $i=1,\cdots,n$. For any $t \neq r_i$, $i=1,\cdots,n$, we can compute that:
    \begin{equation*}
        R'(t) = \alpha\gamma\frac{\left(\int_t^\infty h(r)dr + \sum_{i=1;r_i\geq t}^nm_i\right)\left[\int_0^t(\frac{1}{t^{1-\alpha}r^\alpha}-\frac{1}{r})h(r)dr+\sum_{i=1;r_i<t}^nm_i(\frac{1}{t^{1-\alpha}r_i^\alpha}-\frac{1}{r_i})\right]}{\left(\int_0^\infty\max\{t/r,1\}h(r)dr + \sum_{i=1}^nm_i\max\{t/r_i,1\}\right)^{\alpha+1}}.
    \end{equation*}
    The sign of $R'(t)$ is determined by the the sign of the square--bracket factor, which is negative, since $1/(t^{1-\alpha}r^\alpha) \leq 1/r$ when $r\leq t$. Therefore, we conclude that $R(t)$ is a monotonically decreasing continuous function in $t$. But since $t=\underline{r}(\underline{v})$ is a monotonically increasing continuous function of $\underline{v}$, the claim about $R(\underline{v})$ follows. To obtain the platform's profit (\ref{platformcompensationschemeprofitvbar}), we start from (\ref{platformrevenue}) which simplifies in the case of the compensation scheme $(\underline{v}, w)$ to:
    \begin{equation*}
        \Pi(\underline{v}) = R(\underline{v}) - (\widetilde{V}^{\underline{v}}(g) + c - \underline{v})\int_{r^{-1}[\underline{r}(\underline{v}), \infty)}q(x)dx.
    \end{equation*}
    Substituting the expression for $q(x)$ from (\ref{generalequilibriumform}) yields (\ref{platformcompensationschemeprofitvbar}). 

    On the other hand, in the case where $\underline{v}$ satisfies the condition (\ref{pureaicondition}) of Proposition \ref{pureailemma}, we have $q = g$. Since $V(x;q) < \underline{v}$ for all $x\in \mathcal{X}$ in this case, none of the content receive a compensation, hence: $\Pi(\underline{v}) = R(\underline{v})$ and the rest of (\ref{pureaiplatformcompensationschemeprofitvbar}) follows.
    
    Lastly, we consider the special case where $\overline{\overline{r}}$ is the only possible point mass. It follows from Lemma \ref{solutionv0lemma} that $\widetilde{V}^{\underline{v}}(g)$, and hence $\Pi(\underline{v})$, is continuous over the domain where $\underline{v}$ satisfies the condition (\ref{admissibleconditionforvbar}) of Proposition \ref{generalequilibriumproposition}. In this case, the platform's profit $\Pi(\underline{v})$ is given by (\ref{platformcompensationschemeprofitvbar}). Since there is no other atom apart from a possible one at $\overline{\overline{r}}$, we have that $M_p(\underline{v})\nearrow 1$ as $\underline{v}\searrow 1-\gamma$ and $\Pi(\underline{v})\rightarrow -\infty$ as $1/(1-M_p(\underline{v}))\rightarrow \infty$. If $(1-\gamma)\sup_{x\in \mathcal{X}}r(x)^\alpha > (1-\gamma)\mathbb{E}_{x\sim g}r(x)^\alpha + c$ then Lemma \ref{solutionv0lemma} further tells us that the solution $\underline{v}_0 \in (1-\gamma, (1-\gamma)\overline{\overline{r}}^\alpha)$ to $\underline{v} = \widetilde{V}^{\underline{v}}(g) + c$ exists, and that $\underline{v} < \widetilde{V}^{\underline{v}}(g)+c$ for all $\underline{v} < \underline{v}_0$. In this case, the domain where $\underline{v}$ satisfies the condition (\ref{admissibleconditionforvbar}) simplifies to $(1-\gamma, \underline{v}_0]$. If $(1-\gamma)\sup_{x\in \mathcal{X}}r(x)^\alpha \leq (1-\gamma)\mathbb{E}_{x\sim g}r(x)^\alpha + c$, then $\lim_{\underline{v}\nearrow (1-\gamma)\overline{\overline{r}}^\alpha}\Pi(\underline{v})\leq \lim_{\underline{v}\searrow(1-\gamma)\overline{\overline{r}}^\alpha}\gamma M(\underline{v}) = \gamma \mathbb{E}_{x\sim g}r(x)^\alpha$, where the discontinuity is due to the compensation $\lim_{\underline{v}\nearrow (1-\gamma)\overline{\overline{r}}^\alpha}(\widetilde{V}^{\underline{v}}(g)+c-\underline{v}) \geq 0$ paid to the point mass of creators at $\underline{v} = (1-\gamma)\overline{\overline{r}}^\alpha$ and note that $\gamma \mathbb{E}_{x\sim g}r(x)^\alpha$ is the platform's profit under no compensation: $(\underline{v}, w=0)$. Therefore, in the context of the platform's profit maximization problem, we can restrict ourselves to a compact subset of $\underline{v}$ inside $(1-\gamma, (1-\gamma)\overline{\overline{r}}^\alpha)$ where $\Pi(\underline{v})$ is continuous, and the possible maximum at $\underline{v} = (1-\gamma)\mathbb{E}_{x\sim g}r(x)^\alpha+c$ corresponding to no compensation. Overall, we have shown that the platform's profit maximizer $\underline{v}^*$ exists.
\end{proof}

\end{document}